\documentclass[11pt]{article}
\pdfoutput=1
\usepackage[T1]{fontenc}
\usepackage{lmodern}
\usepackage{amsthm}
\usepackage{graphicx} 
\usepackage{array} 

\usepackage{amsmath, amssymb, amsfonts, verbatim, bbm}
\usepackage{hyphenat,subfigure,multirow}
\usepackage{algorithm,algpseudocode}
\usepackage{algorithmicx}
\usepackage{algpseudocode}
\usepackage{mathtools}
\usepackage{stmaryrd}

\usepackage[usenames,dvipsnames]{xcolor}

\definecolor{DarkRed}{rgb}{0.5,0.1,0.1}
\definecolor{DarkBlue}{rgb}{0.1,0.1,0.5}

\usepackage{hyperref}
\hypersetup{
colorlinks=true,
pdfnewwindow=true,
citecolor=ForestGreen,
linkcolor=DarkRed,
filecolor=DarkRed,
urlcolor=DarkBlue
}

\usepackage{bm}
\usepackage{url}
\usepackage{xspace}
\usepackage[mathscr]{euscript}

\usepackage{mdframed}

\usepackage{cite}
\usepackage{enumitem}

\usepackage[margin=1in]{geometry}

\newtheorem{theorem}{Theorem}
\newtheorem{lemma}{Lemma}
\newtheorem{prop}{Proposition}
\newtheorem{cor}{Corollary}

\theoremstyle{definition}
\newtheorem{rem}{Remark}
\newtheorem{definition}{Definition}

\newtheorem*{claim*}{Claim}
\newtheorem*{proposition*}{Proposition}
\newtheorem*{lemma*}{Lemma}
\newtheorem*{problem*}{Problem}

\newtheorem{mdresult}[theorem]{Theorem}

\newtheorem{mdinvariant}{Invariant}

\newcommand{\Blasiok}{B\l{}asiok}

\newcommand{\rb}[2]{\raisebox{#1 mm}[0mm][0mm]{#2}}
\newcommand{\istrut}[2][0]{\rule[- #1 mm]{0mm}{#1 mm}\rule{0mm}{#2 mm}}

\newcommand{\ignore}[1]{}

\newcommand{\Var}{\operatorname{Var}}

\newcommand{\E}{\mathbb{E}}
\newcommand{\ceil}[1]{\left\lceil{#1}\right\rceil}
\newcommand{\floor}[1]{\left\lfloor{#1} \right\rfloor}
\newcommand{\bydef}{\stackrel{\mathrm{def}}{=}}
\newcommand{\poly}{\operatorname{poly}}
\newcommand{\Geometric}{\text{Geometric}}
\newcommand{\Poisson}{\text{Poisson}}
\newcommand{\Normal}{\mathcal{N}}
\newcommand{\paren}[1]{\mathopen{}\left(#1\right)\mathclose{}}
\newcommand{\fish}{\textsf{FiSh}}
\newcommand{\fishmonger}{\textsf{Fishmonger}}
\newcommand{\Recordinality}{\textsf{Recordinality}}
\newcommand{\MVP}{\textsf{MVP}}
\newcommand{\Martingale}{\textsf{Martingale}}
\newcommand{\Dartboard}{\textsf{Dartboard}}
\newcommand{\DartboardModel}{\textsf{Dartboard Model}}

\newcommand{\Indicator}[1]{\left\llbracket\istrut[1]{3.5} #1 \right\rrbracket}

\newcommand{\LL}{\textsf{LL}}
\newcommand{\Hyper}{\textsf{Hyper}}
\newcommand{\HyperLogLog}{\textsf{HyperLogLog}}
\newcommand{\HyperBitBit}{\textsf{HyperBitBit}}
\newcommand{\HyperTwoBits}{\textsf{HyperTwoBits}}
\newcommand{\LogLog}{\textsf{LogLog}}

\newcommand{\qLL}{q\text{-}\LL}
\newcommand{\PCSA}{\textsf{PCSA}}
\newcommand{\qPCSA}{q\text{-}\PCSA}
\newcommand{\MinCount}{\textsf{MinCount}}
\newcommand{\Min}[1]{(#1)\text{-}\textsf{Min}}
\newcommand{\MIN}{\textsf{Min}}
\newcommand{\Bottom}{\textsf{Bottom}}
\newcommand{\AdaptiveSampling}{\textsf{AdaptiveSampling}}
\newcommand{\SBitmap}{\textsf{S-Bitmap}}
\newcommand{\GRA}{\textsf{GRA}}
\newcommand{\CPC}{\textsf{CPC}}

\allowdisplaybreaks

\DeclareMathOperator*{\argmax}{arg\,max}

\title{Information Theoretic Limits of Cardinality Estimation:\\
Fisher Meets Shannon\thanks{This work was supported by NSF grants CCF-1637546, CCF-1815316, CCF-2221980, CCF-2446604.  An extended abstract of this work was presented at the 53rd Annual ACM Symposium on Theory of Computing (STOC).}}

\author{Seth Pettie\\ University of Michigan \\ \footnotesize
\texttt{pettie@umich.edu} \and
Dingyu Wang\\ University of Michigan\\
\footnotesize \texttt{wangdy@umich.edu}
}

\date{}

\begin{document}
\maketitle
\begin{abstract}

Estimating the \emph{cardinality} (number of distinct elements)
of a large multiset is a classic problem in streaming and sketching, 
dating back to Flajolet and Martin's 
classic \textsf{Probabilistic Counting} ($\PCSA$) algorithm from 1983.

\medskip

In this paper we study the intrinsic tradeoff between the 
\emph{space complexity}
of the sketch and its \emph{estimation error} 
in the \textsc{random oracle} model.  
We define a new measure of
efficiency for data sketches 
called the \emph{Fisher}-\emph{Shannon} 
($\fish$) number $\mathcal{H}/\mathcal{I}$.
It captures the tension between the limiting Shannon entropy ($\mathcal{H}$)
of the sketch and its normalized Fisher information ($\mathcal{I}$) that characterizes the variance of a statistically efficient, asymptotically unbiased estimator.  
Our aim in introducing the $\fish$-number is to build the mathematical machinery necessary to argue for \emph{precise} optimality, rather than \emph{asymptotic} 
optimality, 
up to large constant factors.

\medskip

Our results are as follows.

\begin{itemize}
\item 
We prove that all base-$q$ variants of Flajolet and Martin's $\PCSA$
sketch have $\fish$-number $H_0/I_0 \approx 1.98016$ 
and that every base-$q$ variant of (\textsf{Hyper})\textsf{LogLog}
has $\fish$-number worse than $H_0/I_0$, but that they tend to
$H_0/I_0$ in the limit as $q\rightarrow \infty$.
Here $H_0,I_0$ are precisely defined constants.
This result \emph{reverses} the common 
conception that \HyperLogLog{} was a strict improvement over \PCSA.

\medskip
\item We describe a sketch called \fishmonger{} 
that is based on a \emph{smoothed}, entropy-compressed 
variant of $\PCSA$ with a different estimator function.  
It is proved that with high probability, 
\fishmonger{} processes a multiset of $[U]$ 
such that \emph{at all times}, 
its space is 
$(1+o(1))(H_0/I_0)m \approx 1.98m$ bits
and its standard error is $1/\sqrt{m}$, when $m\gg (\log\log U)^2$.
For example, to achieve a 1\% standard error, 
one needs a little more than 19,800 bits,
or $\approx 2.42$ kilobytes.

\medskip
\item Finally, we give circumstantial evidence that $H_0/I_0$ is the optimum $\fish$-number of mergeable sketches for Cardinality Estimation.  
We define a natural subset of mergeable sketches called \emph{linearizable} sketches and 
prove that no member of this class can beat $H_0/I_0$.  
The popular mergeable sketches are, in fact, 
also linearizable.  
This result supports the hypothesis that the space complexity of \fishmonger{} is not merely asymptotically optimal, 
but optimal up to an additive $o(m)$ bits.
\end{itemize}
\end{abstract}

\thispagestyle{empty}
\setcounter{page}{0}

\newpage

\setcounter{section}{-1}

\section{Preamble and Apologia}\label{sect:preable}

One way to read this paper is as a work of 20th century speculative mathematical fiction, which begins with the following natural premise. What if the research field of \emph{streaming and data sketching} 
were \emph{not} developed in the 1980s and further in the 90s 
by Flajolet \& Martin~\cite{FlajoletM85},
Alon, Matias, and Szegedy~\cite{AlonMS99}, 
Cohen~\cite{Cohen97}, and Gibbons et al.~\cite{GibbonsM99,GibbonsMP02,AlonGMS02,AcharyaGPR99,AcharyaGPR99a},
but instead were developed in the 1950s \emph{using the 
mathematical tools available at the time}.  
Even if the great Claude Shannon and Ronald Fisher managed
to collaborate on such an anachronistic project, 
\emph{surely} they would produce 
a wholly inadequate theory, being deprived of all the modern notions
of computer science such as computational complexity theory, 
communication complexity, hash functions and $k$-wise independence, 
data structures, liberal use of asymptotic notation, and so on.

The main take-away message of this work is that a fortuitous and productive Fisher--Shannon collaboration would
not just have been \emph{different} than the modern trajectory of streaming and 
sketching, but would have been more useful and enlightening, in some measurable ways.

\section{Introduction}\label{sect:introduction}


\emph{Cardinality Estimation} (aka \emph{Distinct Elements} or \emph{$F_0$-estimation})
is a fundamental problem in streaming/sketching, with widespread industrial deployments
in databases, networking, and sensing.
Sketches for Cardinality Estimation are evaluated along three axes:
\emph{\textbf{space complexity}} (in bits), \emph{\textbf{estimation error}}, 
and \emph{\textbf{algorithmic complexity}} of the update function.
In the end we want a perfect understanding of the \emph{three-way} tradeoff 
between these measures, but that is only possible if we truly understand
the \emph{two-way} tradeoff between space complexity and estimation error,
which is information-theoretic in nature.  In this paper we investigate
this two-way tradeoff in the \textsc{random oracle} model.  

Prior work in Cardinality Estimation has assumed either the 
\textsc{random oracle} model,
in which we have query access to a uniformly random hash function,
or what we call the
\textsc{standard model},
in which unbiased random bits can be generated, 
but all hash functions are stored explicitly.
Sketches in the \textsc{random oracle} model
typically pay close attention to constant factors in both space and 
estimation error~\cite{FlajoletM85,Flajolet90,DurandF03,Giroire09,ChassaingG06,FlajoletFGM07,Lumbroso10,EstanVF06,BeyerGHRS09,ChenCSN11,CliffordC12,Cohen15,Ting14,HelmiLMV12,Sedgewick,PettieWY21,LukasiewiczU22}.  
Sketches in the \textsc{standard model}~\cite{AlonMS99,GibbonsT01,Bar-YossefKS02,Bar-YossefJKST02,KaneNW10,IndykW03,Blasiok20}
use \emph{explicit} (e.g., $O(1)$-wise independent) hash functions and generally
pay less attention to the leading constants in space and estimation error.  
Sketches in the \textsc{random oracle} model have had a bigger impact
on the \emph{practice} of Cardinality 
Estimation~\cite{HeuleNH13,DataSketches,Sedgewick}.
Their empirical performance with imperfect hash functions
agrees with
theoretical predictions under the \textsc{random oracle model}~\cite{FlajoletM85,FlajoletFGM07,HeuleNH13,DataSketches}.

\medskip

\paragraph{Random Oracle Model.}
It is assumed that we have oracle access to a uniformly random function 
$h\colon [U]\rightarrow \{0,1\}^\infty$, where $[U]$ is the universe of our multisets
and the range is interpreted as a point in $[0,1]$.
(To put prior work on similar footing we assume in Table~\ref{table:history} 
that any stored hash values 
are stored to $\log U$ bits of precision.)
For practical purposes, elements in $[U]$ and $[0,1]$ can be regarded as 64-bit integers/floats.

\paragraph{Problem Definition.} A sequence $\mathcal{A} = (a_1,\ldots,a_N)\in [U]^N$ over 
some universe $[U]$ is revealed one element at a time. 
We maintain a $b$-bit sketch 
$S\in \{0,1\}^b$ such that if 
$S_i$ is its state after 
seeing $(a_1,\ldots a_i)$, 
$S_{i+1}=f(S_i,h(a_{i+1}))$ 
for some 
\emph{transition function} $f$.
The goal is to be able to estimate 
the \emph{cardinality} 
$\lambda = |\{a_1,\ldots,a_N\}|$ of the set.
Define $\hat{\lambda}(S) \colon \{0,1\}^b \rightarrow \mathbb{R}$ to be 
the estimation function.
An estimator is \emph{$(\epsilon,\delta)$-approximate} if 
$\Pr\left(\hat{\lambda} \not\in [(1-\epsilon)\lambda,(1+\epsilon)\lambda]\right) < \delta$.
Most results in the \textsc{random oracle} model use estimators that are almost unbiased or 
asymptotically unbiased,
as $b\rightarrow \infty$.
Given that this holds it is natural to measure the 
distribution of $\hat{\lambda}$ relative to $\lambda$.
We pay particular attention to the \emph{relative} 
variance $\frac{1}{\lambda^2}\Var(\hat{\lambda} \mid \lambda)$
and the \emph{relative} standard deviation 
$\frac{1}{\lambda}\sqrt{\Var(\hat{\lambda} \mid \lambda)}$, also
called the \emph{standard error}.

\begin{rem}
Table~\ref{table:history} summarizes prior work.  To compare \textsc{random oracle}
and \textsc{standard model} algorithms, note that 
an asymptotically unbiased $\tilde{O}(m)$-bit sketch with standard error $O(1/\sqrt{m})$ is 
morally similar to an $\tilde{O}(\epsilon^{-2})$-bit sketch 
with $(\epsilon,\delta)$-approximation guarantee, $\delta=O(1)$.
However, the two guarantees are formally incomparable.  
The $(\epsilon,\delta)$-guarantee
does not specifically claim anything about bias or variance, 
and with probability $\delta$ the error is technically not bounded.
\end{rem}

\subsection{Mergeability}

It is desirable in many applications that data sketches be \emph{mergeable} (aka \emph{composable}), 
meaning that the sketch $S(A\cup B)$ of $A\cup B$ is a function of $S(A)$ and $S(B)$.  
In the cardinality estimation problem, 
one may conclude that a sketch is mergeable by checking if the transition function $f$ is \emph{commutative} and \emph{idempotent}, 
i.e.,
\begin{align*}
\textbf{Idempotency.} \qquad f(f(S,h(a)),h(a)) &= f(S,h(a)).\\
\textbf{Commutativity.} \qquad f(f(S,h(a)),h(a')) &= f(f(S,h(a')),h(a)).
\end{align*}
If $f$ satisfies these properties then 
$S_i$ is a function of $\{h(a_1),\ldots,h(a_i)\}$,
that is, it is not influenced by the order or multiplicity of the input elements.  
As a consequence, we can characterize the 
state $S_i$ as a partition of the set of 
hash values into $V_0(S_i) \cup V_1(S_i)$, 
where $V_0$ is the set of hash values that
would cause no state change, and $V_1$ is the 
set of hash values that would cause a state change.
(At the very least, $V_0 \supseteq \{h(a_1),\ldots,h(a_i)\}$.)
This observation lead Ting~\cite{Ting14} to 
visualize cardinality sketches in a unified manner
as what he called 
\emph{area-cutting processes}, and what we 
formalize as the 
\emph{\DartboardModel}.

\subsection{The \DartboardModel}

The \emph{\Dartboard} is the unit square $[0,1]^2$, 
which is partitioned into a set $\mathcal{C}$ of
\emph{cells}.  When $a_i$ is inserted a \emph{dart}
is thrown, and lands at position $h(a_i) \in [0,1]^2$, \emph{hitting} the corresponding cell.
The \emph{state} of a \Dartboard{} sketch consists
of a partition of its cells into 
\emph{occupied} and \emph{free}, obeying the following axioms:

\begin{description}
    \item[Axiom 1.] Every cell hit by a dart is occupied, but 
    occupied cells may contain no darts.
    \item[Axiom 2.] If a dart hits an occupied cell, the state does not change.
\end{description}

Axiom 1 implies that if a dart hits a free cell, the state \emph{must} change, and Axioms 1 and 2 imply another property, 
which we state explicitly as an ``axiom.''
\begin{description}
\item[Axiom 3.] A cell, once occupied, may not subsequently become free.
\end{description}
The \emph{state space} of the sketch 
is a subset of $2^{\mathcal{C}}$,
reflecting the subsets of $\mathcal{C}$ that can be simultaneously occupied.
Let us illustrate how several canonical sketches can be expressed in a uniform way as \Dartboard{} sketches.

Flajolet and Martin's \PCSA{} sketch\footnote{(Probabilistic Counting with Stochastic Averaging)}~\cite{FlajoletM85} 
maps each element $a\in [U]$ to 
$h(a)=(i,j)\in [m]\times \mathbb{Z}^+$
with probability $2^{-j}/m$; the state of the sketch consists of $\{h(a_1),\ldots,h(a_N)\}$, 
i.e., all distinct hash values encountered.
This corresponds to partitioning the \Dartboard{} into $m$ equal width columns, and each column into 
geometrically decaying cells of height 
$1/2, 1/4, 1/8,$ and so on.  The set of occupied cells (the state of \PCSA) is precisely the set of cells hit by a dart.
Figure~\ref{fig:dartboard-PCSA-HLL}(a) illustrates the \PCSA-partition of the \Dartboard{} into $m=16$ columns and a 
plausible arrangement
of $\lambda=32$ darts, while Figure~\ref{fig:dartboard-PCSA-HLL}(b) is the state of the \PCSA{} sketch, with occupied cells in blue.
For practical purposes we can truncate the \PCSA{} partition at $\log U$ cell-rows, so it requires $m\log U$ bits to \emph{explicitly} store the state of the \PCSA{} sketch.

\begin{figure}
\centering
    \begin{tabular}{c@{\hspace{.9cm}}c@{\hspace{.9cm}}c}
        \includegraphics[scale=.15]{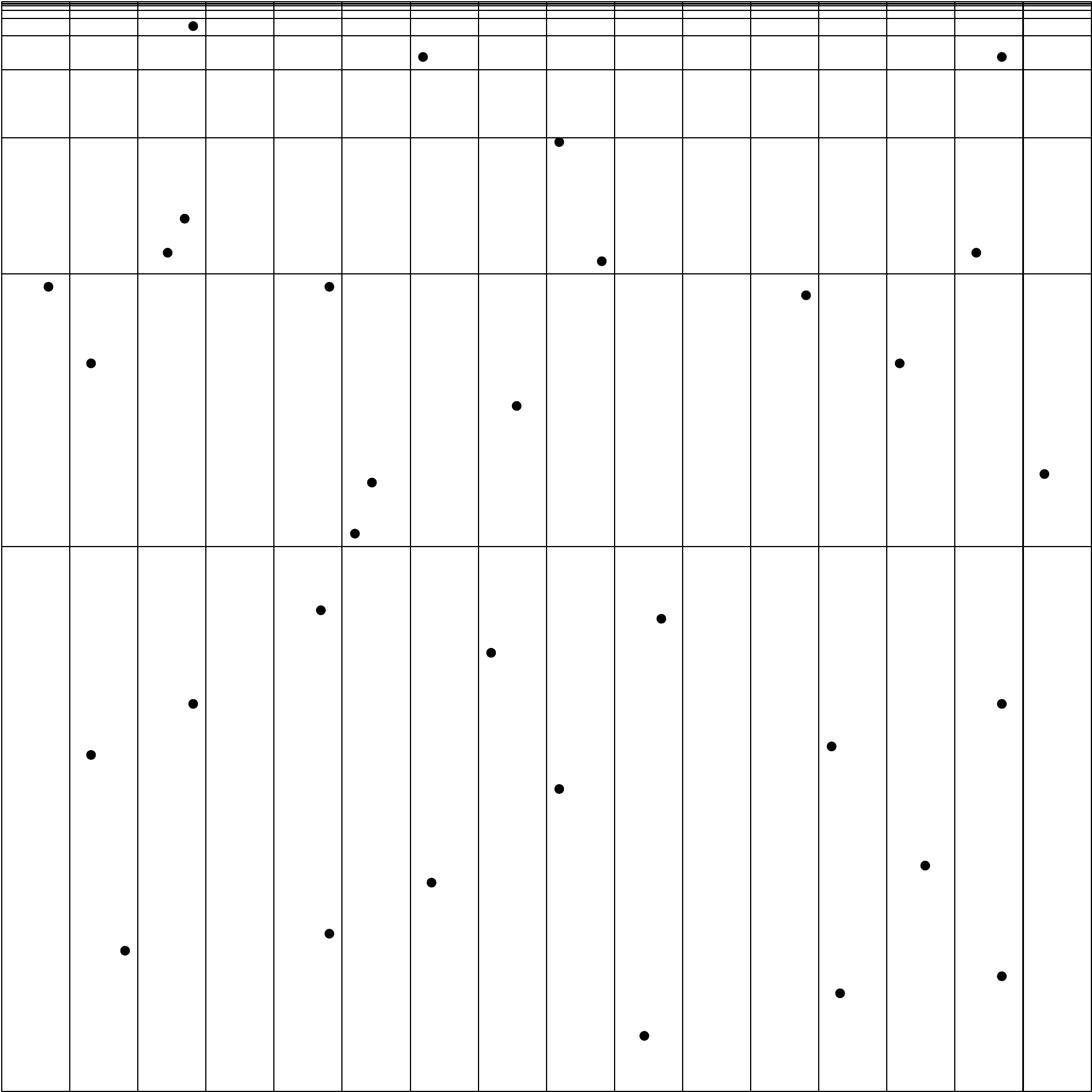}
    &   \includegraphics[scale=.15]{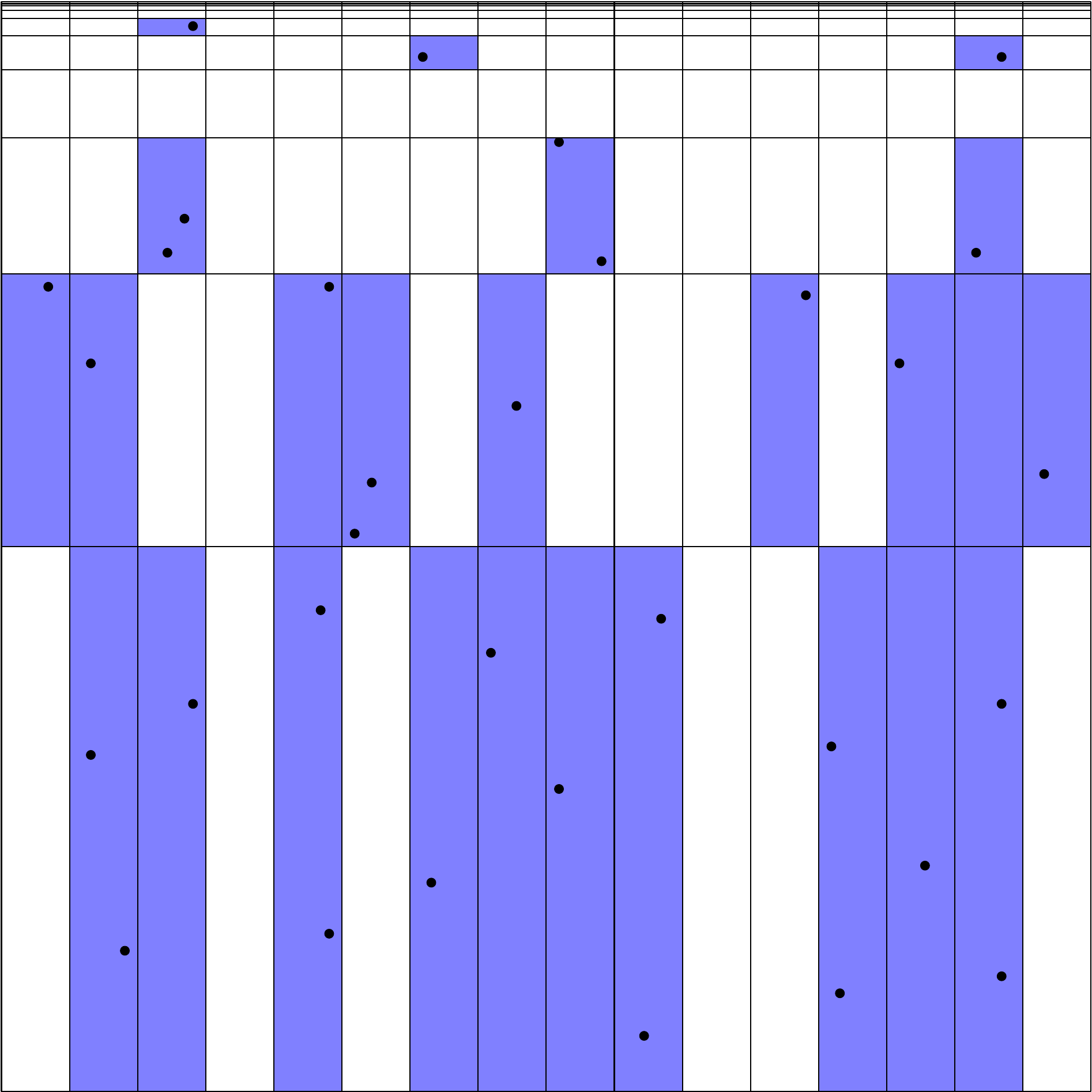}
    &   \includegraphics[scale=.15]{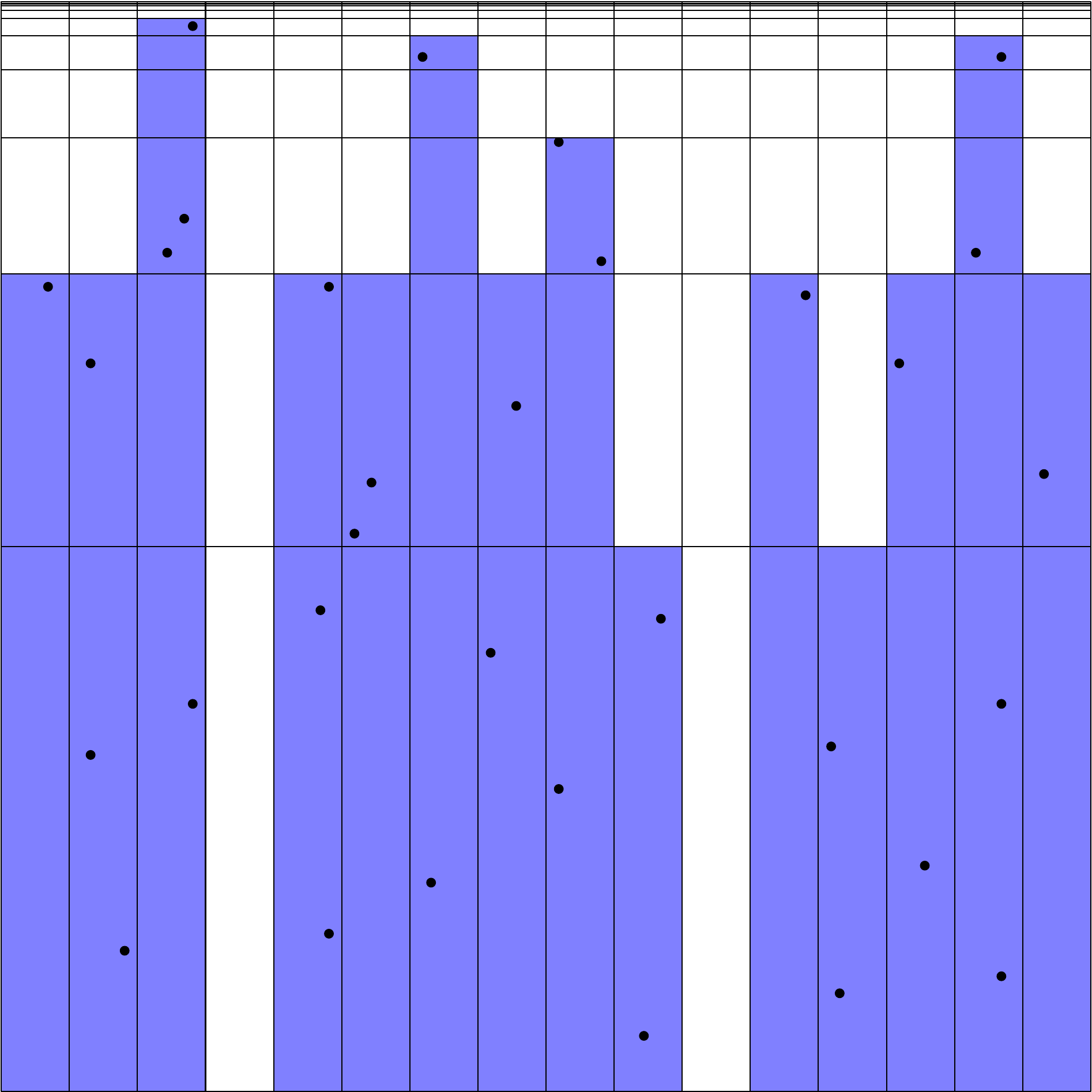}\\
    \textbf{(a)} &  \textbf{(b)} &  \textbf{(c)} 
    \end{tabular}
    \caption{\label{fig:dartboard-PCSA-HLL}\textbf{(a)} The \PCSA-partition of the \Dartboard{} into $m=16$ columns, with 32 darts.  \textbf{(b)} The state of the \PCSA{} sketch; occupied cells are blue.  \textbf{(c)} The state of the (\textsf{Hyper})\LogLog{} sketch, which uses the same partition.  Every cell hit by a dart or below one hit by a dart is occupied.}
\end{figure}

The popular \LogLog~\cite{DurandF03,Durand04} and \HyperLogLog~\cite{FlajoletFGM07} sketches use the same cell partition as \PCSA, but deem any cell occupied if it is hit by a dart, or below one hit by a dart in the same column; see Figure~\ref{fig:dartboard-PCSA-HLL}(c).  As a consequence, an explicit
representation of the (\textsf{Hyper})\LogLog{} state requires only $m\log\log U$ bits.

One may regard a \emph{cell} as the set of all hash values 
that, if encountered,
would have the same action on the sketch.
Some sketches do not admit a non-trivial cell partition.  
For example, Cohen's \Bottom-$k$ sketch~\cite{Cohen97}
stores the smallest $k$ hash values encountered, which
can be expressed as a \Dartboard{} sketch with one cell per hash value.  Figure~\ref{fig:bottom-k-dartboard}(a) depicts 
a discrete cell partition with $\lambda=14$ darts, and Figure~\ref{fig:bottom-k-dartboard}(b) shows the state of the \Bottom-4 sketch.
    
\begin{figure}
    \centering
    \begin{tabular}{c@{\hspace{.6cm}}c}
    \includegraphics[scale=.25]{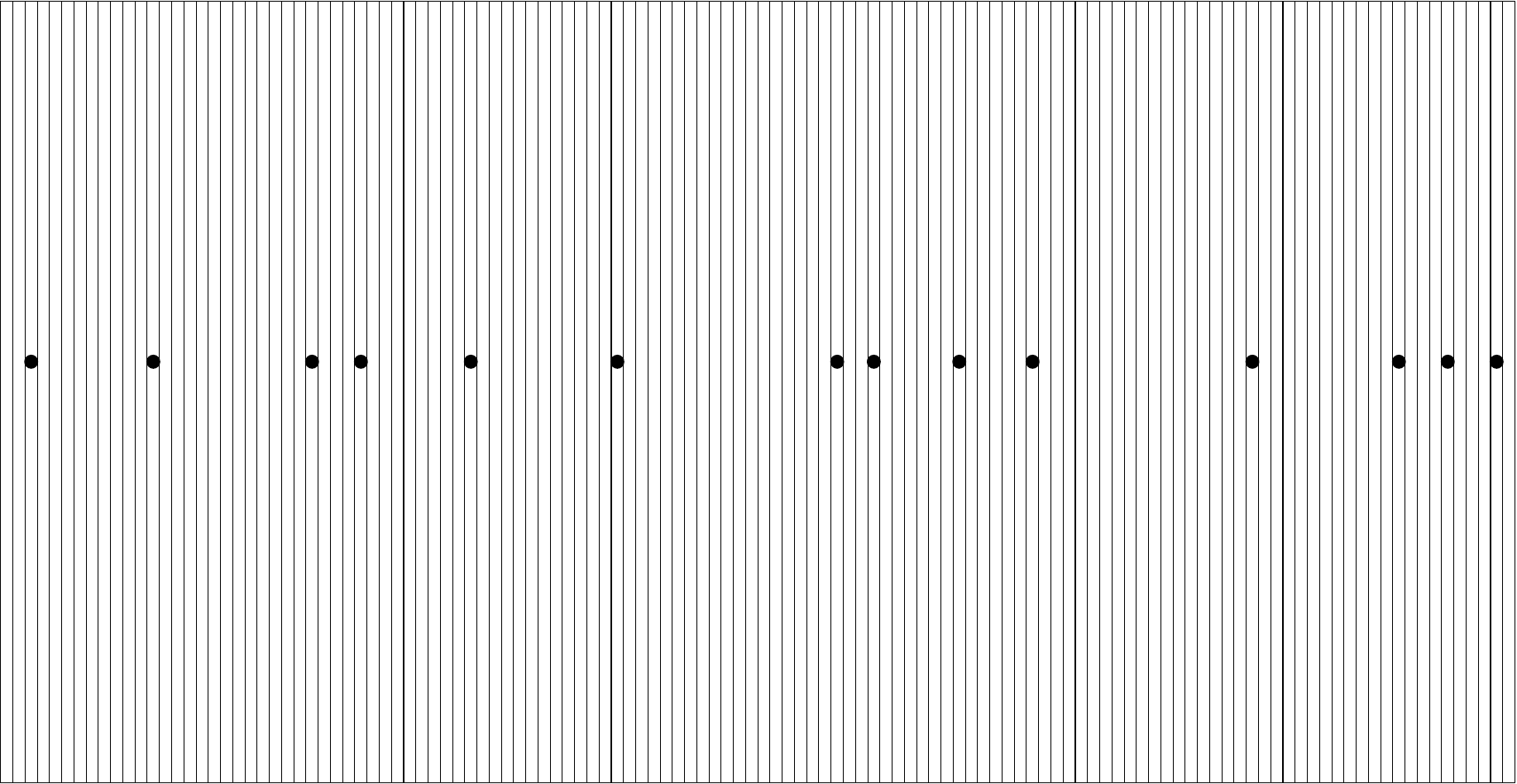}
&   \includegraphics[scale=.25]{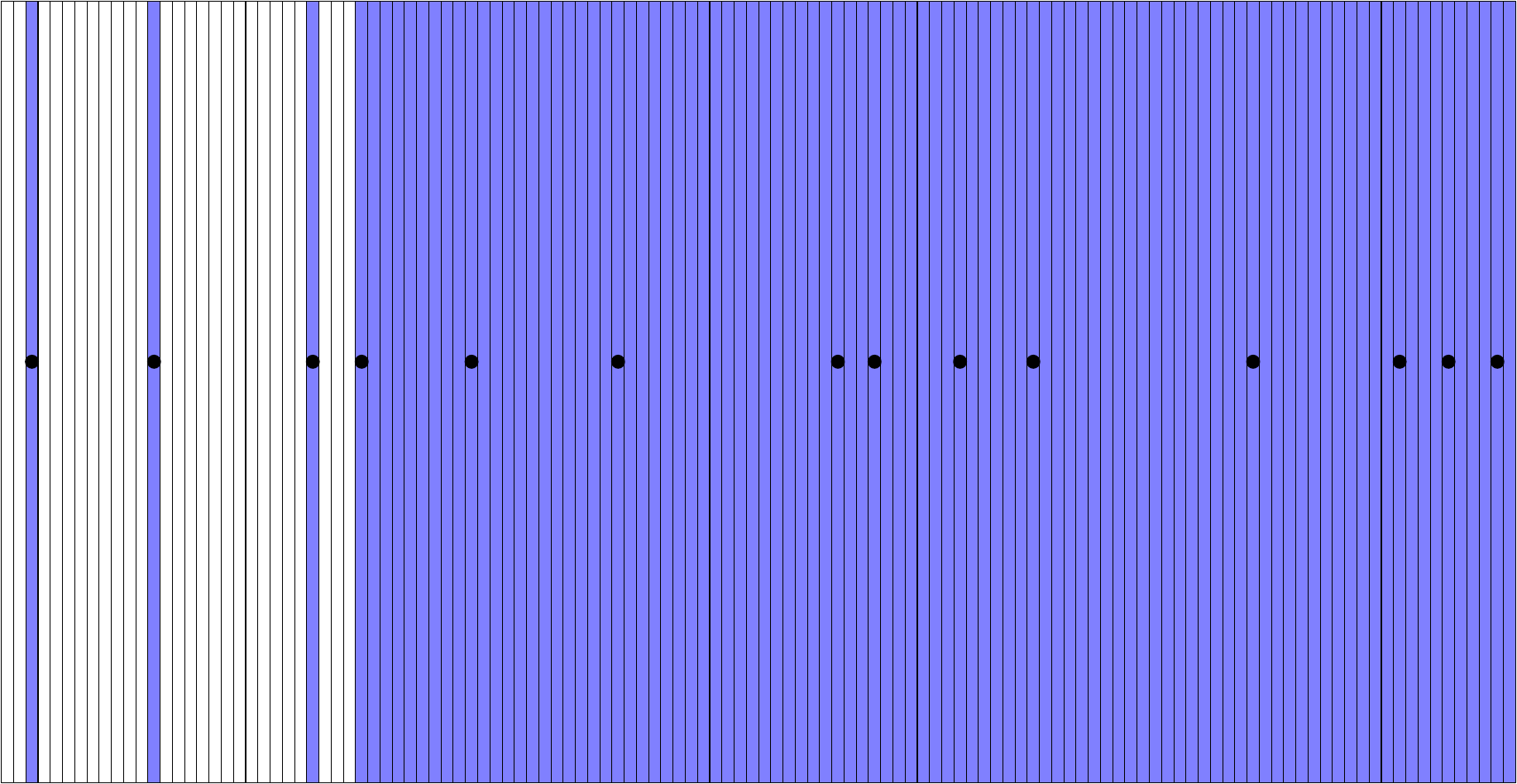}\\
    \textbf{(a)} & \textbf{(b)}
    \end{tabular}
    \caption{\textbf{(a)} The \Dartboard{} is partitioned into individual hash values. \textbf{(b)} A \Bottom-4 sketch keeps the smallest 4 hash values, hence all occupied (blue) cells have no effect on the sketch.}
    \label{fig:bottom-k-dartboard}
\end{figure}

The \Dartboard{} is only used to visualize the state and behavior of the sketch.  The problem of computing an estimate $\hat{\lambda}(S)$ of the true cardinality from the state $S$ is a separate, and challenging problem.  In the original \PCSA{} sketch, Flajolet and Martin~\cite{FlajoletM85} do not pay attention to the entire state $S$ but the vector $(z(1),\ldots,z(m))$, where $z(j)$ is the index of the largest free cell in column $j$.  Then
\[
\hat{\lambda}(S) \propto m\cdot 2^{m^{-1}\sum_{j=1}^m z(j)}.
\]
The rationale is that $2^{z(j)}$ is a reasonable 
estimate for the number of elements in column $j$, 
which is $\lambda/m$ in expectation, 
and $2^{m^{-1}\sum_{j=1}^m z(j)}$ represents the \emph{geometric mean} of the individual column estimates.
Flajolet and Martin~\cite{FlajoletM85} proved that this estimator achieves a standard error of $0.78/\sqrt{m}$.
A simpler \emph{coupon collector} estimator for \PCSA{} 
was proposed later by Lang~\cite{Lang17}, which only considers the
total number of occupied cells.  Lang's informal analysis 
showed it has standard error $.7/\sqrt{m}$. 
Wang and Pettie~\cite{WangP23} introduced a class of \emph{generalized remaining area} estimators 
called $\tau$-\GRA, and proved 
that for $\tau\approx 0.34$, the $\tau$-\GRA{} estimator 
for \PCSA{} has standard error $0.66/\sqrt{m}$.\footnote{Lang's estimator corresponds to $0$-\GRA.}

In (\Hyper)\LogLog{} every estimator naturally pays attention to
the vector $(y(1),\ldots,y(m))$, where $y(j)$ is the index of the maximum occupied cell in column $j$.  
Following~\cite{FlajoletM85}, 
Durand and Flajolet's \LogLog{} estimator takes the geometric mean of the individual column estimates:
\[
\hat{\lambda}(S) \propto m\cdot 2^{m^{-1}\sum_{j=1}^m y(j)},
\]
which they proved has standard error $1.3/\sqrt{m}$.
This estimator is heavily influenced by a few outliers (large $y$-values), and Durand and Flajolet found empirically that the standard error improves if the largest 30\% of the $y$-vector is simply discarded.  The \HyperLogLog{} estimator of Flajolet, Fusy, Gandouet, and Meunier~\cite{FlajoletFGM07} takes a more principled approach to weighting large $y$-values appropriately,
by taking the \emph{harmonic mean} of the individual column estimates rather than the geometric mean:
\[
\hat{\lambda}(S) \propto m^2 \left(\sum_{j=1}^m 2^{-y(j)}\right)^{-1}.
\]
Flajolet et al.~\cite{FlajoletFGM07} proved the 
standard error is $1.04/\sqrt{m}$.
Wang and Pettie~\cite{WangP23} noted that
the Durand-Flajolet~\cite{DurandF03} and Flajolet et al.~\cite{FlajoletFGM07} estimators corresponded to $0$-\GRA{} 
and $1$-\GRA, respectively, and that by choosing an appropriate fractional $\tau\approx 0.89$, the $\tau$-\GRA{} 
estimator
for the (\Hyper)\LogLog{} sketch has a slightly improved standard error of $1.036/\sqrt{m}$.

See Cohen~\cite{Cohen97}, Giroire~\cite{Giroire09}, 
and Chassaing and Gerin~\cite{ChassaingG06} for cardinality estimators in the \Bottom-$k$ family of sketches.

\begin{table}
    \centering
\scalebox{.76}{
    \begin{tabular}{|l|l|l|}
\multicolumn{1}{l}{{\bf\Large\sc Random Oracle Model}}\\
\multicolumn{1}{l}{\large\sc Mergeable Sketches}
& \multicolumn{1}{l}{\large\sc Sketch Size (Bits)}
& \multicolumn{1}{l}{\large\sc Approximation Guarantee}\\\hline
Flajolet \& Martin \hfill {\small (\textsf{PCSA})} 1983 & $m\log U$   & Std.~err.~$0.78/\sqrt{m}$\\\hline
Flajolet \hfill {\small (\textsf{AdaptiveSampling})} 1990 & $m\log U + \log\log U$   & Std.~err.~$1.21/\sqrt{m}$\\\hline
Durand \& Flajolet \hfill {\small (\textsf{LogLog})} 2003 & $m\log\log U$           & Std.~err.~$1.3/\sqrt{m}$\\\hline
Giroire \hfill {\small (\textsf{MinCount})}  2005         & $m\log U$             & Std.~err.~$1/\sqrt{m}$\\\hline
Chassaing \& Gerin \hfill {\small (\textsf{MinCount})}  2006 & $m\log U$     & Std.~err.~$1/\sqrt{m}$\\\hline
Estan, \hfill \ \ \ \ \ \rb{-2.0}{{\small (\textsf{Multires.Bitmap})} 2006} & \rb{-2.0}{$m\log U$}  & \rb{-2.0}{Std.~err.~$O(1/\sqrt{m})$}\\
Varghase \& Fisk                             & & \\\hline
Beyer, Haas, Reinwald \hfill \rb{-2.5}{2007}   & \rb{-2.5}{$m\log U$}     & \rb{-2.5}{Std.~err. $1/\sqrt{m}$}\\
Sismanis \& Gemulla         &&\\\hline
Flajolet, Fusy,  \hfill \ \ \ \ \ \ \ \ \ \ \  \rb{-2.0}{{\small (\textsf{HyperLogLog})}}
\rb{-2}{2007} & \rb{-2}{$m\log\log U$}    & \rb{-2}{Std.~err.~$1.04/\sqrt{m}$}\\
Gandouet \& Meunier    & &\\\hline
Lumbroso                    \hfill 2010 & $m\log U$                     & Std.~err.~$1/\sqrt{m}$\\\hline
Lang \hfill {\small (\textsf{Compressed FM85})} 2017 & $\approx \log U + 2m$ {\small (experimental)} & Std.~err.~$\approx 1/\sqrt{m}$ {\small (experimental)}\\\hline
                    & $(1+o(1))(H_0/I_0)m$ &\\
{\bf new} \hfill {\small (\textsf{Fishmonger})} \phantom{2020} 
    & \hspace*{0.5cm}where  $m=\omega(\log^2\log U),$ & Std. err. $1/\sqrt{m}$\\
    & \hspace*{0.5cm}and  $H_0/I_0 \approx 1.98016$ & \\\hline\hline
\multicolumn{3}{l}{}\\
\multicolumn{1}{l}{\large\sc Non-mergeable Sketches}\\\hline 
Chen, Cao, Shepp\hfill \rb{-2.5}{{\small (\textsf{S-Bitmap})} 2009}  & \rb{-2.5}{$m$}                           & \rb{-2.5}{Std.~err.~$\approx \frac{\ln(eU/m)/2}{\sqrt{m}}$} \\
\& Nguyen            &                            &\\\hline
Helmi, Lumbroso,  \hfill \rb{-2.5}{{\small (\textsf{Recordinality})} 2012} & \rb{-2.5}{$(1+o(1))m\log U$}  & \rb{-2.5}{Std.~err.~$\tilde{O}(1)/\sqrt{m}$}\\
Mart\'{i}nez \& Viola               &                                     &\\\hline
Cohen \hfill {\small (\textsf{Martingale LogLog})} \rb{-2.5}{2014}  & $m\log\log U + \log U$      & Std.~err.~$0.833/\sqrt{m}$\\
Ting \hfill {\small (\textsf{Martingale Bottom}-$m$)} \phantom{2014} & $(m+1)\log U$                  & Std.~err.~$0.71/\sqrt{m}$\\\hline
Janson, Lumbroso, \hfill \rb{-2}{\small (\textsf{HyperTwoBits}) 2025} & \rb{-2}{$2m+\log U$} \hfill \rb{-2}{$(m\leq 4096)$}                      & \rb{-2}{Std.~err.~$1.46/\sqrt{m}$}\hfill \rb{-2}{$(m\leq 4096)$}\\
and Sedgewick &&\\\hline\hline
\multicolumn{3}{c}{\ }\\
\multicolumn{1}{l}{{\Large\sc Standard Model}}\\\hline
Alon, Matias \& Szegedy \hfill 1996     & $O(\log U)$                    & $(\epsilon, 2/\epsilon)$-approx., $\epsilon\ge 2$\\\hline
Gibbons \& Tirthapura    \hfill 2001   & $O(\epsilon^{-2}\log U\log\delta^{-1})$ & $(\epsilon,\delta)$-approx.\\\hline
Bar-Yossef, Kumar \& Sivakumar \hfill 2002 & $O(\epsilon^{-3}\log U\log\delta^{-1})$        & $(\epsilon,\delta)$-approx.\\\hline
Bar-Yossef, Jayram, Kumar,\hfill \rb{-2.5}{2002}  & \rb{-2.5}{$O\left(\left[\epsilon^{-2}\log\log U +\log U\right]\log\delta^{-1}\right)$}        & \rb{-2.5}{$(\epsilon,\delta)$-approx.}\\
Sivakumar \& Trevisan                   &               & \\\hline
Kane, Nelson \& Woodruff \hfill 2015    & $O([\epsilon^{-2} + \log U]\log\delta^{-1})$           & $(\epsilon,\delta)$-approx.\\\hline
\Blasiok  \hfill 2018  & $O(\epsilon^{-2}\log\delta^{-1} + \log U)$     &  $(\epsilon,\delta)$-approx.\\\hline\hline
\multicolumn{3}{c}{\ }\\
\multicolumn{1}{l}{{\Large\sc Lower Bounds}}\\\hline
Trivial                     & $\Omega(\log\log U)$      &   $(O(1),O(1))$-approx. \hfill (\textsc{rand.~oracle})\\\hline
Alon, Matias \& Szegedy \hfill 1996 & $\Omega(\log U)$    &   $(O(1),O(1))$-approx. \hfill \ \ (\textsc{std.~model})\\\hline
Indyk \& Woodruff \hfill 2003 & $\Omega(\epsilon^{-2})$ & $(\epsilon,O(1))$-approx. \hfill (Both)\\\hline
Jayram \& Woodruff \hfill 2011 & $\Omega(\epsilon^{-2}\log\delta^{-1})$ & $(\epsilon,\delta)$-approx. \hfill (Both)\\\hline
{\bf new}  & $(H_0/I_0)m$ & Std.~err.~$ 1/\sqrt{m}$\hfill (Linearizable)\\\hline\hline
    \end{tabular}
}
    \caption{\small Algorithms analyzed in the \textsc{random oracle} model assume oracle access to a 
    uniformly random hash function $h : [U]\rightarrow [0,1]$.  Algorithms
    in the \textsc{standard model} can generate uniformly random bits, but must store any hash functions
    explicitly.
    }
    \label{table:history}
\end{table}

\subsection{Survey of the Standard Model}

In the \textsc{Standard Model} one must explicitly account for the space of
every hash function.  Specifically, a $k$-wise independent function
$h:[D]\rightarrow [R]$ requires $\Theta(k\log(DR))$ bits.  Typically
an $\epsilon$-approximation ($\hat{\lambda}\in [(1-\epsilon)\lambda,(1+\epsilon)\lambda]$) 
is guaranteed with constant probability, and then amplified to
$1-\delta$ probability by taking the median of $O(\log\delta^{-1})$ trials.
The following algorithms are all mergeable.

Gibbons and Tirthapura~\cite{GibbonsT01} rediscovered Flajolet's \textsf{AdaptiveSampling}\footnote{Flajolet~\cite{Flajolet90} attributes the sketch design to Wegmen.} \cite{Flajolet90} 
and proved that it achieves
an $(\epsilon,\delta)$-guarantee using an $O(\epsilon^{-2}\log U\log\delta^{-1})$-bit 
sketch and $O(1)$-wise independent hash functions.
\textsf{AdaptiveSampling} is a variant
of \Bottom-$k$ where we 
store \emph{all} hash values encountered in $[0,2^{-\ell})$, where $\ell\in \mathbb{Z}^+$ is incremented whenever we attempt to store strictly more than $k$ 
hash values.
Bar-Yossef et al.~\cite{Bar-YossefJKST02} 
considered versions of \Bottom-$k$, \LogLog, and \textsf{AdaptiveSampling}, the best of which used space 
$O(((\epsilon^{-2}(\log\epsilon^{-1}+\log\log U) + \log U)\log\delta^{-1})$.
Kane, Nelson, and Woodruff~\cite{KaneNW10} designed a sketch based on \LogLog{}
that has size $O((\epsilon^{-2} + \log U)\log\delta^{-1})$, which
is optimal when $\delta^{-1}=O(1)$ 
as it meets the $\Omega(\epsilon^{-2})$
lower bound of Indyk and Woodruff~\cite{IndykW03} 
(see also Brody and Chakrabarti~\cite{BrodyC09})
and the $\Omega(\log U)$ lower bound of 
Alon, Matias, and Szegedy~\cite{AlonMS99}.
Using more sophisticated techniques, 
\Blasiok~\cite{Blasiok20} derived an 
optimal sketch, also based on \LogLog, 
for all $(\epsilon,\delta)$ with space
$O(\epsilon^{-2}\log\delta^{-1}+\log U)$, which meets
the $\Omega(\epsilon^{-2}\log\delta^{-1})$ lower bound of
Jayram and Woodruff~\cite{JayramW13} for 
any $(\epsilon,\delta)$ guarantee.

\subsection{New Results}

Our goal is to understand the 
\emph{intrinsic tradeoff} between 
space and accuracy in Cardinality Estimation.  This question has been 
answered up to a very large constant factor 
in the \textsc{standard model}
with matching upper and lower bounds of 
$\Theta(\epsilon^{-2}\log\delta^{-1}+\log U)$~\cite{KaneNW10,Blasiok20,IndykW03,JayramW13}.
However, in the \textsc{random oracle} model we can aspire to understand this tradeoff
\emph{precisely}, and to identify the \emph{best} sketch design.

To be specific, consider a sketch 
$S=(S(0),\ldots,S(m-1))$
composed of $m$ i.i.d.~\emph{subsketches} 
over a multiset with cardinality $\lambda$.
The space required to store this sketch is clearly linear in $m$, and the variance of the cardinality estimates will scale with $1/m$.  A natural way to measure the efficiency of the sketch design is to analyze its 
\emph{memory-variance product} (\MVP), 
which should have no dependence on $m$, 
as $m\to \infty$.
One immediately sees that the \MVP{} of a sketch 
might not be an intrinsic measure of the sketch design itself, but on precisely how the sketch state is 
encoded and how the 
cardinality estimates are computed 
from the sketch state.

To evaluate sketch \emph{designs} on a level playing field we shall assume their states are encoded optimally and that they use statistically optimal estimators.
To make these notions precise we appeal
to two of the influential notions
of ``information'' defined in the 20th century,
namely 
\emph{Shannon entropy} and \emph{Fisher information}.
Shannon entropy~\cite{CoverT06} controls
the (expected) space in bits needed to encode an object drawn from some distribution 
and \emph{Fisher information}
limits the variance of an asymptotically 
unbiased estimator, via the
Cram\'{e}r-Rao lower bound~\cite{CasellaB02,Vaart98}.

Each subsketch is the outcome of some 
\emph{experiment}. 
We assume 
these experiments are \emph{informative}, in the sense that any two cardinalities
$\lambda_0,\lambda_1$ induce distinct distributions on the sketch state $S$.  Under this condition
and some mild regularity conditions, it is well known~\cite{CasellaB02,Vaart98} that the 
Maximum Likelihood Estimator (MLE):
\[
\hat{\lambda}(S) = \argmax_{\lambda} \Pr(S \mid \lambda)
\]
is asymptotically unbiased and meets the Cram\'{e}r-Rao lower bound:
\[
\lim_{m\rightarrow \infty} 
\sqrt{m}\left(\hat{\lambda}(S)-\lambda\right) 
    \sim 
\Normal\left(0, \frac{1}{I_{S(0)}(\lambda)}\right).
\]
Here $I_{S(0)}(\lambda)$ is the \emph{Fisher information} 
number of $\lambda$ associated with any one component 
of the vector $S$, 
where $I_S(\lambda)=I_{S(0)}(\lambda)+\cdots+I_{S(m-1)}(\lambda) = mI_{S(0)}(\lambda)$.
This implies that as $m$ gets large, 
$\hat{\lambda}(S)$ tends toward a normal 
distribution $\Normal\left(\lambda, \frac{1}{I_S(\lambda)}\right)$
with variance $1/I_S(\lambda) = 1/(m\cdot I_{S(0)}(\lambda))$.
(See Section~\ref{sect:info-theory}.)

Suppose for the moment that $I_{S}(\lambda)$ is scale-free, 
in the sense that we can write
it as $I_{S}(\lambda) = \mathcal{I}(S)/\lambda^2$, 
where $\mathcal{I}(S)$ is some number that
does not depend on $\lambda$.
We can think of $\mathcal{I}(S)$ as measuring the \emph{value}
of experiment $S$ to estimating the parameter $\lambda$,
but it also has a \emph{cost}, namely the space required to store
the outcome of $S$.  By Shannon's source-coding theorem 
we cannot beat $H(S \mid \lambda)$ bits on average, which we 
also assume for the time being is scale-free, and can be written $\mathcal{H}(S)$, 
independent of $\lambda$.
We measure the \emph{efficiency} of an experiment by its
\underline{Fi}sher-\underline{Sh}annon ($\fish$) 
number, defined to be the limiting ratio of its cost to its value, as $m\to \infty$.
\[
\fish(S) = \frac{\mathcal{H}(S)}{\mathcal{I}(S)}.
\]
In particular, this implies that using sketching scheme $S$ to
achieve a standard error of $\sqrt{1/b}$ (variance $1/b$) 
requires $\fish(S)\cdot b$ bits of storage on 
average,\footnote{Set $m$ such that $b = \mathcal{I}(S) = m\cdot \mathcal{I}(S(0))$.  
The expected space required
is $m \cdot \mathcal{H}(S(0)) = b(\mathcal{H}(S(0))/\mathcal{I}(S(0))) 
= b\cdot \fish(S)$.} i.e.,
lower $\fish$-numbers are superior.
The actual definition of $\fish$ (Section~\ref{sect:fish})
is slightly more complex in order to deal with sketches $S$
that are not \emph{strictly} scale-invariant.

\medskip

Our main results are as follows.

\begin{enumerate}
    \item[(1)] Let $\qPCSA$ be the natural base-$q$ analogue of $\PCSA$, which is $2$-$\PCSA$.
    We prove that the $\fish$-number of $\qPCSA$ actually does not depend on $q$ at all, and is precisely:
    \begin{align*}
        \fish(\qPCSA) &= \frac{H_0}{I_0} \approx 1.98016.
    \intertext{where}
            H_0 &= \frac{1}{\ln 2} + \sum_{k=1}^\infty\frac{1}{k}\log_2 \left(1+1/k\right),\\
            I_0 &= \zeta(2) = \frac{\pi^2}{6}.
    \end{align*}
    Let $\qLL$ be the natural base-$q$ analogue of \textsf{LogLog} = $2$-$\LL$.  
    Whereas the Fisher information for $\qPCSA$
    is expressed in terms of the \emph{Riemann zeta} function ($\zeta(2)$), 
    the Fisher information of $\qLL$ is expressed in terms
    of the \emph{Hurwitz zeta} function $\zeta(2,\frac{q}{q-1}) = \sum_{k\ge 0} (k+\frac{q}{q-1})^{-2}$.  
    We prove that $\qLL$ is always worse than $\PCSA$, but approaches
    the efficiency of $\PCSA$ in the limit, i.e.,
    \[
        \forall q.\; \fish(\qLL) > H_0/I_0 \quad \mbox{ but } \lim_{q\to\infty} \fish(\qLL) = H_0/I_0.
    \]
    \item[(2)] The results of (1) should be thought of as \emph{lower bounds} on implementing
    compressed representations of $\qPCSA$ and $\qLL$.  We give a new sketch called $\fishmonger$ 
    based on an entropy-compressed version of \PCSA{} with maximum likelihood estimation.
    It is guaranteed that with probability $1-1/\poly(m\log U)$,
    the space of \fishmonger{} 
    is $(1+o(1))(H_0/I_0)m + O(\sqrt{m\log(m\log U)} + \log^2\log U)$ bits \emph{at all times}
    and its standard error
    is $1/\sqrt{m}$ \emph{at all times}.  When $m=\omega(\log^2\log U)$ this is roughly $1.98m$ bits.
    
    \item[(3)] Is it possible to improve the space-variance tradeoff offered by \fishmonger?  Specifically, is it possible to design a mergeable 
    sketch whose \fish-number is strictly smaller than $H_0/I_0$? 
    This is a difficult problem, 
    as it \emph{is} possible to beat $H_0/I_0$ 
    when one drops the \emph{mergeability} criterion (see Section~\ref{sect:departures-from-the-dartboard-model}) and characterizing the space of all mergeable sketches is non-trivial.  
    We take two steps toward answering this question.  First, we give a simple characterization of mergeable sketches in terms of the state space.  
    Second, we define a natural subclass of mergeable sketches called \emph{linearizable} sketches, and prove that no member of this class has $\fish$-number strictly 
    smaller than $H_0/I_0$.\footnote{A \Dartboard{} sketch is \emph{linearizable} iff there is a permutation of its cells
    so that the occupied/free status of each cell depends only on (1) whether it was hit by a dart, and (2) the occupied/free status of earlier cells in the permutation.}
    All of the popular canonical 
    sketch designs are linearizable, 
    such as \Bottom-$k$, 
    (\Hyper)\LogLog, and \PCSA.
    (Technically \textsf{AdaptiveSampling} is not linearizable.)
    We take this as strong circumstantial evidence that $\fishmonger$ 
    is not merely an efficient sketch, 
    but an \emph{optimal} sketch, 
    up to $o(m)$ bits.
    Indeed, the structure of the proof suggests 
    that the \PCSA{} sketch design is essentially
    the \emph{unique}
    optimum design.
\end{enumerate}

\subsection{Related Work}

After completing this project we learned that many of
the facts proved in this paper were discovered experimentally earlier.

The idea of compressing a sketch state goes all the way back to Flajolet and Martin~\cite{FlajoletM85}, who observed that the \PCSA{} matrix has mostly 1s in the low-order bits and 0s in the high order bits.  They suggested encoding a sliding window of the 8 most relevant bits across the sketch matrix but did not attempt to analyze anything along these lines.
In her Ph.D.~thesis, Durand~\cite[p. 136]{Durand04}
proved that there is a prefix-free code for encoding the $m$ columns of a \LogLog{} sketch with expected length at most $3.01m$.  However, she did not analyze the entropy of \LogLog.

The first work that explicitly foreshadowed ours was Scheuermann and Mauve~\cite{ScheuermannM07}, who demonstrated experimentally that an entropy-compressed version of \PCSA{} (using Flajolet and Martin's original estimator~\cite{FlajoletM85}) is \emph{slightly} 
superior than an entropy-compressed version of \HyperLogLog~\cite{FlajoletFGM07}, 
in terms of their memory-variance products (\MVP).
Lang~\cite{Lang17} took the next step by replacing
the estimators of \PCSA{} and \HyperLogLog{} with
the \emph{minimum description length} estimator, which is essentially equivalent to MLE in this context.
Under optimal compression and asymptotically optimal estimation, Lang~\cite{Lang17} showed experimentally
that the \MVP{} of (compressed + MLE) \PCSA{} is roughly 2, 
and superior to the \MVP{} of (compressed + MLE) \LogLog, 
thus foreshadowing our result (1).
Lang's experiments~\cite{Lang17} on compressed sketching and maximum likelihood estimation had an immediate impact on the \emph{practice} of sketching.  
The Apache DataSketches library~\cite{DataSketches} 
now includes a version of Lang's sketch called \CPC{} (Compressed Probabilistic Counting),
which uses a fast ``off-the-shelf'' compressor in lieu 
of entropy-compression and Lang's coupon-collector estimator~\cite{Lang17} in lieu of MLE.

The idea of applying maximum likelihood estimation to sketches is not new.  
See Chassaing and Gerin~\cite{ChassaingG06} and Clifford and Cosma~\cite{CliffordC12} for MLE estimators for sketches in the \Bottom-$k$ family.  
Ertl~\cite{Ertl17} studied the computational complexity of MLE in \LogLog{} sketches.
Cohen, Katzir, and Yehezkel~\cite{CohenKY17}
looked at MLE estimators for cardinality of set intersections.

\subsection{Inside and Outside the \DartboardModel}\label{sect:departures-from-the-dartboard-model}

So far we have distinguished a number of classes of 
sketches such as 
\textsf{Linearizable} sketches, 
\textsf{Mergeable} sketches, 
and \Dartboard{} 
sketches, which have the following relationship:

\[
\textsf{Linearizable} \subseteq \textsf{Mergeable} \subseteq \textsf{Dartboard} \subseteq \text{All Sketches}
\]

Each of these containments is strict.  
\AdaptiveSampling{} is the only published sketch that is in $\textsf{Mergeable}\cap\overline{\textsf{Linearizable}}$.\footnote{It is easy to manufacture artificial transition functions that are in
$\textsf{Mergeable}\cap \overline{\textsf{Linearizable}}$,
but none of the organically proposed sketches are in this set.}
The \SBitmap{} sketch of Chen, Cao, Shepp, and Nguyen~\cite{ChenCSN11} 
is the only published sketch in
$\Dartboard \cap\overline{\textsf{Mergeable}}$.  
The sketch state $S\in\{0,1\}^m$ is a bitmap, initially 0.
Hash values $h(a) \in [m]\times [0,1]$ are interpreted as a column index and a height in the column.
The transition function is defined by a sequence of real values $1 \geq p_0 > p_1 > \cdots > p_{m-1} > 0$, 
where upon seeing $h(a)=(i,\nu)$, if $S(i)=0$ we set $S(i)\gets 1$ iff $\nu\leq p_{\textsf{weight}(S)}$. Here $\textsf{weight}(S)$ is the number of 1s set in $S$.
The \SBitmap{} can be expressed as a \Dartboard{} sketch partitioned 
into $m$ equally spaced columns and up to $m+1$ rows determined by the $(p_0,\ldots,p_{m-1})$ sequence.  However, one may convince oneself that the transition function is \underline{not} commutative and idempotent: the state of the sketch depends on the \emph{order} in which elements are processed, 
but is not sensitive to duplicates.
Figure~\ref{fig:classification}(b,c) gives an example of 
two states that can be reached by scanning the input in 
different orders.

\begin{figure}
\centering
    \begin{tabular}{c@{\hspace{.5cm}}c@{\hspace{.5cm}}c}
        \includegraphics[scale=.15]{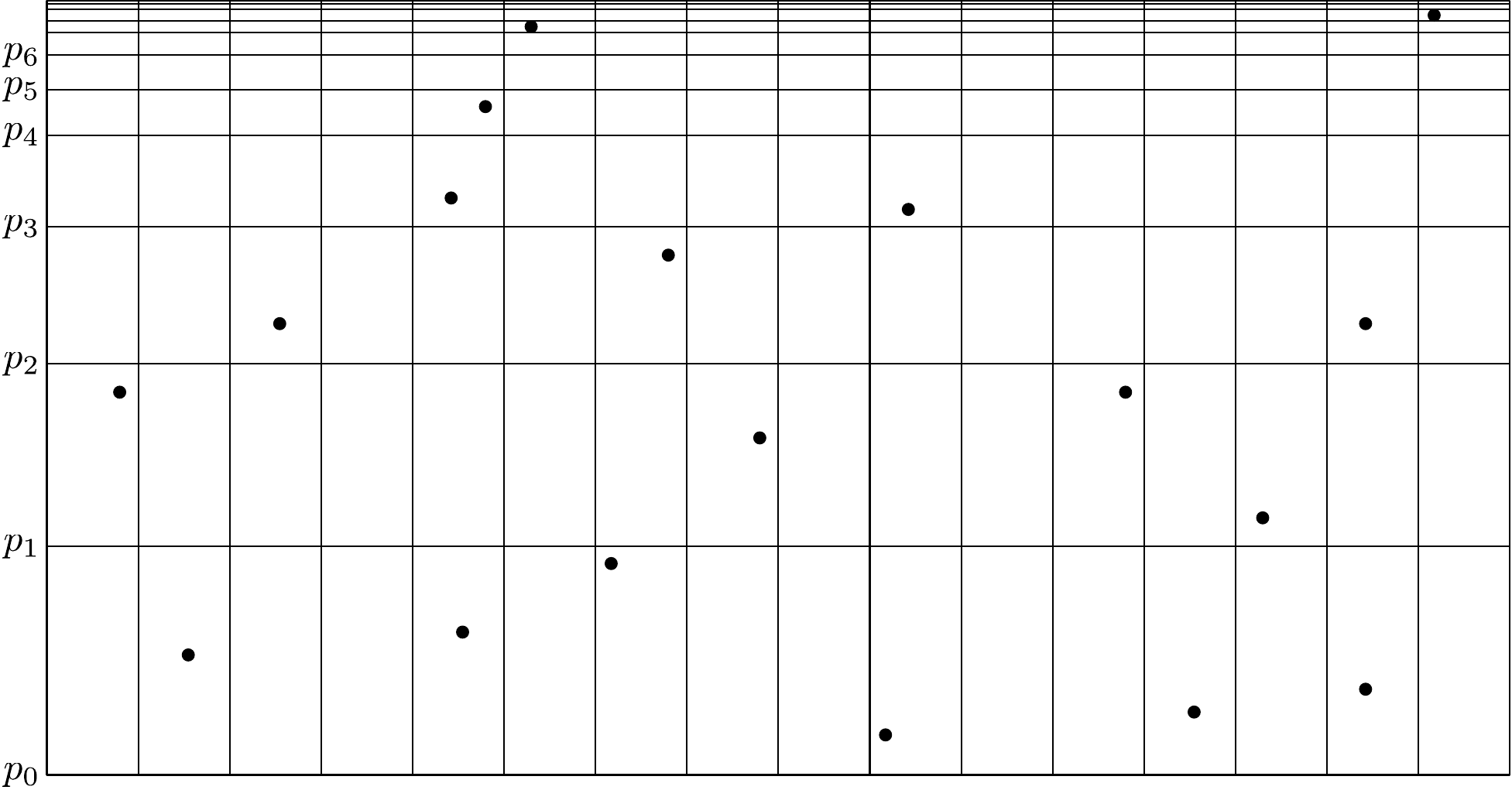}
    &   \includegraphics[scale=.15]{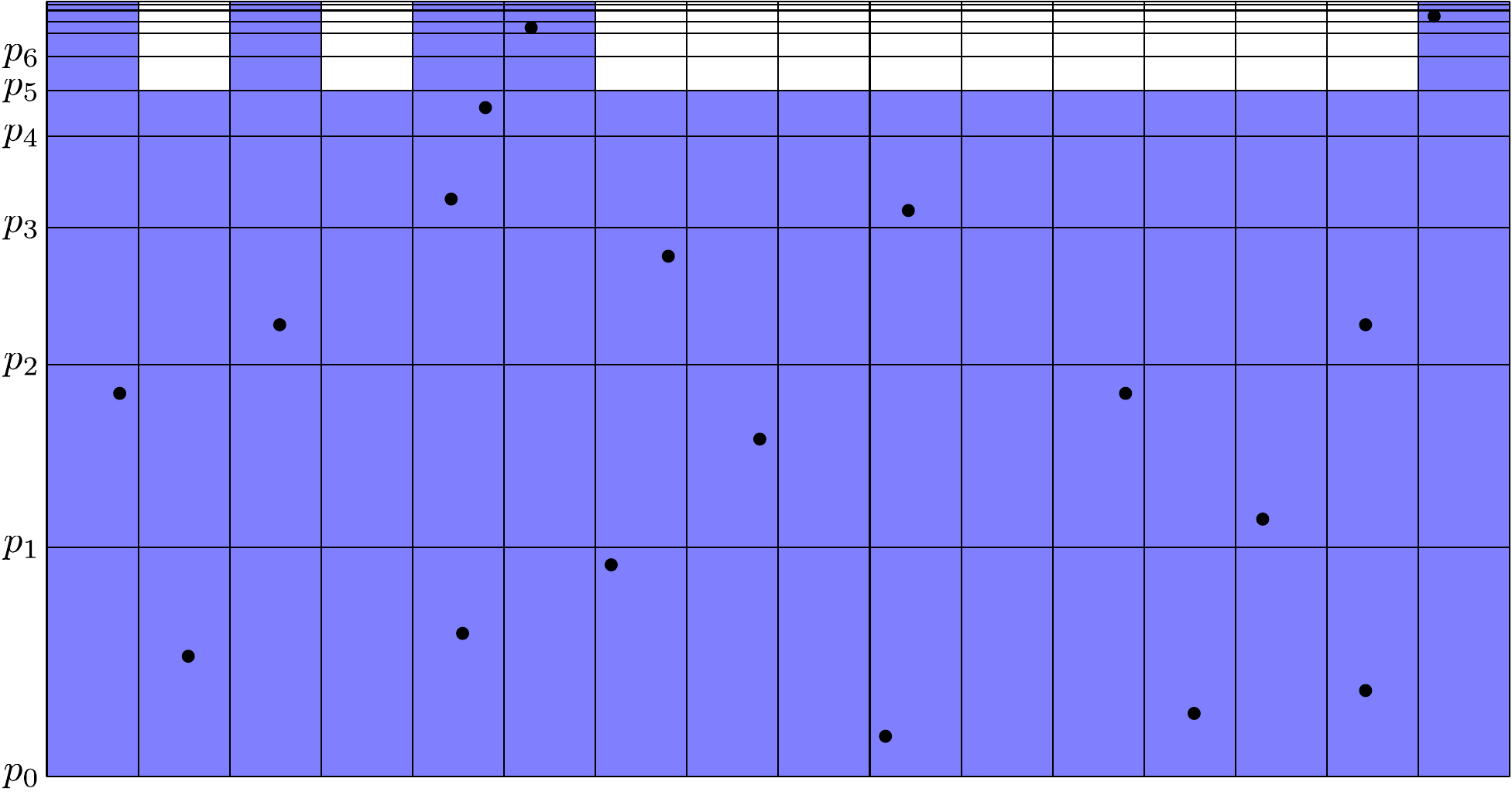}
    &   \includegraphics[scale=.15]{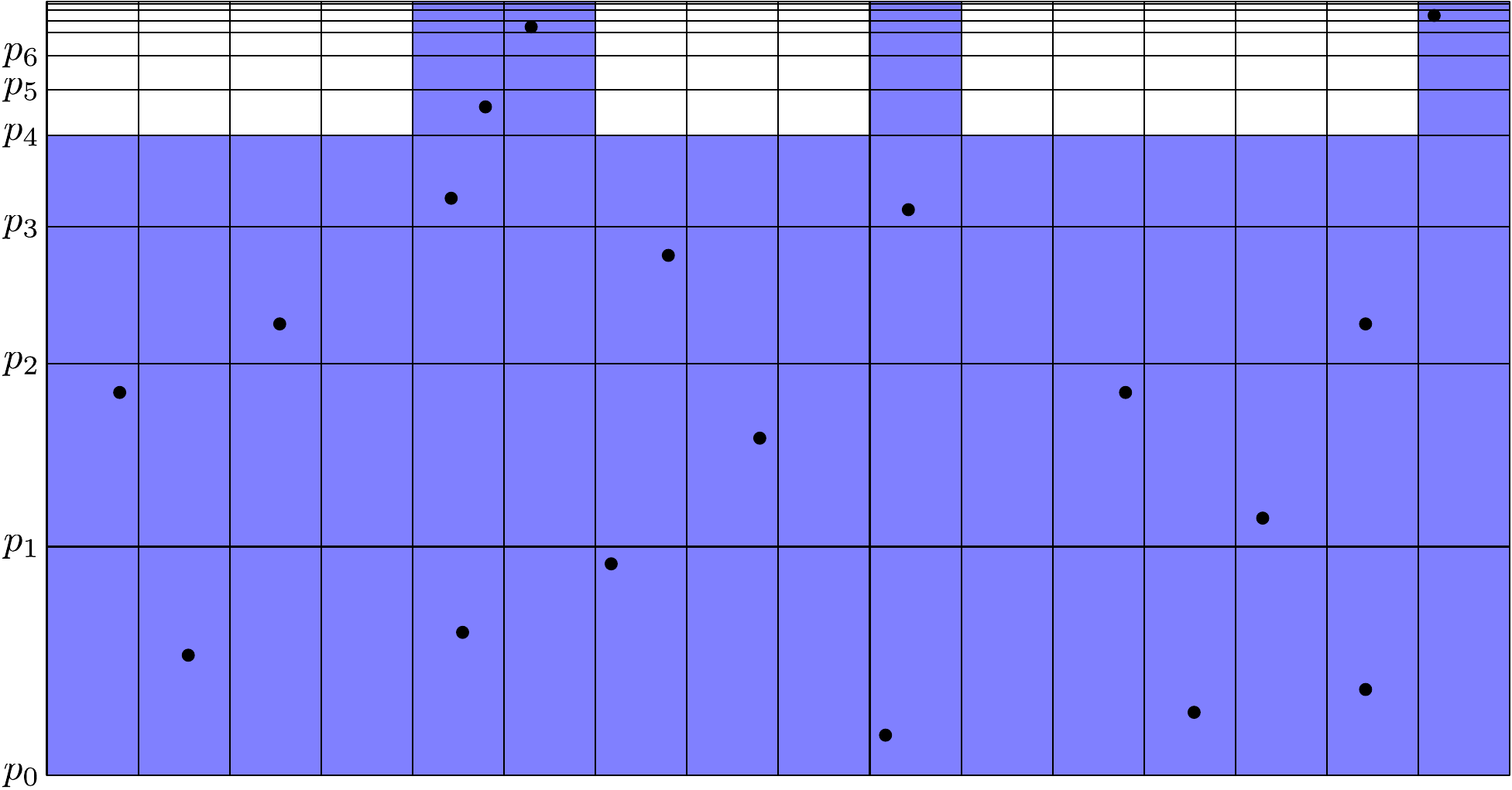}\\
    \textbf{(a)} &  \textbf{(b)} &  \textbf{(c)} 
    \end{tabular}
    \caption{\label{fig:S-Bitmap}\textbf{(a)} The \SBitmap{} \Dartboard{} partition, with $\lambda=17$ darts.  \textbf{(b)} The state of 
    \SBitmap, if the darts were processed in left-to-right order.  \textbf{(c)} The state of the \SBitmap, if the darts were processed in top-to-bottom order.  The \SBitmap's transition function is not commutative and idempotent and is therefore not mergeable.}
\end{figure}

\begin{figure}[h!]
    \centering
    \scalebox{.4}{\includegraphics{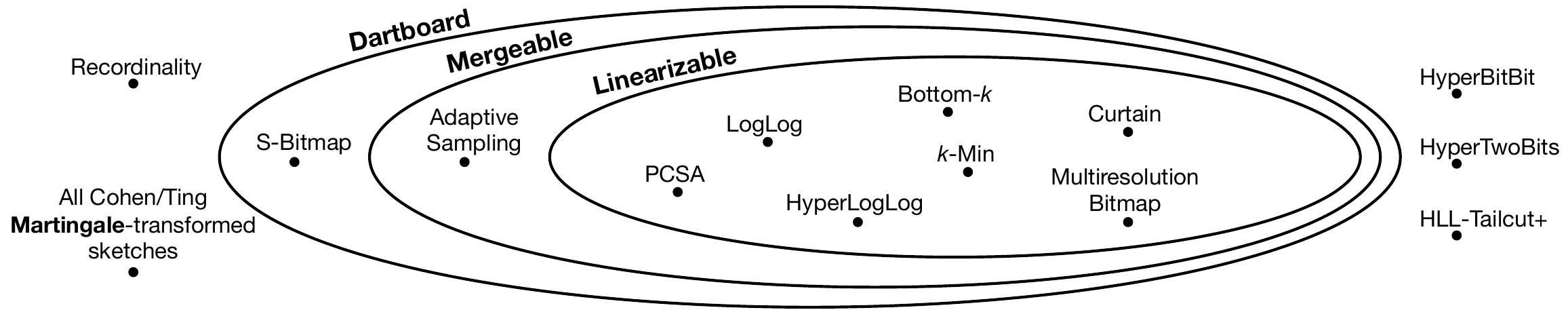}}
    \caption{A  classification of sketching algorithms for cardinality estimation.}
    \label{fig:classification}
\end{figure}

\subsubsection{Departures from the \DartboardModel}

The axioms of the \Dartboard{} model define a natural set of constraints on cardinality sketches, but do they just rule out fruitful lines of inquiry?  We now have many examples of violations of the \Dartboard{} model, which can be put in two classes:

\begin{enumerate}
    \item[(1)] Sketches that directly violate the axioms of the \DartboardModel, usually allowing \emph{occupied} cells to become \emph{free} as a result of trying to compress the state of a \Dartboard{} sketch in a ``lossy'' fashion.  (Sketches of this type appear to the right of the \Dartboard{} oval in Figure~\ref{fig:classification}.)

    \item[(2)] Sketches that consist of $(S,\hat{\lambda})$, where $S$ is a valid \textsf{Dartboard} sketch and $\hat{\lambda}$ is some extra information that is either an explicit cardinality estimate, 
    or used to generate one.
    (Sketches of this type appear to the left of the \Dartboard{} oval in Figure~\ref{fig:classification}.)
\end{enumerate}

Many have noted that in the (\Hyper)\LogLog{} vector $(y_1,\ldots,y_m)$, 
the $y_i$ are typically $L \approx \log(\lambda/m)$, plus or minus some small deviation.
If one only permits $y_i \in [L,L+C]$, $C=O(1)$, this effectively allows occupied cells 
to later become free.  (At one point in time some $y_i=L+C$ cannot be incremented, 
but when $L\gets L+1$ is incremented, $y_i$ becomes eligible to be incremented.)
As a consequence, the distribution of the sketch is not solely a function of the cardinality, 
but on the \emph{multiplicity} of elements and the \emph{order} in which they are processed.  
Several sketches can be viewed as lossy compressions of (\Hyper)\LogLog, such
as \textsf{HLL-Tailcut+}~\cite{XiaoCZL20}, \HyperBitBit~\cite{Sedgewick}, and \HyperTwoBits~\cite{JansonLS25}.
They all suffer from the same problem: as the number of
subsketches $m\to \infty$, the standard error is lower bounded by an absolute
constant, rather than scaling as $1/\sqrt{m}$.  
See~\cite[Appendix B]{PettieW20-arxiv} to see experimental results from 
an attack that applies to any sketch that violates
Axiom 3 of the \DartboardModel~\cite{Sedgewick,XiaoCZL20,JansonLS25}.
(Note, however, that Janson et al.~\cite{JansonLS25} 
do not claim that \HyperTwoBits{} should be used with \emph{any} number $m$ of subsketches. Their rigorous analysis holds when $m\leq 4096$ and makes no guarantees for larger $m$.)

Chakraborty et al.~\cite{ChakrabortyVM22} proposed a sketch that
is similar in spirit to \AdaptiveSampling, but rather than use
a single hash value $h(a)$ whenever $a$ is encountered, 
it randomly generates a new hash value for each copy of $a$ 
in the data stream.  As a consequence, the distribution of the sketch
depends on the order and multiplicity of elements.  In particular,
for any fixed size sketch $m$ and any $\lambda > m$ and 
$\lambda^* > \lambda$, there is a sufficiently long data stream 
$\mathcal{A}=(x_1,\ldots,x_N)$ with $|\mathcal{A}|=\lambda$ that,
with high probability, causes the Chakraborty et al.~\cite{ChakrabortyVM22} 
data structure to report an estimate greater than $\lambda^*$.  
The length of the data stream $N=N(m,\lambda,\lambda^*)$ is quite large, so this attack can be mitigated
by making $m$ depend on the stream length $N$ (\emph{including} duplicates), 
a property never seen in the cardinality estimation literature.
The space complexity of~\cite{ChakrabortyVM22} is $O(\epsilon^{-2}\log(N/\delta)\log U)$ bits to achieve an $(\epsilon,\delta)$-guarantee. The primary goal of Chakraborty et al.~\cite{ChakrabortyVM22} was not efficiency but mathematical accessibility, and along this metric
they succeeded in showing that it is possible 
to estimate cardinality without 
using explicit hash functions or random oracles.

\medskip 

The \SBitmap{} was discovered in 2011. It managed to stay within the \textsf{Dartboard} model while losing mergeability. 
Independently, in 2012, Helmi, Lumbroso, Mart\'{i}nez, and Viola~\cite{HelmiLMV12}
proposed the \Recordinality{} sketch $(S,c)$, 
which consists of a standard \Bottom-$k$ sketch $S$ 
and a counter $c$ of the \emph{number} of times $S$ has changed state. 
The premise is that $h(\mathcal{A})$ induces a random permutation on the first occurrences 
of distinct elements in $\mathcal{A}$, and the distribution of $c$ is a function of $\lambda$.
Although $S$ itself is mergeable, there is no 
way to merge the $c$ components of two \Recordinality{} sketches.

In 2014 Cohen~\cite{Cohen15} and Ting~\cite{Ting14}
independently invented a \emph{generic} method for transforming any \textsf{Dartboard} sketch into 
a better non-mergeable sketch,
which was later called 
the \emph{martingale transform} by Pettie, Wang, and Yin~\cite{PettieWY21}.
The premise of the martingale transform is very simple.  Let $S_0$ be the initial state of the \Dartboard,
and $(S_0,\ldots,S_t)$ be the sequence of \emph{distinct} states that we pass through when sequentially scanning a multiset $\mathcal{A}=(a_1,a_2,\ldots,a_N)$.  A mergeable sketch must estimate 
$\lambda = |\mathcal{A}|$ based solely on the final state $S_t$, 
but what if we had access to the entire history $(S_0,\ldots,S_t)$?  
Recall that the \Dartboard{} has unit area. Define $p_i$
to be the total area of free cells in $S_i$.  If we conflate time with \emph{distinct} elements encountered, the time that we spend in state $S_i$ 
is a $\Geometric(p_i)$ random variable, 
and hence $\hat{\lambda} \bydef \sum_{i=0}^{t-1} p_i^{-1}$ is a \emph{strictly} 
unbiased estimator for $\lambda$. 
Cohen~\cite{Cohen15} and Ting~\cite{Ting14} observed that $\hat{\lambda}$ could simply be stored as a running sum, without having to explicitly 
store the entire history.  
In particular,  let $S^{(i)}$ be the state of the \Dartboard{} after seeing
the prefix $(a_1,\ldots,a_i)$, 
let $\lambda^{(i)}$ be the true cardinality of the prefix,
and let $\hat{\lambda}^{(i)}$ be the 
cardinality estimate. 
Define $\Indicator{\mathcal{E}}\in\{0,1\}$ to 
be the indicator variable for event $\mathcal{E}$.
Then
\[
\hat{\lambda}^{(i+1)} = \hat{\lambda}^{(i)} + 
\Indicator{S^{(i)}\neq S^{(i+1)}}\cdot \left(\Pr(S^{(i+1)}\neq S^{(i)} \;\:\middle|\;\: a_{i+1}\not\in \{a_1,\ldots,a_{i+1}\}\right)^{-1}.
\]
Letting $\lambda^{(0)}=\hat{\lambda}^{(0)}=0$, it follows that 
$\E(\hat{\lambda}^{(i)}) = \lambda^{(i)}$ and that
$\left(\hat{\lambda}^{(i)} - \lambda^{(i)}\right)_{i\geq 0}$ is 
a martingale.  To see this, note that if $a_{i+1}\in\{a_1,\ldots,a_i\}$ then by the properties of a 
\Dartboard{} sketch, $S^{(i)}=S^{(i+1)}$ and hence $\hat{\lambda}^{(i+1)}=\hat{\lambda}^{(i)}$.
If $a_{i+1}\not\in\{a_1,\ldots,a_i\}$, then $\lambda^{(i+1)}=\lambda^{(i)}+1$ and $\E\left(\hat{\lambda}^{(i+1)} \;\,\middle|\,\; \hat{\lambda}^{(i)}\right) = \hat{\lambda}^{(i)}+1$, 
hence 
$\E\left(\hat{\lambda}^{(i)}\right) = \lambda^{(i)}$.

If \textsf{X} is a \Dartboard{} sketch, call 
\Martingale{} \textsf{X} the transformed sketch 
$(S,\hat{\lambda})$
that stores the \Dartboard{} state $S$ 
and an explicit running cardinality 
estimate $\hat{\lambda}$.
Cohen~\cite{Cohen15} and Ting~\cite{Ting14} showed 
that \Martingale-transformed sketches have essentially
the same space bound but lower variances.
For example, whereas the best estimator~\cite{WangP23} 
for \HyperLogLog~\cite{FlajoletFGM07} 
has relative variance $1.075/m$,
\Martingale{} \HyperLogLog{} has relative variance $0.69/m$.
Pettie, Wang, and Yin~\cite{PettieWY21} 
gave a rigorous analysis of all \Martingale-transformed sketches, 
bounding their variances in terms of the limiting
normalized free area of the \Dartboard.  
See Figure~\ref{fig:martingale-transformed} for a comparison 
of the memory-variance products of canonical sketches
and their \Martingale-transformed versions.

\begin{table}
\centering
\begin{tabular}{|l|l|l|}
\multicolumn{1}{l}{} & \multicolumn{2}{c}{\bf Memory-Variance Product (\MVP)}\\
\multicolumn{1}{l}{\bf Sketch} &
\multicolumn{1}{l}{\bf Original} &
\multicolumn{1}{l}{\bf \Martingale-transformed}\\\hline
    (\Martingale) \Bottom-$k$~\cite{Cohen97,Giroire09,ChassaingG06}              &  \quad $\log U$      &   \quad $\frac{1}{2}\log U$\\
    (\Martingale) \PCSA~\cite{FlajoletM85,Lang17,WangP23}\quad\  &  \quad $0.45\log U$  &   \quad $0.35\log U$\\
    (\Martingale) \HyperLogLog~\cite{FlajoletFGM07,WangP23}           &    \quad $1.075\log\log U$    &   \quad $0.69\log\log U$\\
    (\Martingale) \fishmonger{} (Section~\ref{sect:fishmonger}; \cite{PettieWY21})       &  \quad $H_0/I_0\approx 1.98$ &   \quad $\frac{1}{2}H_0\approx 1.63$\\\hline
\end{tabular}
\caption{\label{fig:martingale-transformed}
All the original sketches are mergeable whereas all the \Martingale-transformed sketches are non-mergeable.  The variances of \Martingale-transformed sketches were analyzed in~\cite{Cohen15,Ting14,PettieWY21}.  
For example, the \MVP{} of \Bottom-$k$ is its space, 
$k\log U$ bits, times its normalized variance, $1/k$.  
The \MVP{} of \PCSA{} and \HyperLogLog{} 
are calculated using their explicit space bounds of 
$m\log U$ and $m\log\log U$ and 
the best explicit $\tau$-\GRA{} estimators
of Wang and Pettie~\cite{WangP23}.}
\end{table}

\medskip 

To recapitulate, sketches that violate the axioms of the \DartboardModel{} either accept a variance independent of the number of subsketches $m$~\cite{XiaoCZL20,Sedgewick,JansonLS25}, 
or space that depends on the length of the stream, including duplicates~\cite{ChakrabortyVM22}.  
Sketches that \emph{augment} a valid \Dartboard{} 
sketch with additional information as in done in \Recordinality~\cite{HelmiLMV12} or \Martingale-transformed sketches~\cite{Cohen15,Ting14,PettieWY21}, 
may improve the variance, 
but at the cost of losing mergeability.

\subsection{Remarks on the Random Oracle Model}

We believe that there is some confusion around the utility and justification
of the \textsc{random oracle model}.  On the one hand, researchers
in the \textsc{standard model} correctly point out that explicitly 
storing a uniformly random hash function is prohibitively expensive.
If we \emph{must} account for the space of the hash function, we have little choice but to 
depend on $O(1)$-wise or $\tilde{O}(1)$-wise independence.
The analyses that start from this perspective
end up introducing \emph{large} constant factors into the space of the data structure, making direct comparisons between data structures based on their \emph{asymptotic} space (as a function of $\epsilon,\delta,U$) essentially impossible.\footnote{This may not be a mathematical consequence of restricting oneself to $\tilde{O}(1)$-wise independence, but it is nonetheless empirically true.}

Flajolet and Martin~\cite{FlajoletM85} and followup work~\cite{Flajolet90,DurandF03,FlajoletFGM07,Giroire09,Cohen15,Ting14,WangP23,JansonLS25}
regard the analysis of algorithms as a \emph{scientific} endeavor.  
A good analysis is one that has \emph{predictive value} on unseen data sets.
Along this metric, adopting the \textsc{random oracle model} to make predictions 
about the performance of algorithms implemented with weak hash functions is empirically successful.  
The space bounds are precisely stated, and the variance calculations predict the 
empirical mean-squared error within a percent or so.\footnote{The estimators~\cite{FlajoletM85,DurandF03,FlajoletFGM07} 
typically begin performing well once $\lambda=\Omega(m)$ and require a 
little correction for small cardinalities; see~\cite{HeuleNH13,Ertl17}.}  
A rebuttal from the \textsc{standard model} point of view is that
predictive accuracy on fair-weather data sets is not the strongest possible guarantee.
What if the data set is adversarially constructed in order to foil a specific family of weak
$O(1)$-wise independent hash functions?

\medskip 

Both of these perspectives are correct
on their own terms, but miss a key aspect of the environment in which sketches are deployed. 
When thousands (or more) 
data sets are collected
and individually sketched, and we want to retain the ability to merge them in 
the future in order to estimate union and intersection sizes, it must be
the case that every sketch uses the \emph{same} hash function.
The hash function is just a one-time space cost, 
so it is reasonable 
to focus on the \emph{marginal} space cost of an additional sketch.
Any system built along these lines can afford to use a high-performance 
$U^\epsilon$-wise independent hash function with $O(1)$ evaluation time; see Christiani, Pagh, and Throup~\cite{ChristianiPT15} and the references therein.
Flajolet et al.~\cite{FlajoletM85,Flajolet90,FlajoletFGM07,Lang17,Ertl17,JansonLS25}~regard the \textsc{random oracle model} as
an assumption that can be scientifically validated. 
We regard the \textsc{random oracle} as a function of sufficiently
high independence that it is indistinguishable from random, 
whose one-time space cost is \emph{amortized} across sufficiently many sketches.

\medskip 

To summarize, the \textsc{random oracle model} is useful from a 
lower bound persepctive (see \cite{IndykW03,JayramW13} and Sections~\ref{sect:fish-numbers} and \ref{sect:lowerbound})
as it focuses on information-theoretic bottlenecks unrelated to hashing.  
From an upper bound perspective, the \textsc{random oracle model} 
is a \emph{practical} assumption that facilitates the analysis
of space complexity \emph{at the margin}.

\subsection{Organization}

In Section~\ref{sect:info-theory} we review Shannon entropy, Fisher information, 
and the asymptotic efficiency of maximum likelihood estimation (MLE).

Section~\ref{sect:scale-invariance-offsetting-fish} builds up to a formal definition of $\fish$-numbers.
In Section~\ref{sect:scaleinvariance} we define a notion of \emph{base-$q$ scale-invariance} 
for a sketch, meaning its Shannon entropy and normalized Fisher information
are invariant when changing the cardinality by multiples of $q$.
Under this definition Shannon entropy and normalized Fisher information
are \emph{periodic} functions of $\log_q \lambda$.  
In Section~\ref{sect:randomoffsets} we 
define \emph{average} entropy/information and show that the average
behavior of any base-$q$ scale-invariant sketch can be realized by a
generic \emph{smoothing} mechanism.  Section~\ref{sect:fish} formally defines the 
$\fish$ number of a scale-invariant 
sketch in terms of average entropy and 
average information.  

Section~\ref{sect:fish-numbers} analyzes the $\fish$ numbers of base-$q$ generalizations
of \textsf{PCSA} and \textsf{LogLog}.
Section~\ref{sect:lowerbound} characterizes the class of mergeable \Dartboard{} sketches, 
defines the subclass of \emph{linearizable} sketches,
and proves that no linearizable sketch 
has $\fish$-number smaller than $H_0/I_0$.
The \fishmonger{} sketch with $\fish$-number $H_0/I_0$ 
is described and analyzed in Section~\ref{sect:fishmonger}.
We conclude with 
some open problems in 
Section~\ref{sect:conclusion}.

Various missing proofs from
Sections~\ref{sect:fish-numbers} 
and \ref{sect:lowerbound}
appear in 
Appendices~\ref{sect:proofs}
and~\ref{sect:proofs_lowerbound}, respectively.

\section{Preliminaries}\label{sect:info-theory}
\subsection{Shannon Entropy}\label{sect:Shannon}

Let $X_1$ be a random variable with probability density/mass function $f$.  
The \emph{entropy} of $X_1$ is defined to be
\[
H(X_1) = \E(-\log_2 f(X_1)).
\]
Let $(X_1,R_1)$ be a pair of random variables with joint probability function $f(x_1,r_1)$.  
When $X_1$ and $R_1$ are independent, entropy is additive: $H(X_1,R_1) = H(X_1) + H(R_1)$.  We can generalize this to possibly dependent random variables by the \emph{chain rule for entropy}. We first define the notion of \emph{conditional entropy}. 
The conditional entropy of $X_1$ given $R_1$ is defined as
\begin{align*}
    H(X_1\mid R_1)=\E\left(-\log_2 f(X_1\mid R_1)\right),
\end{align*}
which is interpreted as the \emph{average} entropy of $X_1$ after knowing $R_1$. 

\begin{theorem}[chain rule for entropy~\cite{CoverT06}]
Let $(X_0,X_1,\ldots,X_{m-1})$ be a tuple of random variables. 
Then $H(X_0,X_1,\ldots,X_{m-1})=\sum_{i=0}^{m-1} H(X_i\mid X_{0},\ldots,X_{i-1})$.
\end{theorem}

Shannon's source coding theorem says that it is impossible to encode the outcome of 
a \emph{discrete} random variable $X_1$ in fewer than $H(X_1)$ bits on average.  
On the positive side, it is possible~\cite{CoverT06}
to assign code words such that the outcome $[X_1=x]$ 
is communicated with less than $\ceil{\log_2(1/f(x))}$ bits,
e.g., using arithmetic coding~\cite{WittenNC87,MoffatNW98}.

\subsection{Fisher Information and the Cram\'{e}r-Rao Lower Bound}\label{sect:Fisher}

Let $F=\{f_\lambda \mid \lambda\in\mathbb{R}\}$ be a family of distributions parameterized by a 
single unknown parameter $\lambda\in\mathbb{R}$.  (We do not assume
there is a prior distribution on $\lambda$.)
A \emph{point estimator} $\hat{\lambda}(X)$ is a statistic that 
estimates $\lambda$ from 
a vector $\mathbf{X}=(X_0,\ldots,X_{m-1})$ of samples drawn i.i.d.~from $f_\lambda$.

The accuracy of a ``reasonable'' point estimator is limited by the properties of the distribution family $F$ itself.
Informally, if every $f_\lambda\in F$ is sharply concentrated and statistically far from other $f_{\lambda'}$ then $f_\lambda$ is \emph{informative}.  
Conversely, if $f_\lambda$ is poorly concentrated
and statistically close to other $f_{\lambda'}$ 
then $f_\lambda$ is uninformative.  
This measure is formalized by the \emph{Fisher information}~\cite{Vaart98,CasellaB02}.

Fix $\lambda=\lambda_0$ and let $X \sim f_\lambda$ be a sample drawn from $f_\lambda$.  
The Fisher information number with respect to the observation $X$ at $\lambda_0$ is defined to be:\footnote{Since in this paper the parameter is always the cardinality, the parameter $\lambda$ is omitted in the notation $I_X(\lambda_0)$. }
\[
I_X(\lambda_0)=\E\left(\frac{\frac{\partial}{\partial\lambda}f_\lambda(X)}{f_\lambda(X)}\right)^2\mid_{\lambda=\lambda_0}.
\]
The \emph{conditional Fisher information} of $X_1$ given $X_0$ at $\lambda=\lambda_0$ is defined as
\begin{align*}
    I_{X_1\mid X_0}(\lambda_0)=\E\left(\frac{\frac{\partial}{\partial\lambda}f_\lambda(X_1\mid X_0)}{f_\lambda(X_1\mid X_0)}\right)^2\mid_{\lambda=\lambda_0}.
\end{align*}

Similar to Shannon's entropy, we also have a chain rule for Fisher information numbers.
\begin{theorem}[chain rule for Fisher information~\cite{zegers2015fisher}]
Let $\mathbf{X}=(X_0,X_1,\ldots,X_{m-1})$ be a tuple of random variables all depending on $\lambda$. 
Under mild regularity conditions, 
$I_{\mathbf{X}}(\lambda)=\sum_{i=0}^{m-1} I_{X_i\mid X_{0},\ldots,X_{i-1}}(\lambda)$.
Specifically
if $\mathbf{X}=(X_0,\ldots,X_{m-1})$ is a set of \emph{independent} 
samples from $f_\lambda$ then 
$I_{\mathbf{X}}(\lambda)= m\cdot I_{X_0}(\lambda)$.
\end{theorem}

The celebrated Cram\'{e}r-Rao lower bound~\cite{Vaart98,CasellaB02} 
states that, under mild regularity conditions (see Section \ref{sect:poisson_regular}), 
for any unbiased estimator $\hat{\lambda}(\mathbf{X})$ 
with finite variance,
\begin{align*}
    \Var\left(\hat{\lambda} \;\middle|\; \lambda\right) \geq \frac{1}{I_{\mathbf{X}}(\lambda)}.
\end{align*}
Suppose now that $\hat{\lambda}(\mathbf{X}=(X_0,\ldots,X_{m-1}))$ is, in fact, 
the Maximum Likelihood Estimator (MLE) 
from $m$ i.i.d.~observations. Under mild regularity conditions, it is asymptotically normal and 
efficient~\cite{Vaart98,CasellaB02}, i.e.,
\[
\lim_{m\rightarrow \infty} m\cdot \Var\left(\hat{\lambda} \;\middle|\; \lambda\right) = \frac{1}{I_{X_0}(\lambda)}
\qquad \text{and} \qquad
\lim_{m\rightarrow \infty}
    \sqrt{m}(\hat{\lambda}-\lambda) \sim \Normal\left(0,\frac{1}{I_{X_0}(\lambda)}\right),
\]
or equivalently, $\hat{\lambda} \sim \Normal\left(\lambda,\frac{1}{I_{\mathbf{X}}(\lambda)}\right)$ as $m\to \infty$.
In the Cardinality Estimation problem we are concerned 
with \emph{relative} variance and \emph{relative} standard deviations (standard error).  Thus, the corresponding lower bound on the relative variance is 
$\left(\lambda^2\cdot I_{\mathbf{X}}(\lambda)\right)^{-1}$. We define the \emph{normalized} Fisher information number of $\lambda$ with respect to the observation $\mathbf{X}$ to be $\lambda^2\cdot I_{\mathbf{X}}(\lambda)$.

\subsection{Regularity Conditions and Poissonization}\label{sect:poisson_regular}

The asymptotic normality of MLE and the Cram\'{e}r-Rao
lower bound depend on various regularity conditions~\cite{zegers2015fisher,abramovich2013statistical,bickel2001mathematical}, 
e.g., that $f_\lambda(x)$ is differentiable with respect to $\lambda$
and that we can swap the operators of 
differentiation w.r.t.~$\lambda$
and integration over observations $x$.  
(We only consider discrete observations here, so this is just a summation.)

A critical regularity condition of Cram\'{e}r-Rao is that
\emph{the support of $f_\lambda$ 
does not depend on $\lambda$}, i.e., 
the set of possible observations is independent of 
$\lambda$. 
A canonical example violating this condition (and one in which the Cram\'{e}r-Rao bound can actually 
be beaten)
is when $\theta$ is the parameter and the observation $X$ is sampled uniformly from $[0,\theta]$; see \cite{CasellaB02}.
Strictly speaking, the sketches we consider 
do not satisfy this property.
For example, 
when $\lambda=1$ the 
only possible $\PCSA$ 
sketches have exactly one occupied cell.
To address this issue 
we \emph{Poissonize} the model
(as in~\cite{DurandF03,FlajoletFGM07}),
which solves this problem and simplifies
other aspects of the analysis.
Consider the following two processes.
\begin{description}
    \item[Discrete counting process.] 
    Starting from time $0$, an element is inserted at every time $k\in\mathbb{N}$. 
    \item[Poissonized counting process.] Starting from time $0$, elements are inserted
    memorylessly with rate 1. This corresponds to a \emph{Poisson point process} of rate 1 on $[0,\infty)$.
\end{description}
For both processes, our goal would be to estimate the current time $\lambda$.
In the discrete process the number of insertions is precisely $\floor{\lambda}+1$
whereas in the Poisson one it is $\tilde{\lambda}\sim\operatorname{Poisson}(\lambda)$.
When $\lambda$ is sufficiently large, any estimator for $\tilde{\lambda}$ 
with standard error $c/\sqrt{m}$ also estimates 
$\lambda$ with standard error 
$(1+o(1))c/\sqrt{m}$, since
$\tilde{\lambda} = \lambda \pm \tilde{O}(\sqrt{\lambda})$ 
with probability $1-1/\poly(\lambda)$.
Since we are concerned with the asymptotic efficiency of sketches, we are
indifferent between these two models.
Let us summarize the specific advantages of the 
Poissonized counting process.
\begin{itemize}
    \item The ``time'' parameter $\lambda$ is continuous, and the distribution of the sketch for any two $\lambda_0\neq \lambda_1$ are distinct. All quantities of interest are differentiable with respect to $\lambda$.
    \item Define $Z_0,\ldots,Z_{|\mathcal{C}|-1}$
    to be the indicator variables for
    the $|\mathcal{C}|$ cells having been hit by a dart.
    These variables are mutually independent.
    \item When $\lambda>0$, 
    \emph{every} outcome of 
    $(Z_0,\ldots,Z_{|\mathcal{C}|-1})$
    has positive probability.
\end{itemize}

Algorithmically, the 
Poisson model could be simulated 
online in various ways.  
When an element $a$
arrives, we could use the random oracle to 
generate $\xi_{a} \sim \operatorname{Poisson}(1)$
and then insert elements $(a,1),\ldots,(a,\xi_{a})$ 
into the sketch as usual. 
This is equivalent to mapping $h(a)$ to a 
bit-string $b_a\in \{0,1\}^{\mathcal{C}}$,
where $\Pr(b_a(c_i)=0)=e^{-p_i}$, 
$p_i$ is the area of cell $c_i$, and $b_a(c_i)=1$ indicates that $c_i$ was hit by one of $a$'s darts.

\section{Scale-Invariance and $\fish$ Numbers}\label{sect:scale-invariance-offsetting-fish}

We are destined to measure the efficiency of observations
in terms of entropy ($H$) and normalized information ($\lambda^2 \times I$),
but it turns out that these quantities are slightly ill-defined, 
being \emph{periodic} when we really want them to be 
\emph{constant},
at least in the limit.

In Section~\ref{sect:IDF} we 
switch from the functional view of mergeable sketches 
(as commutative, idempotent transition functions) 
to a distributional interpretation.
In Section~\ref{sect:scaleinvariance} 
we define a weak notion of \emph{scale-invariance} for sketches
and
in Section~\ref{sect:randomoffsets} we give a generic method
to iron out periodic behavior 
in scale-invariant sketches,
Section~\ref{sect:fish} 
formally defines the $\fish$ number
of a sketch.

\subsection{Induced Distribution Family of Sketches}\label{sect:IDF}

Given a sketch scheme, 
Cardinality Estimation can be viewed as a point estimation problem,
where the unknown parameter is the cardinality $\lambda$
and $f_\lambda$ is the distribution over the final state of the sketch.

\begin{definition}[Induced Distribution Family]
Let $A$ be the name of a sketch having a countable state space 
$\mathcal{M}$.
The \emph{Induced Distribution Family (IDF)} of $A$ is a parameterized 
distribution family 
\begin{align*}
    \Psi_A = \{\psi_{A,\lambda}: \mathcal{M}\to [0,1]\mid \lambda>0\},
\end{align*}
where $\psi_{A,\lambda}(x)$ is the probability of $A$ being in state
$x$ at cardinality $\lambda$.
Define $X_{A,\lambda} \sim \psi_{A,\lambda}$ to be a random state drawn from $\psi_{A,\lambda}$.
\end{definition}

We can now directly characterize existing sketches as 
induced distribution families.
For example, the state space of 
a single \textsf{LogLog} ($2$-$\LL$) sketch~\cite{DurandF03}\footnote{In 
any real implementation it would
be truncated at some finite maximum value, typically 64.}
is $\mathcal{M} = \mathbb{N}$
and $\Psi_{\LL}$ contains, for each $\lambda>0$, the function $\psi_{\LL,\lambda}$:\footnote{Recall the cells have size $\{2^{-k} \mid k\geq 1\}$. The number of darts hitting a cell of size $p$ is $\Poisson(p\lambda)$, so the probability
that cell $k$ is hit and cells $k+1,k+2,\ldots$ are not
is $(1-e^{-\lambda/2^k})e^{-\lambda/2^k} = e^{\lambda/2^k}-e^{\lambda/2^{k-1}}$.
Without Poissonization the probability would be $\psi_{\LL,\lambda}(k) = 
(1-\frac{1}{2^{k}})^\lambda-(1-\frac{1}{2^{k-1}})^\lambda$.}
\begin{align*}
    \psi_{\LL,\lambda}(0) = e^{-\lambda} \qquad \text{ and for $k\geq 1$, } \quad
    \psi_{\LL,\lambda}(k) = e^{-\frac{\lambda}{2^{k}}}-e^{-\frac{\lambda}{2^{k-1}}}.
\end{align*}

We usually consider just the basic version of each 
sketch, e.g., a single bit-vector for $\PCSA$ or a single counter for $\LL$.
When we apply the machinery laid out in Section~\ref{sect:info-theory}
we take $m$ \emph{independent} 
copies of the basic sketch.

\subsection{Weak Scale-Invariance}\label{sect:scaleinvariance}

Consider a basic sketch $A$ with induced distribution family $\Psi_A$, and let
$A^m$ denote a vector of $m$ independent $A$-sketches.\footnote{Note that because of the Poissonization assumption, the distribution of a vector of, say, $m$ independent basic (\Hyper)\LogLog{} sketches at cardinality $\lambda$ is identical to a (\Hyper)\LogLog{} sketch composed of $m$ subsketches at cardinality $m\lambda$.  We choose the former formalism for notational convenience.}
From the Cram\'{e}r-Rao lower bound we know the variance of an unbiased estimator is at least $\frac{1}{I_{A^m}(\lambda)} = \frac{1}{m\cdot I_A(\lambda)}$. 
(Here $I_{A^m}(\lambda)$ is short for $I_{X_{A^m,\lambda}}(\lambda)$, 
where $X_{A^m,\lambda}$
is the observed final state of $A^m$ at time $\lambda$.)
The memory required to store it 
is at least $H(X_{A^m,\lambda}) = m\cdot H(X_{A,\lambda})$.
Thus the product of the memory 
and the \emph{relative} variance is lower bounded by 
\begin{align*}
    \frac{H(X_{A,\lambda})}{\lambda^2\cdot I_{A}(\lambda)},
\end{align*}
which only depends on the distribution family $\Psi_A$ and the 
unknown parameter $\lambda$.  However, 
ideally it would depend only on $\Psi_A$.

Essentially every existing sketch is insensitive to the scale
of $\lambda$, up to some coarse approximation.
However, it is difficult to design a sketch with 
a countable state space that is \underline{\emph{strictly}} 
scale-invariant. 
It turns out that a weaker form is just as good for our 
purposes.

\begin{definition}[Weak Scale-Invariance]\label{def:scaleinvariant}
Let $A$ be a sketch with induced distribution family
$\Psi_A$ and $q>1$ be a real number. 
We say $A$ is \emph{weakly scale-invariant with base $q$} 
if for any $\lambda>0$, we have
\begin{align*}
                 H(X_{A,\lambda}) &= H(X_{A,q\lambda}),\\
                 I_{A}(\lambda) &= q^2\cdot I_{A}(q\lambda).
\end{align*}
\end{definition}

\begin{rem}
For example, the original
(\textsf{Hyper})\textsf{LogLog} 
and $\PCSA$ sketches~\cite{FlajoletM85,FlajoletFGM07,DurandF03} 
are, after Poissonization,
base-2 weakly scale-invariant in the limit, as $\lambda\to\infty$.  For small values of $\lambda$ they
run into the problem that cells have height $2^{-k}$ 
for $k\in \mathbb{N}$.  In Section~\ref{sect:fish-numbers} we extend the domain from $\mathbb{N}$ to $\mathbb{Z}$ to achieve 
strict weak scale invariance.
\end{rem}

Observe that if a sketch $A$ is weakly scale-invariant with base $q$, 
then the ratio 
\begin{align*}
    \frac{H(X_{A,q\lambda})}{(q\lambda)^2\cdot I_{A}(q\lambda)}=\frac{H(X_{A,\lambda})}{\lambda^2\cdot I_{A}(\lambda)}
\end{align*}
becomes multiplicatively periodic with period $q$.  See Figure~\ref{fig:h_i} for illustrations
of the periodicity of the entropy ($H$) and normalized information ($\lambda^2 I$) of
the base-$q$ \textsf{LogLog} sketch.

\begin{figure}[h!]
    \centering
    \includegraphics[width=0.49\linewidth]{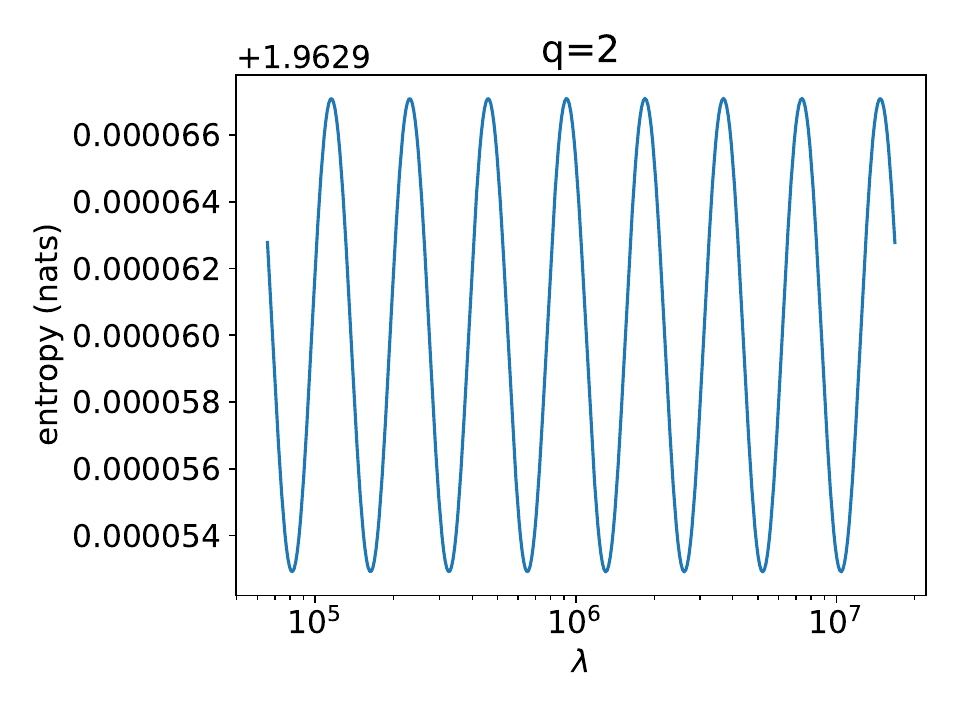}
    \includegraphics[width=0.49\linewidth]{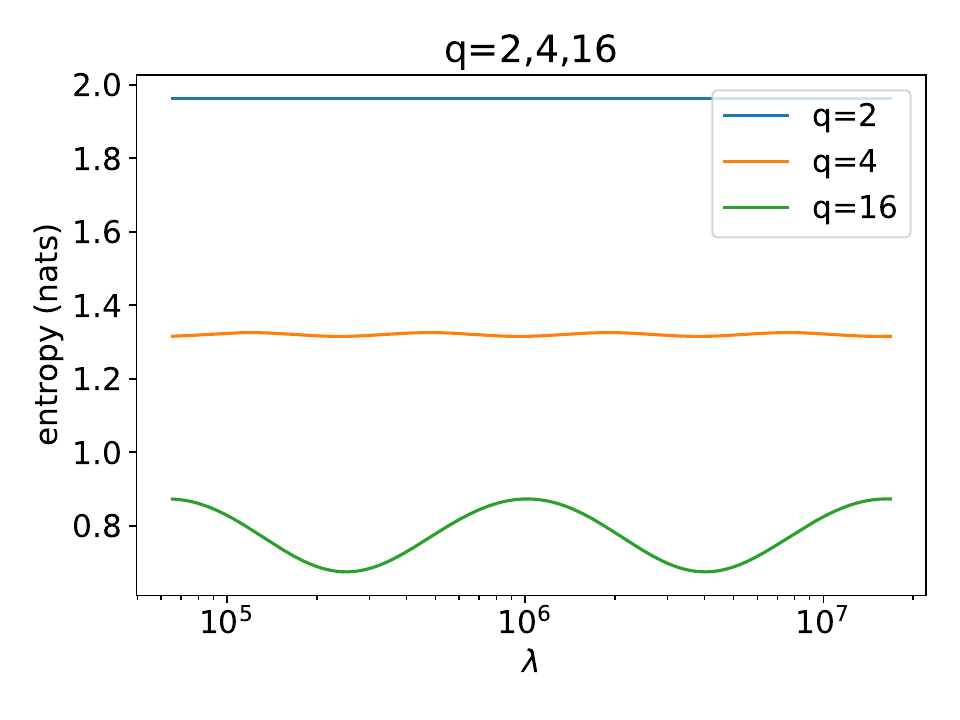}
    \includegraphics[width=0.49\linewidth]{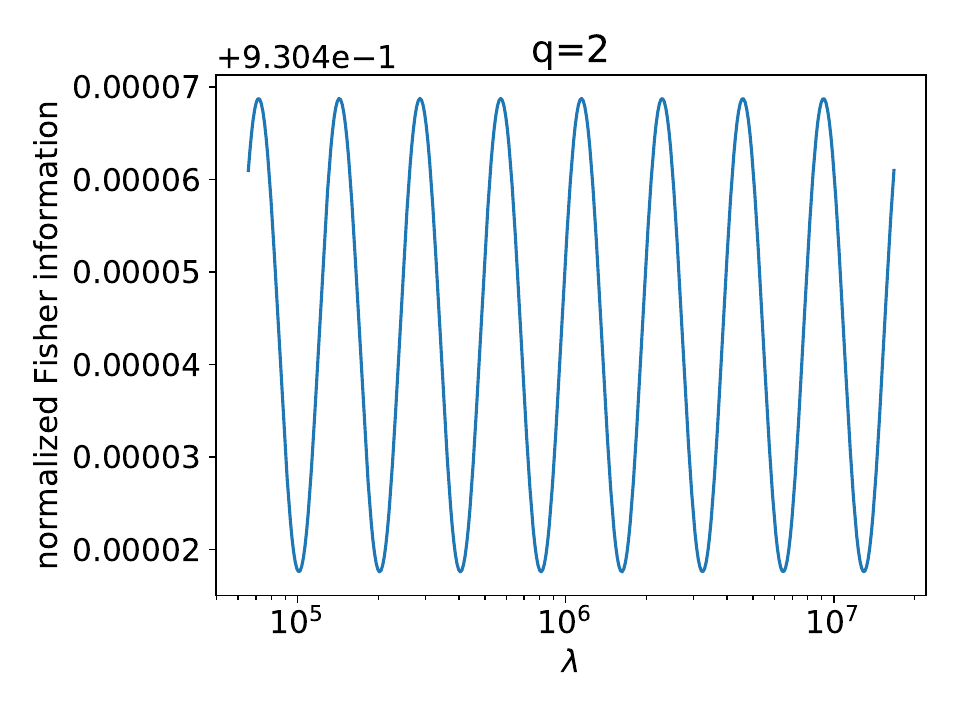}
    \includegraphics[width=0.49\linewidth]{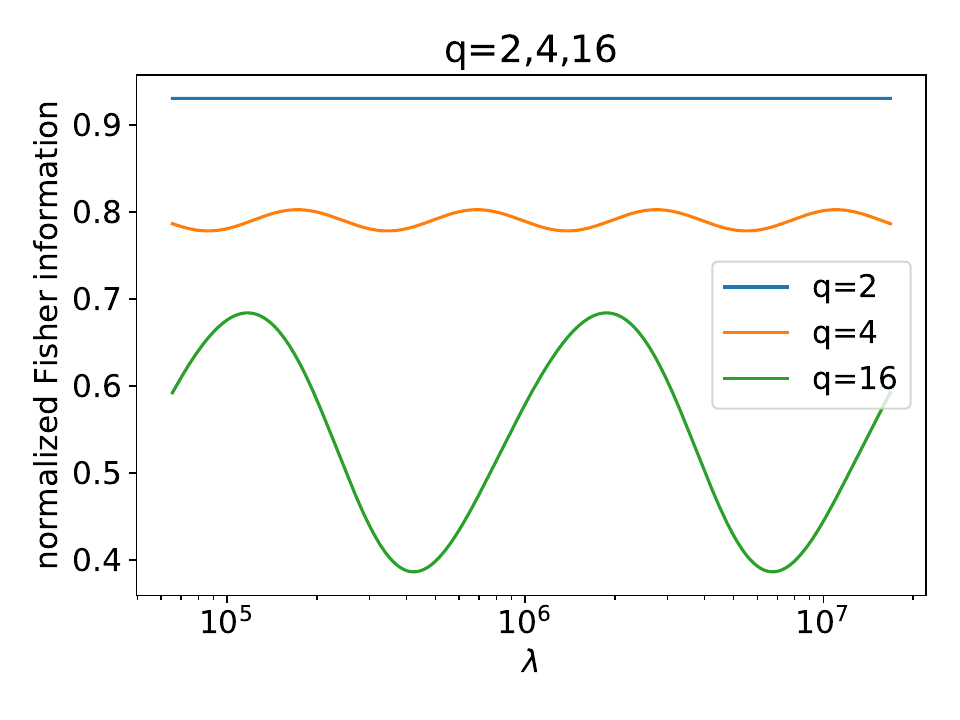}
    \caption{\small Entropy and normalized Fisher information number for $q$-\textsf{LogLog} skecthes for $\lambda\in[2^{16},2^{24}]$.
    See Section \ref{sect:qll_def} for the precise definitions. 
    Left: At a sufficiently small scale, the oscillations
    in entropy (top) and normalized information (bottom) of 2-$\LL$
    become visible.  Right: At higher values of $q \in \{2,4,16\}$,
    the oscillations in entropy (top) and normalized information (bottom) of 
    $\qLL$ are clearly visible.  
    \label{fig:h_i}}
\end{figure}

\medskip

\subsection{Smoothing via Random Offsetting}\label{sect:randomoffsets}

The \textsf{LogLog} sketch has an oscillating 
asymptotic relative variance but since its magnitude is very small
(less than $10^{-4}$~\cite{DurandF03,FlajoletFGM07}), 
it is often ignored.  
However, when we consider base-$q$ generalizations 
of \textsf{LogLog}, e.g., $q=16$, 
the oscillation becomes too large to ignore; see Figures~\ref{fig:h_i} and \ref{fig:oscillation}. 
Here we give a simple mechanism to \emph{smooth} these functions.

Rather than combine $m$ i.i.d.~copies of the basic sketch, we will
combine $m$ \emph{randomly offsetted} copies of the sketch. 
Specifically, the algorithm is hard-coded with a random vector $(R_0,\ldots,R_{m-1}) \in [0,1)^m$ 
and for all $i\in[m]$, the $i$th 
sketch is scaled down by a $q^{-R_i}$ factor.
Thus, after seeing $\lambda$ distinct elements, 
the $i$th sketch will have seen 
$\lambda q^{-R_i}$ distinct elements in expectation.
All sketches are distributed identically, 
and in the limit $m\to\infty$,
the memory size (entropy) and the relative variance tend to their average values.\footnote{Using the set of uniform offsets $(0,\frac{1}{m},\ldots,\frac{m-1}{m})$ will also work.}
Figure \ref{fig:oscillation} illustrates the effectiveness of this smoothing operation for reasonably small values of $q=16$ and $m=128$.

\begin{figure}[h!]
    \centering
    \includegraphics[width=0.49\linewidth]{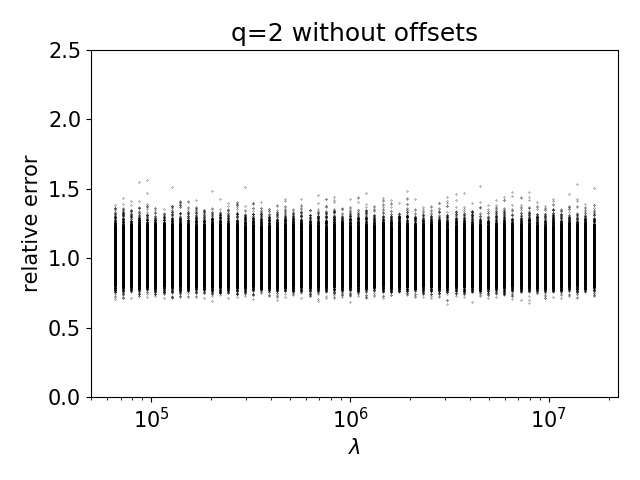}
    \includegraphics[width=0.49\linewidth]{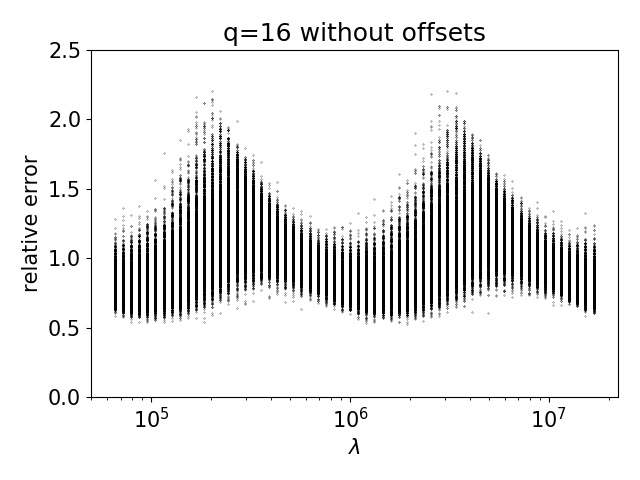}
    \includegraphics[width=0.49\linewidth]{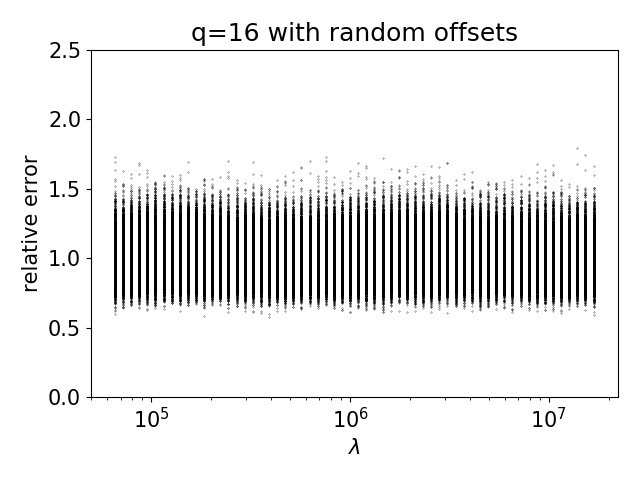}
    \includegraphics[width=0.49\linewidth]{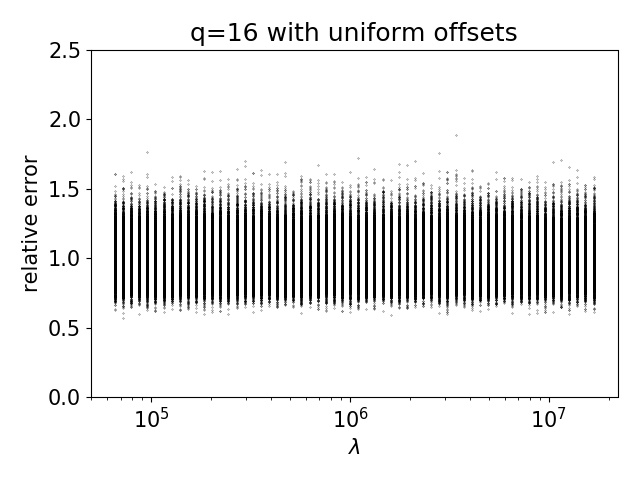}
    \caption{\small The empirical relative error ($\hat{\lambda}/\lambda$) distribution (for $\lambda\in[2^{16},2^{24}]$) of $q$-\textsf{LogLog} for four cases. \textbf{Top left:} $q=2$ without offsets.
    \textbf{Top right:} $q=16$ without offsets.
    \textbf{Bottom left:} $q=16$ with random offsets. 
    \textbf{Bottom right:} $q=16$ with uniform offsets. All use $m=128$ and the number of experiments is $5000$ for each cardinality. 
    We use a \textsf{HyperLogLog}-type estimator
    $ \hat{\lambda}(S) = \alpha_{q,m,r} \cdot m (\sum_{k\in[m]} q^{-S(k)-r_k})^{-1}$ (without stochastic averaging), 
    where  $S(k)$ is the final state of the $k$th sketch and $r_k$ is the offset for the $k$th sketch. The sketches without offsets have $r_k=0$ for all $k\in[m]$. The sketches with random offsets have $r = (r_k)_{k\in [m]}$
    uniformly distributed in $[0,1)^m$. Sketches with uniform offsets
    use the offset vector $r = (0,1/m,\ldots,(m-1)/m)$.
    The constant $\alpha_{q,m,r}$ is determined experimentally for each case.}
    \label{fig:oscillation}
\end{figure}

Throughout this section we let
$A$ be a weakly scale-invariant sketch with base $q$, having
state space $\mathcal{M}$, and induced distribution family $\Psi_A$.
Let $(R_1,Y_1) \in [0,1)\times \mathcal{M}$ be a pair where $R_1$ is uniformly random in $[0,1)$,
and $Y_1$ is the state of $A$ after seeing $\lambda q^{-R_1}$ distinct insertions.
Then
\begin{align*}
    \Pr(Y_1 = y_1\mid \lambda, R_1 = r_1) = \psi_{A,\lambda q^{-r_1}}(y_1).
\end{align*}
Thus the joint density function is
\begin{align*}
    f_\lambda(r_1,y_1)=\psi_{A,\lambda q^{-r_1}}(y_1).
\end{align*}

\begin{lemma}\label{lem:smoothed-fisher}
Fix the unknown cardinality (parameter) $\lambda$. 
The Fisher information of $\lambda$ with respect to $(R_1,Y_1)$ is equal to
\begin{align*}
    \frac{1}{\lambda^2}\int_0^1 q^{2r} I_{A}(q^{r}) dr.
\end{align*}
\end{lemma}
\begin{proof}
We can calculate the Fisher information of $\lambda$ with respect to $(R_1,Y_1)$ as follows.
\begin{align}
I_{(R_1,Y_1)}(\lambda) &=\E\left(\frac{\frac{d}{d\lambda}f_\lambda(R_1,Y_1)}{f_\lambda(R_1,Y_1)}\right)^2
= \int_0^1\sum_{y_1\in\mathcal{M}}\left(\frac{\frac{d}{d\lambda }\psi_{A,\lambda q^{-r_1}}(y_1) }{\psi_{A,\lambda q^{-r_1}}(y_1)}\right)^2 \psi_{A,\lambda q^{-r_1}}(y_1) dr_1.\label{eqn:lem:smoothed-fisher-1}
\intertext{Let $r=r_1$ and $w=\lambda q^{-r}$.  
Then we have 
$$\frac{d}{d\lambda }\psi_{A,\lambda q^{-r}}(y_1) =\frac{dw}{d\lambda}\frac{d}{dw }\psi_{A,w}(y_1)=q^{-r}\frac{d}{dw }\psi_{A,w}(y_1).$$ 
Continuing, (\ref{eqn:lem:smoothed-fisher-1}) is equal to}
    &= \int_0^1q^{-2r}\sum_{y_1\in\mathcal{M}}\left(\frac{\frac{d}{dw }\psi_{A,w}(y_1) }{\psi_{A,w}(y_1)}\right)^2 \psi_{A,w}(y_1) dr\nonumber\\
    &= \int_0^1q^{-2r} I_A(w) dr
    \;\,=\; \int_0^1q^{-2r} I_A(\lambda q^{-r}) dr.\label{eqn:lem:smoothed-fisher-2}
\intertext{Let $g(x)=q^{2x} I_A(q^x)$. 
By the weak scale-invariance of $A$, 
we have $g(x+1)=g(x)$ for any $x\in\mathbb{R}$. Applying the definition of $g$, (\ref{eqn:lem:smoothed-fisher-2})
is equal to}
    &=\frac{1}{\lambda^2}\int_0^1 g(-r+\log_q \lambda) dr
    \;=\; \frac{1}{\lambda^2}\int_0^1 g(r) dr
    = \frac{1}{\lambda^2}\int_0^1 q^{2r} I_A(q^{r}) dr.\nonumber
\end{align}
\end{proof}

\begin{lemma}
Fix the unknown cardinality (parameter) $\lambda$.
The conditional entropy $H(Y_1\mid R_1)$ is equal to 
\begin{align*}
    \int_0^1 H(X_{A,q^{r}}) dr.
\end{align*}
\end{lemma}
\begin{proof}
By the definition of the conditional entropy, we have
\[
    H(Y_1\mid R_1)
    =\int_0^1 H(Y_1\mid r) dr
    = \int_0^1 H(X_{A,\lambda q^{-r}}) dr.
\]
Let $g(x)=H(X_{A,q^{x}})$. By the weak scale-invariance of $A$, 
we know that $g(x)=g(x+1)$ for any $x\in\mathbb{R}$. Thus, we conclude that
\begin{align*}
\int_0^1 H(X_{A,\lambda q^{-r}}) dr
&= \int_0^1 g(-r+\log_q\lambda) dr
= \int_0^1 g(r) d r
= \int_0^1 H(X_{A,q^{r}}) dr.
\end{align*}
\end{proof}

In conclusion, with random offsetting we can transform any weakly 
scale-invariant sketch $A$ 
so that the product of the memory and the relative variance is
\[
    \frac{\int_0^1 H(X_{A,q^{r}}) dr}{\lambda^2\cdot \frac{1}{\lambda^2}\int_0^1 q^{2r} I_{A}(q^{r}) dr}
    =
    \frac{\int_0^1 H(X_{A,q^{r}}) dr}{\int_0^1 q^{2r} I_{A}(q^{r}) dr},
\]
and hence independent of the cardinality $\lambda$.

\subsection{The $\fish$ Number of a Sketch}\label{sect:fish}

Let $A_q$ be a weakly scale-invariant sketch with base $q$.
The \emph{Fisher-Shannon} ($\fish$) \emph{number} of $A_q$ captures 
the maximum performance we can potentially 
extract out of $A_q$, after applying
random offsets (Section~\ref{sect:randomoffsets}), optimal estimators (Section~\ref{sect:Fisher}), and compression to the entropy bound (Section~\ref{sect:Shannon}), as $m\rightarrow \infty$. In particular, any sketch composed of independent copies of $A_q$ with standard
error $\frac{1}{\sqrt{b}}$ 
must use at least $\fish(A_q)\cdot b$ bits. 
Thus, smaller
$\fish$-numbers are better.

\begin{definition}\label{def:fish}
Let $A_q$ be a weakly scale-invariant sketch with base $q$. 
The \emph{$\fish$ number} of $A_q$ is defined to be
$\fish(A_q) \bydef \mathcal{H}(A_q)/\mathcal{I}(A_q)$,
where
\[ \mathcal{H}(A_q)\bydef\int_0^1 H(X_{A_q,q^{r}}) dr
\text{ \ and \ } 
\mathcal{I}(A_q)\bydef \int_0^1 q^{2r} I_{A_q}(q^{r}) dr.
\]
\end{definition}

\section{$\fish$ Numbers of $\PCSA$ and $\LL$}\label{sect:fish-numbers}

In this section, we will find the $\fish$ numbers of 
base-$q$ 
generalizations of $\PCSA$~\cite{FlajoletM85}
and (\textsf{Hyper})\textsf{LogLog}~\cite{DurandF03,FlajoletFGM07}.
The results are expressed in terms of two important constants, $H_0$ and $I_0$.

\begin{definition}\label{def:H_0I_0}
Let $h(x)=-x\ln x-(1-x)\ln (1-x)$ and $g(x)=\frac{x^2}{e^x-1}$. We define
\begin{align*}
    H_0\bydef \frac{1}{\ln 2} \cdot \int_{-\infty}^{\infty}  h\left(e^{-e^{w}}\right)dw
    \quad\text{ and }\quad 
    I_0\bydef \int_{-\infty}^\infty   g\left(e^{w}\right) dw.
\end{align*}
\end{definition}

Lemma~\ref{lem:H0I0} derives simplified expressions for $H_0$ and $I_0$.
All missing proofs from this section 
appear in clearly marked subsections of 
Appendix~\ref{sect:proofs}.

\begin{lemma}\label{lem:H0I0}
\begin{align*}
    H_0 = \frac{1}{\ln 2} + \sum_{k=1}^\infty\frac{1}{k}\log_2 \left(1+1/k\right),
    \quad\text{ and }\quad
    I_0 = \zeta(2) = \frac{\pi^2}{6},
\end{align*}
where $\zeta(s)=\sum_{n=1}^\infty \frac{1}{n^s}$ is the Riemann zeta function.
\end{lemma}

\subsection{The $\fish$ Numbers of $\qPCSA$ Sketches}

In the discrete counting process, a natural base-$q$ generalization of $\PCSA$ ($\qPCSA$) maintains a bit vector 
$\mathbf{b}=(b_k)_{k\in\mathbb{N}}$ where $\Pr(b_i = 0) = (1-q^{-i})^\lambda \approx e^{-\lambda/q_i}$ 
after processing a multiset with cardinality $\lambda$. 
The easiest way to effect this, conceptually, is to interpret $h(a)$ as
a sequence $\mathbf{x}\in\{0,1\}^\infty$ of bits,\footnote{If we are interested
in cardinalities $\ll U$, we would truncate the hash at $\log U$ bits.} 
then update $\mathbf{b}\leftarrow \mathbf{b}\vee \mathbf{x}$, where $\vee$ is bit-wise OR.
Before Poissonization $\Pr(x_i=1)=q^{-i}$
and $\mathbf{b}$ has weight 1, while after 
Poissonization $\Pr(x_i=1)=1-e^{-q^{-i}}$ 
and $\mathbf{x}$ has weight distributed as $\Poisson(1)$.
Thus, after Poissonization and at time $\lambda$,
\begin{enumerate}
    \item The probability that the $i$th bit of $\mathbf{b}$ 
    is zero is \emph{exactly} 
    $\Pr(b_i = 0) = e^{-\lambda/q^i}$.
    \item All bits of the sketch $\mathbf{b}$ are independent.
\end{enumerate}
Since we are concerned with the \emph{asymptotic} behavior of the sketch when $\lambda\rightarrow \infty$ we also assume that the domain of the sketch $\mathbf{b}$ is extended 
from $\mathbb{N}$ 
to $\mathbb{Z}$, e.g., together with
    Poissonization, we have $\Pr(b_{-5}=0)=e^{-q^5\lambda}$.
The resulting abstract sketch 
is strictly weakly scale-invariant with base $q$, according to Definition~\ref{def:scaleinvariant}.

\begin{definition}[Induced distribution of $q$-$\PCSA$ Sketches]
For any base $q>1$, the state space\footnote{Strictly speaking the
state space is not countable.  However, it suffices to consider
only states with finite Hamming weight.} of $\qPCSA$
$\mathcal{M}_{\PCSA}=\{0,1\}^\mathbb{Z}$
and the induced distribution for cardinality $\lambda$ is 
\begin{align*}
    \psi_{\qPCSA,\lambda}(\mathbf{b}) = \prod_{k=-\infty}^\infty e^{-\frac{\lambda(1-b_k)}{q^k}} (1-e^{-\frac{\lambda}{q^k}})^{b_k}.
\end{align*}
\end{definition}

\begin{theorem}\label{thm:qPCSA}
For any $q>1$, $\qPCSA$ is weakly scale-invariant with base $q$. Furthermore, we have 
\begin{align*}
    \mathcal{H}(\qPCSA)=\frac{H_0}{\ln q}\quad \text{ and }\quad
    \mathcal{I}(\qPCSA)=\frac{I_0}{\ln q}
\quad\text{ \ \ and hence \ \ }\quad
    \fish(\qPCSA) = \frac{H_0}{I_0}
    \approx 1.98016.
\end{align*}
\end{theorem}
\begin{proof}
Let $\lambda$ be the unknown cardinality (the parameter) and 
$X_{\qPCSA,\lambda}=(Z_{\lambda,k})_{k\in\mathbb{Z}}\in\{0,1\}^{\mathbb{Z}}$ 
be the final state of the bit-vector.
For each $k$, $Z_{\lambda,k}$ is a Bernoulli random variable with 
probability mass function $f_{\lambda,k}(b_k)=e^{-\frac{\lambda(1-b_k)}{q^k}} (1-e^{-\frac{\lambda}{q^k}})^{b_k}$. 
Recall that $h(x)=-x\ln x-(1-x)\ln (1-x)$. 
Since the $\{Z_{\lambda,k}\}$ are independent, we have
\begin{align*}
    H(X_{\qPCSA,\lambda}) 
    &= \sum_{k=-\infty}^\infty H(Z_{\lambda,k})
    \;=\; \frac{1}{\ln 2}\sum_{k=-\infty}^\infty h\left(e^{-\frac{\lambda}{q^k}}\right)
    \;=\; \frac{1}{\ln 2}\sum_{k=-\infty}^\infty h\left(e^{-\frac{q\lambda}{q^k}}\right)
    \;=\; 
    H(X_{\qPCSA,q\lambda}),
\end{align*}
meaning $\qPCSA$ satisfies the first criterion of weak scale-invariance.
We now turn to the second criterion regarding Fisher information.

Let $g(x)=\frac{x^2e^{-2x}}{e^{-x}}+\frac{x^2e^{-2x}}{1-e^{-x}}
=\frac{x^2}{e^{x}-1}$. 
Observe that the Fisher information of $\lambda$ with respect to the observation $Z_{\lambda,k}$ (i.e., $I_{Z_{\lambda,k}}(\lambda)$)
is equal to
\begin{align*}
  \E\left(\frac{\frac{d}{d\lambda}f_{\lambda,k}(Z_{\lambda,k})}{f_{\lambda,k}(Z_{\lambda,k})}\right)^2=&\frac{\left(\frac{d}{d\lambda}(1-e^{-\frac{\lambda}{q^k}})\right)^2}{1-e^{-\frac{\lambda}{q^k}}}+ \frac{\left(\frac{d}{d\lambda}e^{-\frac{\lambda}{q^k}}\right)^2}{e^{-\frac{\lambda}{q^k}}}\\
  =& \frac{\left(\frac{1}{q^k}e^{-\frac{\lambda}{q^k}}\right)^2}{1-e^{-\frac{\lambda}{q^k}}}+ \frac{\left(\frac{1}{q^k}e^{-\frac{\lambda}{q^k}}\right)^2}{e^{-\frac{\lambda}{q^k}}} = \frac{1}{\lambda^2} g\left(\frac{\lambda}{q^k}\right).
\end{align*}

Since the $\{Z_{\lambda,k}\}$ are independent, we have
\begin{align*}
   I_{\qPCSA}(\lambda)
   &= \sum_{k=-\infty}^\infty \frac{1}{\lambda^2} g\left(\frac{\lambda}{q^k}\right)
   \;=\; q^2\sum_{k=-\infty}^\infty \frac{1}{q^2\lambda^2} g\left(\frac{q\lambda}{q^k}\right)
   = q^2I_{\qPCSA}(q\lambda).
\end{align*}
We conclude that $\qPCSA$ is weakly scale-invariant with base $q$.
Now we compute the $\mathcal{H}(\qPCSA)$ and $\mathcal{I}(\qPCSA)$.
\begin{align*}
    \mathcal{H}(\qPCSA) = \int_0^1 H(X_{\qPCSA,q^r})dr
    &= \frac{1}{\ln 2}\int_0^1 \sum_{k=-\infty}^\infty h\left(e^{-\frac{q^r}{q^k}}\right)dr
    = \frac{1}{\ln 2}\sum_{k=-\infty}^\infty\int_0^1  h\left(e^{-e^{(r-k)\ln q}}\right)dr\\
    &= \frac{1}{\ln 2}\sum_{k=-\infty}^\infty\int_{-k}^{1-k}  h\left(e^{-e^{r\ln q}}\right)dr
    = \frac{1}{\ln 2}\int_{-\infty}^{\infty}  h\left(e^{-e^{r\ln q}}\right)dr\\
    &= \frac{1}{\ln 2}\cdot \frac{1}{\ln q}\int_{-\infty}^{\infty}  h\left(e^{-e^{w}}\right)dw \;\,=\,\; \frac{H_0}{\ln q}.
\end{align*}
The final line uses the change of variable $w=r\ln q$.
We use similar techniques to calculate the normalized information $\mathcal{I}(\qPCSA)$.
\begin{align*}
    \mathcal{I}(\qPCSA) &= \int_0^1 q^{2r} I_{\qPCSA}(q^{r}) dr
    =\int_0^1 q^{2r} \sum_{k=-\infty}^\infty \frac{1}{q^{2r}}g(q^{r-k}) dr
    =\sum_{k=-\infty}^\infty\int_0^1   g(q^{r-k}) dr\\
    &=\int_{-\infty}^\infty   g(q^{r}) dr
    \;=\; \frac{1}{\ln q}\int_{-\infty}^\infty   g(e^{w}) dw
    \;=\; \frac{I_0}{\ln q}.
\end{align*}
\end{proof}

\subsection{The $\fish$ Numbers of $q$-\textsf{LogLog} Sketches}\label{sect:qll_def}

In a discrete counting process, the natural base-$q$ generalization of 
the (\textsf{Hyper})\textsf{LogLog} sketch ($\qLL$)
works as follows.  Let $Y = \min_{a\in\mathcal{A}} h(a) \in [0,1]$ be the minimum hash value seen.
The $\qLL$ sketch stores the integer $S = \floor{-\log_q Y}$, so when the cardinality is $\lambda$,
\[
    \Pr(S=k) 
    = \Pr(q^{-k}\leq Y <q^{-k+1})
    = (1-q^{-k})^\lambda - (1-q^{-k+1})^\lambda \approx e^{-\lambda/q^k} - e^{-\lambda/q^{k-1}}.
\]

Once again the state space of this sketch is $\mathbb{N}$ but to show weak scale-invariance
it is useful to extend it to $\mathbb{Z}$. Together with Poissonization, we have the following. 
\begin{enumerate}
    \item $\Pr(S=k)$ is \emph{precisely} $e^{-\lambda/q^k}-e^{-\lambda/q^{k-1}}$.
    \item The state space is $\mathbb{Z}$, e.g., together with (1) we have $\Pr(S=-1)=e^{-q\lambda}-e^{-q^2\lambda}$.
\end{enumerate}

\begin{definition}[Induced distribution family of $\qLL$ sketches]
For any base $q>1$, the state space of $\qLL$ is $\mathcal{M}_{\LL} = \mathbb{Z}$ 
and the induced distribution for cardinality $\lambda$ is
\begin{align*}
    \psi_{\qLL,\lambda}(k)=e^{-\lambda/q^k}-e^{-\lambda/q^{k-1}}.
\end{align*}
\end{definition} 

In Lemma~\ref{lem:HLL} we express the $\fish$ number of $\qLL$
in terms of two quantities $\phi(q)$ and $\rho(q)$, defined as follows.

\begin{definition}
\begin{align*}
    \phi(q)\bydef& \int_{-\infty}^\infty -(e^{-e^{r}}-e^{-e^{r}q})\log_2 (e^{-e^{r}}-e^{-e^{r}q})dr.\\
    \rho(q)\bydef&  \int_{-\infty}^{\infty} \frac{(-e^{r}e^{-e^{r}}+e^{r}qe^{-e^{r}q})^2}{e^{-e^{r}}-e^{-e^{r}q}}dr.
\end{align*}
\end{definition}

Lemma~\ref{lem:HLL_tool} gives simplified expressions for $\phi(q)$ and $\rho(q)$.
See Appendix~\ref{sect:proofs} for proof.

\begin{lemma}\label{lem:HLL_tool}
Let $\zeta(s,t)=\sum_{k\ge 0} (k+t)^{-s}$ be the 
Hurwitz zeta function.  Then $\phi$ and $\rho$ can be expressed as:
\begin{align*}
    \phi(q) &= \frac{1-1/q}{\ln 2} + \sum_{k=1}^\infty\frac{1}{k}\log_2 \left(\frac{k+\frac{1}{q-1}+1}{k+\frac{1}{q-1}}\right).\\
    \rho(q) &= \zeta\left(2,\frac{q}{q-1}\right)  = \sum_{k=0}^\infty\frac{1}{(k+\frac{q}{q-1})^2}.
\end{align*}
\end{lemma}

Refer to Appendix~\ref{sect:proofs} for proof
of Lemma~\ref{lem:HLL}.
\begin{lemma}\label{lem:HLL}
For any $q>1$, $\qLL$ is weakly scale-invariant with base $q$. Furthermore, we have
\begin{align*}
    \mathcal{H}(\qLL)=\frac{\phi(q)}{\ln q}\quad \text{ and }\quad
    \mathcal{I}(\qLL)= \frac{\rho(q)}{\ln q}.
\end{align*}
\end{lemma}

\begin{theorem}\label{thm:fish-qLL}
For any $q>1$, the $\fish$ number of $\qLL$ is
\begin{align*}
    \fish(\qLL) &> \frac{H_0}{I_0}.
\intertext{Furthermore, we have}
    \lim_{q\to\infty} \fish(\qLL) &= \frac{H_0}{I_0}.
\end{align*}
\end{theorem}

\begin{proof}
The first statement 
follows from Lemmas \ref{lem:HLL_qhigh} and \ref{lem:HLL_qlow}, which are stated below and proved in Appendix~\ref{sect:proofs}.
We prove the second statement. 
By Lemma \ref{lem:HLL}, we have
\begin{align*}
    \lim_{q\to\infty} \fish(\qLL) 
    &= \lim_{q\to\infty} \frac{\mathcal{H}(\qLL)}{\mathcal{I}(\qLL)}\\
    &= \lim_{q\to\infty} \frac{\displaystyle \frac{1-1/q}{\ln 2}+\sum_{k=1}^\infty\frac{1}{k}\log_2 \left(\frac{k+\frac{1}{q-1}+1}{k+\frac{1}{q-1}}\right)}{\displaystyle\sum_{k=1}^\infty\frac{1}{(k+\frac{1}{q-1})^2}}\\
    &=\frac{\displaystyle \frac{1}{\ln 2}+\sum_{k=1}^\infty\frac{1}{k}\log_2 \left(\frac{k+1}{k}\right)}{\displaystyle \sum_{k=1}^\infty\frac{1}{k^2}}
    \;=\; \frac{H_0}{I_0}.
\end{align*}
\end{proof}

Refer to Appendix~\ref{sect:proofs} for proofs of Lemmas~\ref{lem:HLL_qhigh} and \ref{lem:HLL_qlow}.

\begin{lemma}\label{lem:HLL_qhigh}
 $\fish(\qLL)$ is strictly decreasing for $q\geq 1.4$.
\end{lemma}

\begin{lemma}\label{lem:HLL_qlow}
$\fish(\qLL) > \fish(2\text{-}\LL)$ for $q\in(1,1.4]$.
\end{lemma}

\section{A Sharp Lower Bound on Linearizable Sketches}\label{sect:lowerbound}

In this section we give a formal definition of \emph{linearizable} sketches and prove a lower bound on the $\fish$-numbers of linearizable sketches.  It would be ideal if we could extend 
this lower bound to the class of all mergeable sketches, but mergeability is a trickier property to characterize and analyze.  
We make some progress in this direction and characterize mergeability in terms of the the state space
of the sketch.

Recall the set-up from Section~\ref{sect:introduction}.
The \Dartboard{} $[0,1]^2$ is partitioned
into cells $\mathcal{C}$. 
A sketching scheme
may nominally treat $\mathcal{C}$ as an infinite
set, but we can always truncate it at a large finite 
number of cells with negligible loss in efficiency.
The hash function $h : [U] \to \mathcal{C}$ chooses a random cell $c_i\in \mathcal{C}$ with probability equal to its area.
A state $S\subseteq \mathcal{C}$ is a subset of occupied cells, which are by definition the union of those hash values that, by Axiom 2, cause no state change if encountered.
The state is determined by a transition 
function $f : 2^{\mathcal{C}} \times \mathcal{C} \to 2^{\mathcal{C}}$
where $f(S,h(a))=S'$ maps the previous state $S$ and hash value of 
$a$ to the next state $S'$.
We define a sketch to be mergeable if $S(\mathcal{A})$
is insensitive to the multiplicity of elements in $\mathcal{A}$
and the order in which they are processed, that is,
if the transition $f$ is commutative and idempotent.\footnote{Given $S(\mathcal{A}),S(\mathcal{B})$ we can generate
$S(\mathcal{A},\mathcal{B})$ by hypothesizing a data set $\mathcal{B}'$ for which $h(\mathcal{B}')=S(B)$ and inserting the elements of $\mathcal{B}'$ into $S(A)$ one-by-one.  
By Axiom 2 and the commutativity of the transition $f$ we have
$S(\mathcal{B},\mathcal{B}')
=S(\mathcal{B}',\mathcal{B})
=S(\mathcal{B})=S(\mathcal{B}')$,
hence
$S(\mathcal{A},\mathcal{B}')
=S(\mathcal{B}',\mathcal{A})
=S(\mathcal{B},\mathcal{A})$,
where the order of sets indicates the order in 
which they were processed.
}

Let $\mathcal{S} \subseteq 2^{\mathcal{C}}$ be the state 
space of a \Dartboard{} 
sketch \textsf{X}.  Theorem~\ref{thm:mergeable-characterization}
gives a complete characterization of mergeability of \textsf{X} 
in terms of $\mathcal{S}$.

\begin{theorem}\label{thm:mergeable-characterization}
    A \Dartboard{} sketch \textsf{X} using finite cell partition $\mathcal{C}$ is mergeable if and only if the state space $\mathcal{S} \subseteq 2^{\mathcal{C}}$
    satisfies \emph{validity} and \emph{closure}.
\begin{description}
    \item[Validity.] For any set $D\subseteq \mathcal{C}$ of at most $U$ cells there exists a state $S\in\mathcal{S}$ such that $D\subseteq S$.
    \item[Closure.] For any two distinct states 
    $S,S'\in \mathcal{S}$, $S\cap S'\in \mathcal{S}$.
\end{description}
\end{theorem}

\begin{proof}
    Validity is necessary for any \Dartboard{} sketch, mergeable or not.  Without it, it is impossible to satisfy the axioms
    of the \Dartboard{} if there is 
    no state $S$ for which $D\subseteq S\in \mathcal{S}$. 
    Let $S_0,S_1 \in \mathcal{S}$ 
    be in the state space,
    and $S^* = S_0\cap S_1$. 
    Suppose $\mathcal{A}_0, \mathcal{A}_1, \mathcal{A}^*$ are data sets 
    for which $h(\mathcal{A}_0), h(\mathcal{A}_1), h(\mathcal{A}^*)$ 
    are equal to $S_0,S_1,S^*$, respectively.
    If, contrary to the claim,
    $S^*\not\in\mathcal{S}$ is not in the state space, then by Axiom 1,
    $S(\mathcal{A}^*)\supset S^*$,
    meaning there is a cell 
    $c\in S(\mathcal{A}^*)-S^*$.
    Suppose $c\not\in S_0$.
    By Axiom 2, 
    $S(\mathcal{A}_0, \mathcal{A}^*)=S(\mathcal{A}_0)$ does not contain $c$, while
    $S(\mathcal{A}^*,\mathcal{A}_0) \supseteq S(\mathcal{A}^*)$ must contain $c$, 
    with the sets written in the order they are processed.  This
    contradicts the mergeability of \textsf{X}
    and the commutativity of the 
    transition function.
\end{proof}

\subsection{Linearizable Sketches}\label{sect:linearizable}

Informally, a sketch in the \DartboardModel{} is 
called \emph{linearizable} if there is a fixed
permutation of its cells $(c_0,c_1,\ldots,c_{|\mathcal{C}|-1})$ 
such that if $S_0 \in\mathcal{S}$ is the state, 
whether $c_i\in S_0$ is a function 
of $S_0 \cap \{c_0,\ldots,c_{i-1}\}$ and 
whether $c_i$ has been hit by a dart.  

\begin{definition}[Linearizability]
Consider a \Dartboard{} sketch \textsf{X} over a cell partition $\mathcal{C}$.  Let $Z_i\in\{0,1\}$ be the indicator for whether cell $c_i\in \mathcal{C}$ has been hit by a dart, and let $Y_i\in\{0,1\}$
be the indicator for whether $c_i$ is occupied.
Axiom 1 implies $Z_i\leq Y_i$.
We call \textsf{X} \emph{linearizable} if there is a fixed permutation of the cells, say it is
$(c_0,\ldots,c_{|\mathcal{C}|-1})$ without loss of generality, and a monotone boolean function $\phi : \{0,1\}^*\to\{0,1\}$ such that 
\[
Y_i = Z_i \vee \phi(\mathbf{Y}_{i-1}), \qquad \text{where } \mathbf{Y}_{i-1} = (Y_0,\ldots,Y_{i-1}).
\]
\end{definition}

We call $\phi$ the \emph{forced occupation} function. 

\begin{lemma}
    The \PCSA, (\Hyper)\LogLog, and \Bottom-$k$ sketches are linearizable, while
    \AdaptiveSampling{} is not.
\end{lemma}

\begin{proof}
    The state of \PCSA{} is $\mathbf{Y}=(Z_0,\ldots,Z_{|\mathcal{C}|-1})$, so it is linearizable with respect to any permutation and $\phi=0$.  (\Hyper)\LogLog{} is linearizable when the cells are put in non-decreasing order by size.  The function $\phi(Y_0,\ldots,Y_{i-1})=1$ if 
    for some cell $c_j$ in the same column as $c_i$, $j<i$ and $Y_j=1$.
    The \Bottom-$k$ sketch is linearizable 
    where the cells/hash values appear in increasing order and $\phi(\mathbf{Y}_{i-1})=1$ iff $\mathbf{Y}_{i-1}$ has weight at least $k$.
    \AdaptiveSampling{} is very similar to \Bottom-$k$ in that it stores all hash values in the range $[0,2^{-\ell})$, where $\ell$ is minimum such that the number of hash values stored is \emph{at most} $k$.
    Because of the variable number of hash values stored, it is technically not linearizable with respect to any permutation of the cells/hash values.
\end{proof}

\begin{lemma}\label{lem:linearizable-subset-mergeable}
All linearizable sketches are mergeable.     
\end{lemma}

\begin{proof}
    Let $\mathbf{Y}^{\mathcal{A}},\mathbf{Y}^{\mathcal{B}}\in\{0,1\}^{\mathcal{C}}$ be the bit-vector representations of 
    states $S(\mathcal{A}),S(\mathcal{B})$ after processing sets $\mathcal{A},\mathcal{B}\subset [U]$.
    We prove it is mergeable by induction, where the base 
    case is 
    just a special case of the inductive case.  Suppose that $\mathbf{Y}_{i-1}^{\mathcal{A}\cup\mathcal{B}}$ has already been computed. 
    We have
    \begin{align*}
    Y_{i}^{\mathcal{A}\cup\mathcal{B}} 
    &= 
    Z_i^{\mathcal{A}\cup\mathcal{B}} \vee \phi(\mathbf{Y}_{i-1}^{\mathcal{A}\cup\mathcal{B}}) & \text{Definition of $Y_i^{\mathcal{A}\cup\mathcal{B}}$.}\\
    &= Z_i^{\mathcal{A}} \vee Z_i^{\mathcal{B}} \vee \phi(\mathbf{Y}_{i-1}^{\mathcal{A}\cup\mathcal{B}}) &  \text{Definition of $Z_i^{\mathcal{A}\cup\mathcal{B}}$.}\\
    &= Z_i^{\mathcal{A}} \vee \phi(\mathbf{Y}_{i-1}^{\mathcal{A}}) \vee Z_i^{\mathcal{B}} \vee \phi(\mathbf{Y}_{i-1}^{\mathcal{B}}) 
    \vee \phi(\mathbf{Y}_{i-1}^{\mathcal{A}\cup\mathcal{B}})
    & \text{Monotonicity of $\phi$.}\\
    &= Y_i^{\mathcal{A}} \vee Y_i^{\mathcal{B}} 
    \vee \phi(\mathbf{Y}_{i-1}^{\mathcal{A}\cup\mathcal{B}}) & \text{Definition of $Y_i^{\mathcal{A}},Y_i^{\mathcal{B}}$.}
    \end{align*}
    The third equality follows since 
    $\mathbf{Y}_{i-1}^{\mathcal{A}\cup\mathcal{B}}$ dominates both $\mathbf{Y}_{i-1}^{\mathcal{A}},\mathbf{Y}_{i-1}^{\mathcal{B}}$,
    and by the monotonicity of $\phi$,
    $\phi(\mathbf{Y}_{i-1}^{\mathcal{A}\cup\mathcal{B}})
    =\phi(\mathbf{Y}_{i-1}^{\mathcal{A}})
    \vee
    \phi(\mathbf{Y}_{i-1}^{\mathcal{B}})
    \vee
    \phi(\mathbf{Y}_{i-1}^{\mathcal{A}\cup\mathcal{B}})$.
    Therefore it is possible to compute $Y_i^{\mathcal{A}\cup\mathcal{B}}$
    from $Y_i^{\mathcal{A}},Y_i^{\mathcal{B}}$ and $\mathbf{Y}_{i-1}^{\mathcal{A}\cup\mathcal{B}}$, which concludes the induction.
    (In the base case $\mathbf{Y}_{-1}^{\mathcal{A}}=\mathbf{Y}_{-1}^{\mathbf{B}}=\mathbf{Y}_{-1}^{\mathcal{A}\cup\mathcal{B}}=\epsilon$ and $\phi(\epsilon)$ is constant.)
\end{proof}

\subsection{The Lower Bound}\label{sect:lb}

When phrased in terms of the \DartboardModel, 
our analysis of the $\fish$-number of $\PCSA$ 
(Section~\ref{sect:fish-numbers})
took the following approach.  
We fixed a moment in \emph{time} $\lambda$
and \emph{aggregated} the Shannon entropy 
and normalized Fisher information over
all \emph{cells} on the \Dartboard.

Our lower bound on linearizable sketches 
begins from the opposite point of view.  
We fix a particular cell $c_i \in \mathcal{C}$
of size $p_i$ and consider how it might contribute
to the Shannon entropy and normalized Fisher information at 
\emph{various times}.  The $\dot{H},\dot{I}$ functions defined in Lemma~\ref{lem:hdot_idot} are useful for describing these contributions.

\begin{lemma}\label{lem:hdot_idot}
Let $Z$ be the indicator variable for whether 
a particular cell of size $p$ has been hit by a dart. 
At time $\lambda$, $\Pr(Z = 0)=e^{-p\lambda}$ and
\begin{align*}
    H(Z)&=\dot{H}(p\lambda) \quad\mbox{ and }\quad \lambda^2\cdot I_Z(\lambda)=\dot{I}(p\lambda),\\
\intertext{where}
    \dot{H}(t)&\bydef \frac{1}{\ln 2}\left(t e^{-t} -(1- e^{-t})\ln( 1-e^{-t})\right),\\
    \dot{I}(t)&\bydef \frac{t^2}{e^t-1}. 
\end{align*}
\end{lemma}
In other words, the number of darts in this cell is a $\operatorname{Poisson}(t)$ random variable, $t=p\lambda$, and both entropy and normalized information can be expressed in terms of $t$ via functions $\dot{H},\dot{I}$.

\begin{proof}
$\Pr(Z = 0)=e^{-p\lambda }$ follows from the definition of the process. Then we have, by the definition of entropy and Fisher information,
\begin{align*}
    H(Z)&= -  e^{-p\lambda} \log_2  e^{-p\lambda} -(1- e^{-p\lambda})\log_2( 1-e^{-p\lambda})\\
    &=p\lambda e^{-p\lambda}/\ln2 -(1- e^{-p\lambda})\log_2( 1-e^{-p\lambda}) \;=\; \dot{H}(p\lambda),\\
    \lambda^2\cdot  I_Z(\lambda)&=\lambda^2\left(\frac{e^{-2p\lambda}p^2}{e^{-p\lambda}}+\frac{e^{-2p\lambda}p^2}{1-e^{-p\lambda}}\right)=\frac{e^{-p\lambda}(p\lambda)^2}{1-e^{-p\lambda}} \;=\; \dot{I}(p\lambda).
\end{align*}
\end{proof}

Still fixing $c_i\in \mathcal{C}$ with size $p_i$, 
let us now aggregate its \emph{potential} 
contributions to entropy/information
\emph{over all time}.  
We say \underline{potential} contribution because in a linearizable sketch, it is possible for cell $c_i$ to be ``killed'';
at the moment  $\phi(\mathbf{Y}_{i-1})$
switches from 0 to 1, $Z_i$
is no longer relevant.
We measure time on a log-scale, so $\lambda = e^x$.  
Unsurprisingly, the potential contributions of $c_i$ do not depend on $p_i$:

\begin{lemma}\label{lem:integration}
\begin{align*}
     \int_{-\infty}^\infty \dot{H}(e^x) dx=H_0 \quad \text{ and }\quad
    \int_{-\infty}^\infty \dot{I}(e^x) dx =I_0.
\end{align*}
\end{lemma}
\begin{proof}
Follows from 
Definition \ref{def:H_0I_0}
and
Lemma \ref{lem:H0I0}, 
since $\frac{1}{\ln 2}\cdot h(e^{-e^w})=\dot{H}(e^w)$ and $g(e^w)=\dot{I}(e^w)$.
\end{proof}

In other words, if we let cell $c_i$ ``live'' forever (fix $\phi(\mathbf{Y}_{i-1})=0$ for all time) 
it would contribute $H_0$ to the aggregate entropy
and $I_0$ to the aggregate normalized Fisher information.  In reality $c_i$ may die at some particular time,
which leads to a natural \emph{optimization} question.
When is the most advantageous time $\lambda$ to kill $c_i$, as a function of 
its density $t_i = p_i\lambda$?

Figure \ref{fig:hi} plots $\dot{H}(t)$, $\dot{I}(t)$ and---most importantly---the ratio 
$\dot{H}(t)/\dot{I}(t)$. 
It \emph{appears} as if $\dot{H}(t)/\dot{I}(t)$ is monotonically decreasing in $t$ and this is, in fact, the case, as established in Lemma~\ref{lem:monotonicity}. See Appendix~\ref{sect:proofs_lowerbound}.
\begin{lemma}\label{lem:monotonicity}
$\dot{H}(t)/\dot{I}(t)$ is decreasing in $t$ on $(0,\infty)$.
\end{lemma}

\begin{figure}[h!]
\centering
    \includegraphics[width=0.6\linewidth]{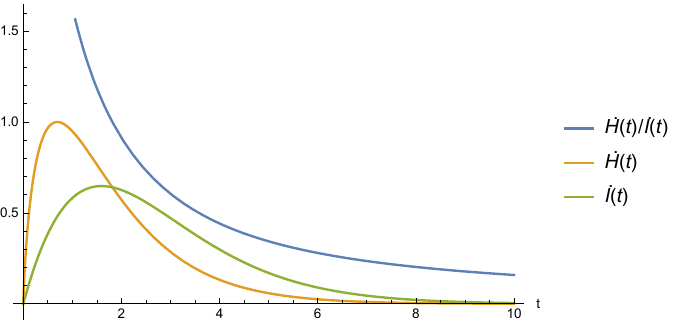}
    \caption{\label{fig:hi}$\dot{H}(t),\dot{I}(t)$ and $\dot{H}(t)/\dot{I}(t)$}
\end{figure}

Lemma~\ref{lem:monotonicity} is
the critical observation.
Although the \emph{cost} $\dot{H}(t)$ 
and \emph{value} $\dot{I}(t)$ fluctuate with $t$, the cost-per-unit-value 
\emph{only improves with time}. 
In other words, the optimum moment to ``kill'' any cell $c_i$ should be \emph{never}, 
and any linearizable sketch that routinely
kills cells prematurely should, on average, 
perform strictly worse than $\PCSA$---the ultimate pacifist sketch.

The rest of the proof formalizes this intuition.  
One difficulty is that $H_0/I_0$ is not a \emph{hard} lower bound
at any particular moment in time.  For example, 
if we just want to perform well when the cardinality $\lambda$
is in, say, $[10^6, 2\cdot 10^6]$, then we can easily beat $H_0/I_0$ by a constant factor.\footnote{Clifford and Cosma~\cite{CliffordC12} calculated
the optimal Fisher information for Bernoulli observables when $\lambda$ was known to lie in a small range.}
However, if we want to perform well over a sufficiently long 
time interval $[a,b]$, then, at best, 
the worst case efficiency over that interval tends
to $H_0/I_0$ in the limit.

\medskip 

Define $Z_{i,\lambda},Y_{i,\lambda}$ to be the variables $Z_i,Y_i$ at time $\lambda$.
Let $\mathbf{Y}=\mathbf{Y}_{|\mathcal{C}|-1} = (Y_0,\ldots,Y_{|\mathcal{C}|-1})$ be the vector of indicators encoding the state of the sketch and $\mathbf{Y}_{[\lambda]} = (Y_{0,\lambda},\ldots,Y_{|\mathcal{C}|-1,\lambda})$ 
refer to $\mathbf{Y}$ at time $\lambda$.

\begin{prop}\label{prop:phi_monotone}
For any linearizable sketch and any $c_i\in\mathcal{C}$, $\Pr(\phi(\mathbf{Y}_{i-1,\lambda})=0)$ is non-increasing with $\lambda$.
\end{prop}
\begin{proof}
Follows from Axiom 3 
and the monotonicity of $\phi$.
\end{proof}

The proof depends on \emph{linearizability} mainly
through Lemma~\ref{lem:chain_rule}, which uses the chain rule to bound aggregate entropy/information in 
terms of a weighted sum
of cell entropy/information.
The weights here correspond
to the probability that the
cell is still alive, which,
by Proposition~\ref{prop:phi_monotone}, 
is non-increasing over time.
\begin{lemma}\label{lem:chain_rule}
For any linearizable sketch and any $\lambda>0$, we have
\begin{align*}
    H(\mathbf{Y}_{[\lambda]}) &=  \sum_{i=0}^{|\mathcal{C}|-1} \ \dot{H}(p_i\lambda)\Pr(\phi(\mathbf{Y}_{i-1,\lambda})=0),\\
    \lambda^2\cdot I_{\mathbf{Y}}(\lambda) & =\sum_{i=0}^{|\mathcal{C}|-1}  \dot{I}(p_i\lambda)\Pr(\phi(\mathbf{Y}_{i-1,\lambda})=0).
\end{align*}
\end{lemma}
\begin{proof}
By the chain rule of entropy, we have
\begin{align*}
    H(\mathbf{Y}_{[\lambda]})&= \sum_{i=0}^{|\mathcal{C}|-1} H(Y_{i,\lambda}\mid \mathbf{Y}_{i-1,\lambda})= \sum_{i=0}^{|\mathcal{C}|-1} H(Z_{i,\lambda})\Pr(\phi(\mathbf{Y}_{i-1,\lambda})=0)
    = \sum_{i=1}^{|\mathcal{C}|-1} \ \dot{H}(p_i\lambda)\Pr(\phi(\mathbf{Y}_{i-1,\lambda})=0),
\end{align*}
where the last equality follows from Lemma \ref{lem:hdot_idot}.
Similarly, by the chain rule of Fisher information number, we have
\begin{align*}
    \lambda^2\cdot I_{\mathbf{Y}}(\lambda)&=\sum_{i=0}^{|\mathcal{C}|-1} \lambda^2\cdot I_{Y_{i}\mid \mathbf{Y}_{i-1}}(\lambda) = \sum_{i=0}^{|\mathcal{C}|-1} \lambda^2\cdot I_{Z_{i}}(\lambda)\Pr(\phi(\mathbf{Y}_{i-1})=0)  =\sum_{i=0}^{|\mathcal{C}|-1}  \dot{I}(p_i\lambda)\Pr(\phi(\mathbf{Y}_{i-1})=0),
\end{align*}
where the last equality follows from Lemma \ref{lem:hdot_idot}.
\end{proof}

Definition~\ref{def:bfHI} introduces
some useful notation for talking about the
aggregate contributions of \emph{some} cells to
\emph{some} period of time (on a log-scale) $W=[a,b]$, i.e., all $\lambda\in[e^a,e^b]$.
\begin{definition}\label{def:bfHI}
Fix a linearizable sketch.
Let $C\subset \mathcal{C}$ be a collection of cells and $W\subset \mathbb{R}$ be an interval
of the reals. Define:
\begin{align*}
    \mathbf{H}(C\to W) &= \int_W \sum_{c_i\in C} \dot{H}(p_ie^x)\Pr(\phi(\mathbf{Y}_{i-1,e^x})=0) dx,\\
    \mathbf{I}(C\to W) &= \int_W \sum_{c_i\in C} \dot{I}(p_ie^x)\Pr(\phi(\mathbf{Y}_{i-1,e^x})=0) dx.
\end{align*}
\end{definition}

A linearizable sketching \emph{scheme} is 
really an algorithm that takes a few parameters, 
such as a desired space bound and a maximum allowable cardinality, and produces a partition $\mathcal{C}$
of the \Dartboard, a function $\phi$ (implicitly defining the state space $\mathcal{S}$), 
and a cardinality estimator $\hat{\lambda} : \mathcal{S}\rightarrow \mathbb{R}$.  
Since we are concerned with asymptotic performance we can assume $\hat{\lambda}$ is MLE,
so the sketch is captured by just $\mathcal{C},\phi$.

In Theorem~\ref{thm:lowerbound} 
we assume that such a linearizable
sketching scheme has produced $\mathcal{C},\phi$ such
that the \emph{entropy} (i.e., space, in expectation)
is \emph{at most} $\tilde{H}$ at all times, and that the
normalized information is \emph{at least} $\tilde{I}$
for all times $\lambda\in [e^a,e^b]$.
It is proved that $\tilde{H}/\tilde{I} \geq (1-o_d(1))H_0/I_0$, where $d=b-a$ is the width of the interval and $o_d(1)\to 0$ as $d\to \infty$.  The take-away message (proved in Corollary~\ref{cor:lb-scaleinvariant-H0I0}) 
is that all 
linearizable sketches have $\fish$-number at least $H_0/I_0$.

\begin{theorem}\label{thm:lowerbound}
Fix reals $a<b$ with $d=b-a>1$. Let $\tilde{H},\tilde{I}>0$. If a linearizable sketch satisfies that 
\begin{itemize}
    \item For all $\lambda>0$, $H(\mathbf{Y}_{[\lambda]})\leq \tilde{H}$,
    \item For all $\lambda \in[e^a,e^{b}]$, $\lambda^2\cdot I_{\mathbf{Y}}(\lambda) \geq \tilde{I}$,
\end{itemize}
 then 
\begin{align*}
    \frac{\tilde{H}}{\tilde{I}} \;\geq\; \frac{H_0}{I_0}\frac{1-\max\{8d^{-1/4},5e^{-d/2}\}}{1+\frac{(344+4\sqrt{d})}{d}\frac{H_0}{I_0} \left(1-\max\{8d^{-1/4},5e^{-d/2}\}\right)}
    \;=\; (1-o_d(1))\frac{H_0}{I_0}
    .
\end{align*}
\end{theorem}

The expression for this $1-o_d(1)$ factor arises from the following two technical lemmas,
proved in Appendix~\ref{sect:proofs_lowerbound}.

\begin{lemma}\label{lem:H_relax}
For any $d>0$ and $t\geq \frac{1}{2}\ln d$,
\begin{align*}
   \frac{\int_{-\infty}^{-t} \dot{H}(e^x) dx}{\int_{-\infty}^{-t+d} \dot{H}(e^x) dx} \leq \max\{8d^{-1/4},5e^{-d/2}\}.
\end{align*}
\end{lemma}

\begin{lemma}\label{lem:I_relax}
Let $d=b-a>1$, $\Delta  = \frac{1}{2}\ln d$ and $\mathcal{C}^*=\{c_i\in \mathcal{C}\mid p_i<e^{-a-\Delta}\}$. Assume that for all $\lambda>0$, $H(\mathbf{Y}_{[\lambda]})\leq \tilde{H}$ (the first condition of Theorem \ref{thm:lowerbound}). Then we have
\begin{align*}
    \mathbf{I}(\mathcal{C}\setminus\mathcal{C}^*\to [a,b]) \leq (344+4e^{\Delta})\tilde{H}.
\end{align*}
\end{lemma}

\begin{proof}[Proof of Theorem~\ref{thm:lowerbound}]
First, since for all $\lambda\in[e^a,e^b]$, we have both $H(\mathbf{Y}_{[\lambda]})\leq \tilde{H}$ and $\lambda^2\cdot I_{\mathbf{Y}}(\lambda) \geq\tilde{I}$, we know, by Lemma \ref{lem:chain_rule},
\begin{align}
    \frac{\mathbf{H}(\mathcal{C}\to [a,b])}{\mathbf{I}(\mathcal{C}\to [a,b])}
    \;=\;
    \frac{\int_a^b H(\mathbf{Y}_{[e^x]})dx}{\int_a^b e^{2x}I_{\mathbf{Y}}(e^{x}) dx} 
    \;\leq\; \frac{ \tilde{H}d}{\tilde{I}d}
    \;=\;
    \frac{\tilde{H}}{\tilde{I}}.\label{eq:step_1}
\end{align}
Thus it is sufficient to bound $\frac{\mathbf{H}(\mathcal{C}\to [a,b])}{\mathbf{I}(\mathcal{C}\to [a,b])}$. 
Define $\Delta  = \frac{1}{2}\ln d$ and $\mathcal{C}^*=\{c_i\in \mathcal{C}\mid p_i<e^{-a-\Delta}\}$. 
We then have
\begin{align}
   &\frac{\mathbf{H}(\mathcal{C}\to [a,b])}{\mathbf{I}(\mathcal{C}\to [a,b])}
    \geq  \frac{\mathbf{H}(\mathcal{C}^*\to [a,b])}{\mathbf{I}(\mathcal{C}\to [a,b])}
    = \frac{\mathbf{H}(\mathcal{C}^*\to [a,b])}{\mathbf{I}(\mathcal{C}^*\to [a,b])}\cdot \frac{\mathbf{I}(\mathcal{C}^*\to [a,b])}{\mathbf{I}(\mathcal{C}\to [a,b])}.\label{eq:step_2}
\end{align}
We shall bound $\frac{\mathbf{H}(\mathcal{C}^*\to [a,b])}{\mathbf{I}(\mathcal{C}^*\to [a,b])}$
and $\frac{\mathbf{I}(\mathcal{C}^*\to [a,b])}{\mathbf{I}(\mathcal{C}\to [a,b])}$ separately.

\medskip

First, for any cell $c_i\in \mathcal{C}^*$, 
let $f(t)=\dot{H}(p_ie^t)$, $g(t)=\dot{I}(p_ie^t)$ and $h(t)=\Pr(\phi(\mathbf{Y}_{i-1,e^t})=0)$. By Lemma \ref{lem:monotonicity} and Proposition \ref{prop:phi_monotone}, we know that 
$f(t)/g(t)$ and $h(t)$ are non-increasing 
in $t$. By Lemma \ref{lem:integration}, we know both $f(t)$ and $g(t)$ have finite integral over $(-\infty,\infty)$. It is also easy to see that $f(t)>0$, $g(t)>0$ and $h(t)\in[0,1]$ for all $t\in \mathbb{R}$. 
By the first part of Lemma \ref{lem:rational_inequality} (Appendix~\ref{sect:lem:rational_inequality}) 
we conclude that
\begin{align*}
    \frac{\int_a^b \dot{H}(p_ie^t) \Pr(\phi(\mathbf{Y}_{i-1,e^t})=0) dt}{\int_a^b \dot{I}(p_ie^t) \Pr(\phi(\mathbf{Y}_{i-1,e^t})=0) dt} \geq \frac{\int_a^b \dot{H}(p_ie^t) dt}{\int_a^b \dot{I}(p_ie^t)  dt}.
\end{align*}
In addition, we have
\begin{align*}
    \frac{\int_a^b \dot{H}(p_ie^t) dt}{\int_a^b \dot{I}(p_ie^t)  dt} \geq \frac{\int_a^b \dot{H}(p_ie^t) dt}{\int_{-\infty}^b \dot{I}(p_ie^t)  dt}=\frac{\int_{-\infty}^b \dot{H}(p_ie^t) dt}{\int_{-\infty}^b \dot{I}(p_ie^t)  dt}\cdot \frac{\int_a^b \dot{H}(p_ie^t) dt}{\int_{-\infty}^b \dot{H}(p_ie^t)  dt}\geq \frac{\int_{-\infty}^\infty \dot{H}(p_ie^t) dt}{\int_{-\infty}^\infty \dot{I}(p_ie^t)  dt}\cdot \frac{\int_a^b \dot{H}(p_ie^t) dt}{\int_{-\infty}^b \dot{H}(p_ie^t)  dt},
\end{align*}
where the last inequality follows from
the second part of Lemma~\ref{lem:rational_inequality}
(Appendix~\ref{sect:lem:rational_inequality}).
By Lemma \ref{lem:integration} we know that $\frac{\int_{-\infty}^\infty \dot{H}(p_ie^t) dt}{\int_{-\infty}^\infty \dot{I}(p_ie^t)  dt}= H_0/I_0$.
By applying Lemma~\ref{lem:H_relax}, we have
\begin{align*}
    \frac{\int_a^b \dot{H}(p_ie^t) dt}{\int_{-\infty}^b \dot{H}(p_ie^t)  dt}=1-\frac{\int_{-\infty}^a \dot{H}(p_ie^t) dt}{\int_{-\infty}^b \dot{H}(p_ie^t)  dt}=1-\frac{\int_{-\infty}^{a+\ln p_i} \dot{H}(e^t) dt}{\int_{-\infty}^{a+\ln p_i +d} \dot{H}(e^t)  dt}\geq 1-\max\{8d^{-1/4},5e^{-d/2}\}.
\end{align*}
Note here that since cell $c_i\in \mathcal{C}^*$,
$a+\ln p_i \leq a+(-a-\Delta)
= -\frac{1}{2}\ln d$, as required by Lemma~\ref{lem:H_relax}. 
Therefore, for any $c_i\in \mathcal{C}^*$, we have
\begin{align*}
     \frac{\int_a^b \dot{H}(p_ie^t) \Pr(\phi(\mathbf{Y}_{i-1,e^t})=0) dt}{\int_a^b \dot{I}(p_ie^t) \Pr(\phi(\mathbf{Y}_{i-1,e^t})=0) dt} \geq \frac{H_0}{I_0}(1-\max\{8d^{-1/4},5e^{-d/2}\}).
\end{align*}
Summing over all cells in $\mathcal{C}^*$, this also implies that
\begin{align}
    \frac{\mathbf{H}(\mathcal{C}^*\to [a,b])}{\mathbf{I}(\mathcal{C}^*\to [a,b])}= \frac{\sum_{c_i\in \mathcal{C}^*}\int_a^b \dot{H}(p_ie^t) \Pr(\phi(\mathbf{Y}_{i-1,e^t})=0) dt}{\sum_{c_i\in \mathcal{C}^*}\int_a^b \dot{I}(p_ie^t) \Pr(\phi(\mathbf{Y}_{i-1,e^t})=0) dt}\geq \frac{H_0}{I_0}(1-\max\{8d^{-1/4},5e^{-d/2}\}). \label{eq:main_1}
\end{align}
Secondly, by Lemma \ref{lem:I_relax}, we have
\begin{align*}
    \frac{\mathbf{I}(\mathcal{C}^*\to [a,b])}{\mathbf{I}(\mathcal{C}\to [a,b])}=1-\frac{\mathbf{I}(\mathcal{C}\setminus\mathcal{C}^*\to [a,b])}{\mathbf{I}(\mathcal{C}\to [a,b])}\geq 1-\frac{(344+4e^{\Delta})\tilde{H}}{\mathbf{I}(\mathcal{C}\to [a,b])}.
\end{align*}
Since for all $\lambda\in[e^a,e^b]$, $\lambda^2\cdot I_{\mathbf{Y}}(\lambda) \geq\tilde{I}$, we have, by Lemma \ref{lem:chain_rule}
\begin{align*}
    \mathbf{I}(\mathcal{C}\to [a,b]) = \int_a^b e^{2t}I_{\mathbf{Y}}(e^t)dt\geq \tilde{I}d.
\end{align*}
Thus we have
\begin{align}
    \frac{\mathbf{I}(\mathcal{C}^*\to [a,b])}{\mathbf{I}(\mathcal{C}\to [a,b])}\geq 1-\frac{(344+4e^{\Delta})\tilde{H}\ln 2}{\tilde{I}d}.\label{eq:main_2}
\end{align}
By combining inequalities (\ref{eq:step_1}), (\ref{eq:step_2}), (\ref{eq:main_1}), and (\ref{eq:main_2}), we have
\begin{align}
    \frac{\tilde{H}}{\tilde{I}}
    &\geq    \frac{H_0}{I_0} \left(1-\max\{8d^{-1/4},5e^{-d/2}\}\right)\left(1-\frac{(344+4\sqrt{d})\tilde{H}}{d\tilde{I}}\right), \label{eq:combine}
\intertext{and by rearranging inequality (\ref{eq:combine}), we finally conclude that}
    \frac{\tilde{H}}{\tilde{I}} 
    &\geq \frac{H_0}{I_0} \frac{1-\max\{8d^{-1/4},5e^{-d/2}\}}{1+\frac{(344+4\sqrt{d})}{d}\frac{H_0}{I_0} \left(1-\max\{8d^{-1/4},5e^{-d/2}\}\right)}
    \;=\; (1-o_d(1))\frac{H_0}{I_0}.\nonumber
\end{align}
\end{proof}

\begin{cor}\label{cor:lb-scaleinvariant-H0I0}
Let $A_q$ be any linearizable, weakly scale-invariant sketch with base $q$. 
Then $\fish(A_q)\geq H_0/I_0$.
\end{cor}
\begin{proof}
Fix $m>0$.  Let $A_q^m$ be a vector of $m$ independent, offsetted $A_q$ sketches
with respect to the uniform offset vector $(0,1/m,2/m,\ldots,(m-1)/m)$.
First note that, since $A_q$ is linearizable, $A_q^m$ is also linearizable
since we can simply concatenate the linear orders on the cells of each independent subsketch.

Let $\tilde{H}_m=\sup\{H(X_{A_q^m,q^r})\mid r\in[0,1/m]\}$ and $\tilde{I}_m=\inf\{q^{2r}I_{A_q^m}(q^r)\mid r\in[0,1/m]\}$.
Since $A_q^m$ is weakly scale-invariant sketch with base $q^{1/m}$, for any $\lambda>0$ 
we have 
\begin{align*}
    H(X_{A_q^m,\lambda}) \leq \tilde{H}_m\quad\text{ and } \quad \lambda^2\cdot I_{A_q^m}(\lambda) \geq \tilde{I}_m.
\end{align*}
Therefore we can apply Theorem \ref{thm:lowerbound} to $A_q^m$ with arbitrary large $d=b-a$. This implies that $\tilde{H}_m/\tilde{I}_m \geq H_0/I_0$. 
On the other hand, note that as $m$ becomes large, the sketch is smoothed, i.e., $\tilde{H}_m/\tilde{I}_m$ converges to $\fish(A_q)$ as $m\to \infty$. We conclude that $\fish(A_q)\geq H_0/I_0$.
\end{proof}

\section{\fishmonger: A Compressed, Smoothed $\PCSA$-based Sketch}\label{sect:fishmonger}

The results of Section~\ref{sect:fish-numbers} can properly 
be thought of as \emph{lower bounds} 
on the performance of $\qPCSA$ and $\qLL$, 
so it is natural to ask whether there are matching \emph{upper bounds}, 
at least in principle. 
Specifically, we can always compress a sketch so
that its \emph{expected} size is equal to its entropy.
Scheuermann and Mauve~\cite{ScheuermannM07} and Lang~\cite{Lang17} 
showed that this is effective experimentally, and Lang~\cite{Lang17}
numerically 
calculated the 
entropy of 2-$\PCSA$ 
and 2-$\LL$. 
However, it may be desirable to store the sketch in 
a \emph{fixed} memory footprint, i.e., to guarantee a certain
worst case size bound \emph{at all times}.

\medskip 

In this section we describe a sketch \fishmonger{} 
that can be stored in 
$(1+o(1))mH_0$ $+ O(\sqrt{m\log\log U}+\log^2 \log U)$ bits and 
achieves a standard error of 
$(1+o(1))\sqrt{1/(mI_0)}$.  
\fishmonger{} is based on a smoothed, compressed $e$-$\PCSA$ sketch, 
with a different estimation function, 
and a fixed level of redundancy.
It is characterized by the following features.
\begin{itemize}
    \item The \emph{abstract} state space of the sketch is 
    $\{0,1\}^{m\times(\ln U+2c\ln m)}$.  Due to compression the 
    \emph{true} state space of the sketch is in correspondence
    with a \emph{subset} of $\{0,1\}^{m\times(\ln U+2c\ln m)}$.
    (Whenever these two states need to be distinguished,
    denote $\check{S}$ to be the abstract state and $S$ the
    true state.  At time zero we have $S_0=\check{S}_0=0$.)
    \item When $a\in [U]$ is 
    processed we interpret 
    $h(a)$ as a random matrix $Z_a \in \{0,1\}^{m\times(\ln U+2c\ln m)}$ 
    where $\Pr(Z_a(i,j)=0) = e^{-e^{-j-i/m}}$
    and $j$ ranges 
    from $-c\ln m+1$ to $\ln U+c\ln m$.
    (The constant $c$ controls the 
    error probabilities.)
    We then set 
    $S \leftarrow S\vee Z_a$ (component-wise OR).
    In other words, the rows 
    $S(0),\ldots,S(m-1)$
    are independent $e$-$\PCSA$-type 
    sketches, but indexed starting below 0, at $-c\ln m+1$, 
    and effecting a uniform offset vector 
    (see Section~\ref{sect:randomoffsets}) of:
    \[
    (0,1/m,2/m,\ldots,(m-1)/m).
    \]
    \item The cardinality is estimated using the 
    Maximum Likelihood Estimator. 
    Define $l(S \mid \lambda)$ to be the $\log_2$-likelihood of seeing $S$ after 
    sketching a set of cardinality $\lambda$. 
    \[
    l(S \mid \lambda) = \log_2\left(\Pr_{Z_1,\ldots,Z_\lambda}(Z_1\vee\cdots \vee Z_{\lambda} = S)\right).
    \]
    The estimator is then defined to be
    \[
    \hat{\lambda}(S) = \argmax_{\lambda} l(S \mid \lambda).
    \]
    The MLE can be computed in $O(m\poly(\log U))$ time via binary search.
    Clifford and Cosma~\cite{CliffordC12} and Ertl~\cite{Ertl17}, 
    discuss MLE algorithms with improved convergence for \textsf{LogLog}-type 
    sketches but the ideas carry over to $\PCSA$ as well.

    \item The sketch stores the estimate $\hat{\lambda}(S)$ explicitly,
    then allocates 
    $(1+o(1))m\cdot \mathcal{H}(e\text{-}\PCSA) + B \leq (1+o(1))m H_0 + B$ 
    bits for storing $S$.  If 
    $-l\left(S \;\middle|\; \hat{\lambda}(S)\right) \leq 
    (1+o(1))m\cdot H_0 + B$ then $S$ is successfully stored.
    If not, then the last update to $S$ cannot be 
    recorded ($\check{S}\neq S$)
    and the state of the sketch reverts to its state before 
    processing the last element.
\end{itemize}

The crux of the analysis is to show that when 
$B=O(\sqrt{m\log(m\log U)}+\log^2\log U)$, 
it is \emph{always} possible to store $\check{S}$ 
in compressed form, with high probability $1-1/\poly(m\log U)$.

\begin{theorem}\label{thm:fishmonger}
The \fishmonger{} algorithm processes a sequence $\mathcal{A} \in [U]^*$ 
and maintains a sketch $S$ using 
\[
O(\sqrt{m\log\log U}+\log^2\log U) + (1+o(1))mH_0 
\approx 3.25724m \mbox{ bits for $m\gg \log^2\log U$.} 
\]
that ideally represents an abstract 
$m\log U$-bit sketch $\check{S}$. 
With probability $1-1/\poly(m)$, 
$S=\check{S}$ at all times, and $\hat{\lambda}(S)$ 
is an asymptotically unbiased
estimate, as $m\to\infty$, of the 
cardinality $\lambda$ of $\mathcal{A}$ with standard error 
\[
\sqrt{\frac{(1+o(1))}{mI_0}} \approx \frac{0.77969}{\sqrt{m}}.
\]
\end{theorem}

The remainder of this section constitutes a proof of Theorem~\ref{thm:fishmonger}.
We make use of Bernstein's inequality.

\begin{theorem}[See {\cite{DubhashiPanconesi09}}] \label{thm:concentration}
Let $X_0, \ldots, X_{m-1}$ be independent random variables
such that $X_i - \E(X_i)\leq M$ for all $i$.
Let $X=\sum_i X_i$ and $V=\sum_i \Var(X_i)$. Then
\[
\Pr(X > \E(X) + B) < \exp\left(-B^2\middle/\left(2V + 2MB/3\right)\right).
\]
\end{theorem}

Observe that the number of times the abstract state 
$\check{S}$ can change is $m' = m(\ln U+2c\ln m)$. 
Since the sketch is idempotent, we can conflate ``time''
with cardinality, and let $S_\lambda,\check{S}_\lambda$ be the states
after seeing $\lambda$ distinct elements.
We will first prove that at any 
particular time $\lambda$, the probability that $\check{S}_\lambda$
cannot be stored in the specified number of bits
is low, namely $1/\poly(m')$.  We then argue 
that this implies that $\forall \lambda. \check{S}_\lambda=S_\lambda$ 
holds with probability $1-1/\poly(m')$, i.e., the actual state is identical to the abstract state \emph{at all times}.

\medskip

Fix any time $\lambda$.
By the independence of the rows $\{\check{S}(i)\}_{i\in [m]}$ of $\check{S}$ we have
\begin{align*}
    H(\check{S} \mid \lambda) 
            &= \sum_{i\in [m]} H(\check{S}(i) \mid \lambda)\\
            &= \sum_{i\in [m]} \E(-l(\check{S}(i) \mid \lambda))\\
            &= (1+o(1))m\cdot \mathcal{H}(e\text{-}\PCSA)
            \;=\; (1+o(1))mH_0,
\intertext{where the last line follows 
from Theorem~\ref{thm:qPCSA} 
and 
the fact that in the limit ($m\to\infty$), 
the offset vector is uniformly dense in $[0,1)$.
By definition of the MLE $\hat{\lambda}(\check{S})$, 
we have for every state $\check{S}$,}
-l\left(\check{S} \mid \hat{\lambda}(\check{S})\right) &\leq -l(\check{S} \mid \lambda).
\end{align*}
In particular,
\[
\Pr\Big(-l(\check{S} \mid \hat{\lambda}(\check{S})) 
\,>\, H(\check{S} \mid \lambda) + B\Big) 
            \;\leq\; \Pr\Big(-l(\check{S} \mid \lambda) \,>\, H(\check{S} \mid \lambda) + B\Big).
\]
Thus, it suffices to analyze the distribution 
of the upper tail of 
$-l(\check{S} \mid \lambda)$.  

Define $X_{i,j}$ to be the log-likelihood $-l(\check{S}(i,j) \mid \lambda)$.
Note that $\check{S}(i,j)$ is Bernoulli with
$p_{i,j} = \Pr(\check{S}(i,j)=0) 
= e^{-\lambda e^{-(j+i/m)}}$.  
In particular, if $j>\ln \lambda$, 
$p_{i,j} = 1-\Theta(\lambda e^{-j})$
and $p_{i,j}\log^2 p_{i,j} = \Theta((\lambda e^{-j})^2)$
and if $j < \ln \lambda$ then
$p_{i,j} = e^{-\Theta(\lambda e^{-j})}$
and $(1-p_{i,j})\log^2(1-p_{i,j}) 
= e^{-\Theta(\lambda e^{-j})}$.
Due to the independence of the $(X_{i,j})$,
the total variance $V$ is therefore
\begin{align*}
    V &= \Var(-l(\check{S} \mid \lambda))\\
    &= \sum_{i\in [m],j\in [-c\ln m+1,\ln U+c\ln m]} \Var(X_{i,j})\\
    &\leq \sum_{i\in [m],j\in [-c\ln m+1,\ln U+c\ln m]} \left(p_{i,j}\log_2^2 p_{i,j} + (1-p_{i,j})\log_2^2(1-p_{i,j})\right)\\
    &\leq Cm,
\end{align*}
for some constant $C$.

Define 
$\mathscr{I} \subset [m]\times [-c\ln m+1, \ln U+c\ln m]$ 
to be the set of all indices $(i,j)$
such that 
\[
\ln\lambda - \ln(c\ln m')   \leq   j+i/m   \leq  \ln \lambda + c\ln m'.
\]
If $(i,j)\in\mathscr{I}$ with $j+i/m \geq \ln \lambda$
then $\Pr(\check{S}(i,j)=1) = 1-p_{i,j} = \Theta(e^{-(j+i/m) + \ln \lambda}) = \Omega((m')^{-c})$.
If $(i,j)\in \mathscr{I}$ with $j+i/m \leq \ln \lambda$
then $\Pr(\check{S}(i,j)=0) = \Theta(e^{\lambda e^{-(j+i/m)}}) = \Omega((m')^{-c})$.
Thus, for each cell $(i,j)$ within $\mathscr{I}$, 
$X_{i,j}$ satisfies a \emph{worst case} deviation of
\[
-l(\check{S}(i,j) \mid \lambda) - E(X_{i,j}) \le c\log_2 m' + O(1) \bydef M.
\]
Redefine $X_{i,j}$ so that this deviation of $M$ is satisfied outside
$\mathscr{I}$ as well.
\begin{align*}
X_{i,j} &= \min\left\{-l\left(\check{S}(i,j) \;\middle|\; \lambda\right),\; M\right\},\\
X       &= \sum_{i\in [m], j\in [-c\ln m+1, \ln U+c\ln m]} X_{i,j}.
\end{align*}
\noindent We choose
\begin{align*}
B &= \sqrt{2Cm\ln\epsilon^{-1}}+ (2/3)M\ln\epsilon^{-1},
\end{align*}
\noindent and apply Theorem~\ref{thm:concentration}.
\begin{align*}
\Pr\left(X > H(\check{S}\mid\lambda) + B\right) 
&\leq \exp\left(-B^2\middle/\left(2V + (2/3)MB\right)\right)\\
&\leq \exp\left(-\frac{B^2}{2Cm + (2/3)MB}\right)\\
&< \epsilon.
\end{align*}
Outside of $\mathscr{I}$, the most probable outcomes (i.e., 
those minimizing negated log-likelihood) are to have
$\check{S}(i,j) = 1$ whenever $j+i/m$ is too small to be in $\mathscr{I}$
and $\check{S}(i,j) = 0$ whenever $j+i/m$ is too large to be in $\mathscr{I}$.
When this occurs, $X$ is identical to $-l(\check{S} \mid \lambda)$.  
By a union bound, this fails to occur with probability at most 
$m'\cdot (m')^{-c} = (m')^{-c+1}$.
Thus, with probability at least $1-\epsilon - (m')^{-c+1}$
we achieve the successful outcome
\[
-l(\check{S} \mid \lambda) 
\;=\; 
X
\;\leq\; H(\check{S} \mid \lambda) + B 
\;\leq\; (1+o(1))mH_0 + B.
\]
We set $\epsilon = (m')^{-c+1}$ and hence
\[
B = O\left(\sqrt{m\log m'} + \log^2(m')\right) = O\left(\sqrt{m\log(m\log U)} + (\log\log U)^2\right).
\]

\medskip

At first glance, setting $\epsilon$ so high 
seems insufficient to the task of proving that w.h.p.,
$\forall \lambda.\, \check{S}_\lambda = S_\lambda$. 
Ordinarily we would take a union bound over all $\lambda\in [1,U]$, 
necessitating an $\epsilon \ll U^{-1}$.
The key observation is that
$S$ changes at most $m'$ times, so it suffices to take a union bound
over a set $\Lambda$ of \emph{checkpoint} 
times that witness all states of the sketch.

Define $\epsilon_0 = \sqrt{\epsilon}$ and 
$\Lambda=\{\lambda_1,\lambda_2,\ldots\}$ to be the set of all times
(i.e., cardinalities) of the form 
\[
\lambda_k = \floor{(1+\epsilon_0)^k} \leq U.
\]
By a union bound, we fail to have 
success at all checkpoint times in $\Lambda$ with probability at most 
\[
|\Lambda|\cdot 2\epsilon 
< 
(\log_{1+\epsilon_0} U)\cdot 2\epsilon 
= 
O(\epsilon\cdot  \epsilon_0^{-1}\log U) = O(\epsilon_0 \log U).
\]
We now need to argue that all states of the data structure can be witnessed, w.h.p.,
by only checking it at times in $\Lambda$, i.e., in any interval $(\lambda_k,\lambda_{k+1})$,
the state changes at most once.  

Observe the the probability that the
next element causes a change to 
the sketch never increases, since
bits in $\check{S}$ or $S$ only 
get flipped from 0 to 1.
Define $P_k$ to be the probability,
at time $\lambda_k$, that the next
element causes a change to the sketch.
Observe that $P_k$ is itself a random variable: 
it is the probability that the next $Z_a$ contains a 1 
in some location that is 0 in $\check{S}$. 
It is straightforward to show that when the true
cardinality is $\lambda_k$, $\E(P_k) = \Theta(m/\lambda_k)$,
and via Chernoff-Hoeffding bounds~\cite{DubhashiPanconesi09},
that $\Pr(P_k > c'm'/\lambda_k) = \exp(-m')$ 
for a sufficiently large constant $c'$.  
Thus we proceed under the assumption
that $P_k = O(m'/\lambda_k)$ for all $k$.

If checkpoints $\Lambda$ do \emph{not} witness all states of the sketch,
then there must have been an index $k$ such that the sketch
changed state \emph{twice} in the interval $(\lambda_k,\lambda_{k+1})$.
For fixed $k$, the probability that this occurs is, by a union bound, at most
\begin{align*}
{\lambda_{k+1}-\lambda_k \choose 2}P_k^2 
 < (\epsilon_0 \lambda_k)^2 (c'm'/\lambda_k)^2
 = O(\epsilon (m')^2).
\end{align*}
Taking another union bound over all $k$ 
shows that $\Lambda$ fails to witnesses all sketch states with probability 
\begin{align*}
O(|\Lambda|\epsilon (m')^2) 
&= O(\epsilon_0^{-1}\log U \epsilon (m')^2)
= O(\epsilon_0 (m')^{3})\\
&= O((m')^{-(c-1)/2 + 3}).
\end{align*}
Setting $c$ sufficiently large, we conclude that
\[
\Pr(\forall \lambda.\, \check{S}_\lambda = S_\lambda) \ge 1 - 1/\poly(m').
\]
Whenever $\check{S}=S$, 
Theorem~\ref{thm:qPCSA} 
implies the standard error of $\hat{\lambda}$ is
\[
\sqrt{\frac{1+o(1)}{m\cdot \mathcal{I}(e\text{-}\PCSA)}}
\;=\; \sqrt{\frac{(1+o(1))}{mI_0}}
\;\leq\; \frac{(1+o(1))0.77969}{\sqrt{m}}.
\]
The space used by the sketch (in bits) is
\begin{align*}
&\log U + O(\sqrt{m\log(m\log U)}+\log^2\log U) + (1+o(1))m\cdot \mathcal{H}(e\text{-}\PCSA)\\
&= 
(1+o(1))\left(\log U + mH_0\right).
\end{align*}

Here the $\log U$ term accounts for the cost of explicitly 
storing the estimate $\hat{\lambda}(S)$.  This can be further
reduced to $O(\log\log U + \log m)$ bits by storing instead 
a floating point approximation 
\[
\tilde{\lambda} \in [\hat{\lambda}, (1+1/m')\hat{\lambda}].
\]
By using $\tilde{\lambda}$ in lieu
of $\hat{\lambda}$ we degrade the efficiency of the arithmetic encoding.
The efficiency loss is
$-l(\check{S} \mid \tilde{\lambda}) + l(\check{S} \mid \hat{\lambda})$.
Fix an entry $(i,j)$. 
Define $\hat{p} = e^{-\hat{\lambda}e^{-(j+i/m)}}$ to be the probability 
that $\check{S}(i,j)=0$, assuming cardinality $\hat{\lambda}$,
and define $\tilde{p} = e^{-\tilde{\lambda}e^{-(j+i/m)}}$ 
analogously for $\tilde{\lambda}$.  
The loss in encoding efficiency for location $(i,j)$ 
is the KL-divergence between 
the two distributions, i.e.,
\begin{align*}
D_{\operatorname{KL}}(\hat{p} \,\|\, \tilde{p}) 
&= \hat{p}\log_2\left(\frac{\hat{p}}{\tilde{p}}\right) + (1-\hat{p})\log_2\left(\frac{1-\hat{p}}{1-\tilde{p}}\right)\\
&\leq \hat{p}\log_2\left(\frac{\hat{p}}{\tilde{p}}\right) 
\hspace*{1cm} \mbox{$\tilde{\lambda}\ge \hat{\lambda}$, hence $\tilde{p} \leq \hat{p}$}\\
&= \hat{p}\frac{1}{\ln 2}(\tilde{\lambda}-\hat{\lambda})e^{-(j+i/m)}\\
&\leq \hat{p}\frac{1}{\ln 2}\hat{\lambda}e^{-(j+i/m)}/m'\\
&= \hat{p}\log_2(\hat{p}^{-1}) / m'
< H(\hat{p})/m' < 1/m'.
\end{align*}
In other words, 
over all $m'$ entries 
in $\check{S}$, 
the total loss in encoding efficiency due to using $\tilde{\lambda}$ 
is less than 1 bit.

\begin{rem}
In the proof of Theorem~\ref{thm:fishmonger} 
we treated the unlikely event that $\check{S}\neq S$ 
as a \emph{failure}, but in practice nothing bad happens.
As these errors occur with probability $1/\poly(m\log U)$ 
they have a negligible effect on the standard error.

The proof could be simplified considerably if we do
not care about the dependence on $U$.  
For example, we could set 
$\epsilon = 1/\poly(U)$ and apply a standard union bound rather than look at the ``checkpoints'' $\Lambda$.
We could have also applied a recent
tail bound of Zhao~\cite{Zhao20} 
for the log-likelihood of a set 
of independent Bernoulli random variables.  These two simplifications
would lead to a redundancy of
\[
B = O\paren{\sqrt{m'\log\epsilon^{-1}}} = O(\sqrt{m}\log U).
\]
\end{rem}

\section{Conclusion}\label{sect:conclusion}

In this paper we introduced a new approach to measuring the efficiency of data sketches.
At a high level, a sketch records the outcomes of some probabilistic \emph{experiments}.
These outcomes must be stored (costing space)
and may be used to estimate statistical quantities such as cardinality.  
The \emph{Fisher-Shannon} ($\fish$) number
of the sketch is the most natural way
to measure the \emph{bang for the buck}
of the experiment, taking into account its storage cost, 
under an efficient encoding,
and value for estimating the parameter of interest, under a statistically optimal estimator.

We proved that, despite decades of work developing
\emph{improved} cardinality sketches, the original
\PCSA{} sketch of Flajolet and Martin~\cite{FlajoletM85} 
has the best \fish-number among popular sketches, 
and is even provably optimal in the class of 
\emph{linearizable} sketches, 
which is a subset of mergeable sketches.
As a practical matter, we proved that to achieve, 
say, $1\%$ standard error, the \fishmonger{} sketch
needs slightly more than 
$(0.01)^2 H_0/I_0 \approx $ 19,800 bits, 
which is about 2.42 kilobytes, while our lower bound
says that no linearizable sketch can 
beat 2.42 kilobytes.
This level of theoretical precision 
is uncommon
within the field of sketching and streaming.

\medskip 

We highlight two open problems.
\begin{itemize}
    \item Shannon entropy and Fisher information are both 
    subject to data processing inequalities, i.e., no
    deterministic transformation can increase entropy/information.
    Our lower bound (Section~\ref{sect:lowerbound}) 
    can be thought
    of as a specialized data processing inequality 
    for $\fish$, where the deterministic transformation maps $(Z_0,\ldots,Z_{|\mathcal{C}|-1})$ to 
    $(Y_0,\ldots,Y_{|\mathcal{C}|-1})$, 
    but with two notable
    features.  First, the transformation has to 
    be of a certain 
    type (the \emph{linearizability} assumption).
    Second, we need to measure $\mathcal{H}/\mathcal{I}$ 
    over a sufficiently long period of \emph{time}.  
    The second feature is essential to the $H_0/I_0$ lower bound.
    The open question is whether the first feature can be relaxed.
    We conjecture that $H_0/I_0$ is a lower bound on \emph{all mergeable sketches}, and believe that the state space characterization of mergeable sketches of Theorem~\ref{thm:mergeable-characterization} may be valuable in resolving this conjecture.

    \item Ever since Flajolet and Martin~\cite{FlajoletM85}, researchers have been rightfully disappointed that the cutting edge in cardinality estimation is not very accessible to the undergraduate (or even graduate) computer science population.
    In contrast, it is easily possible to teach \textsf{AMS} sketches, \textsf{CountMin}, and \textsf{CountSketch} rigorously.
    \L{}ukasiewicz and Uzna\'{n}ski~\cite{LukasiewiczU22}
    introduced a variant of (\Hyper)\LogLog{} using
    Gumble-distributed random varianbles, which simplified
    some aspects of the analysis.
    Chakraborty, Vinodchandran, and Meel~\cite{ChakrabortyVM22} showed that it is possible to do cardinality estimation using only ``fresh'' random bits, not a hash function \emph{per se}, 
    which incurs some significant loss in efficiency.
    We consider it an open problem to find a \emph{teachable} analysis of \PCSA{} or (\Hyper)\LogLog{} that demonstrates a standard error of $O(1/\sqrt{m})$ plus a small periodic term.
\end{itemize}

\paragraph{Acknowledgement.} 
We thank Liran Katzir for suggesting references~\cite{Ertl17,Ertl18,CliffordC12} and 
an anonymous reviewer for bringing the work of
Lang~\cite{Lang17} and Scheuermann and 
Mauve~\cite{ScheuermannM07} 
to our attention.  We also 
thank Bob Sedgewick and J{\'e}r{\'e}mie Lumbroso 
for discussing the 
cardinality estimation problem 
at Dagstuhl 19051.


\newcommand{\etalchar}[1]{$^{#1}$}

\appendix

\section{Proofs from Section \ref{sect:fish-numbers}}\label{sect:proofs}

\subsection{Lemma~\ref{lem:int_id}}\label{sect:int_id}

Lemma~\ref{lem:int_id} is applied in the proofs of Lemmas~\ref{lem:H0I0} and \ref{lem:HLL_tool},
in Sections~\ref{sect:lemH0I0} and \ref{sect:lemHLL_tool}, respectively.

\begin{lemma}\label{lem:int_id} For any $b>a>0$ we have the following identity.
\begin{align*}
    \int_0^\infty \frac{e^{-a x}-e^{-bx}}{x}dx= \ln b - \ln a.
\end{align*}
\end{lemma}

\begin{proof}
\begin{align*}
    \ln b-\ln a &= \int_a^b \frac{1}{t} dt\\
    &=- \left.\int_a^b\frac{e^{-x t}}{t}dt \; \right|^\infty_0\\
    &= -\left.\int_{ax}^{bx} \frac{e^{-r}}{r}dr\;\right|^\infty_0 
    & \mbox{\{Change of variable: $r=tx$\}}\\
    &=-\left.\left(\int_{ax}^\infty\frac{e^{-r}}{r}dr-\int_{bx}^\infty\frac{e^{-r}}{r}dr\right)\right|^\infty_0
\end{align*}
Note that 
\begin{align*}
    \frac{d}{dx}\int_{ax}^\infty\frac{e^{-r}}{r}dr=-\frac{e^{-ax}}{ax}a=-\frac{e^{-ax}}{x}.
\end{align*}
Thus we have 
\begin{align*}
    \int_0^\infty 
\frac{e^{-a x}-e^{-bx}}{x}dx
&= \left.\left(-\int_{ax}^\infty\frac{e^{-r}}{r}dr+\int_{bx}^\infty\frac{e^{-r}}{r}dr\right)\right|^\infty_0
= \ln b - \ln a.
\end{align*}
\end{proof}

\subsection{Proof of Lemma~\ref{lem:H0I0}}\label{sect:lemH0I0}

\begin{proof}[Proof of Lemma \ref{lem:H0I0}]
Let $u=e^w$, then $\frac{du}{dw}=u$.
\begin{align*}
    \ln 2\cdot H_0&=\int_{0}^\infty \frac{e^{-u}u -(1-e^{-u})\ln (1-e^{-u})}{u} du\\
    &=1+\int_{0}^\infty \frac{-(1-e^{-u})\ln (1-e^{-u})}{u} du\\
    &=1+\int_{0}^\infty \frac{(1-e^{-u}) \sum_{k=1}^\infty \frac{e^{-ku}}{k}}{u} du & \mbox{\{Taylor exp.\}}\\
    &=1+ \sum_{k=1}^\infty \frac{1}{k}\int_{0}^\infty \frac{(e^{-ku}-e^{-(k+1)u})}{u} du.
\intertext{Applying Lemma \ref{lem:int_id} (Appendix~\ref{sect:int_id}), we have}
    H_0 &= \frac{1}{\ln 2}+\sum_{k=1}^\infty \frac{1}{k} \log_2\left(\frac{k+1}{k}\right).
\end{align*}
We now prove that $I_0=\pi^2/6$. 
Letting $u=e^r$, we have $\frac{du}{dr}=u$ and can write $I_0$ as
\[
I_0 = \int_{0}^\infty \frac{(e^r)^2}{e^{e^r}-1}dr
    = \int_{0}^\infty \frac{u^2}{u(e^{u}-1)}du
    = \int_{0}^\infty \frac{u}{e^{u}-1}du,
\]
which is exactly the integral representation of the 
Riemann zeta function evaluated at $s=2$. We conclude that
\[
I_0 = \zeta(2)=\frac{\pi^2}{6}.
\]
\end{proof}

\subsection{Proof of Lemma~\ref{lem:HLL_tool}}\label{sect:lemHLL_tool}

\begin{proof}[Proof of Lemma \ref{lem:HLL_tool}]
Let $w=e^r$ and $\frac{dw}{dr}=w$. We have
\begin{align*}
    \lefteqn{\ln 2\cdot \phi(q)}\\
    &=\int_{0}^\infty \frac{-(e^{-w}-e^{-qw})\ln(e^{-w}-e^{-qw})}{w}dw\\
    &= \int_{0}^\infty \frac{-(e^{-w}-e^{-qw})\ln(1-e^{-(q-1)w})}{w}dw
    +\int_{0}^\infty (e^{-w}-e^{-qw})dw.
\intertext{Note that $\int_{0}^\infty (e^{-w}-e^{-qw})dw=1-\frac{1}{q}$.  Continuing with a Taylor expansion of the logarithm, we have}
    &= \int_{0}^\infty \frac{(e^{-w}-e^{-qw})\sum_{k=1}^\infty \frac{e^{-k(q-1)w}}{k}}{w}dw+1-\frac{1}{q} \\
    &= \sum_{k=1}^\infty\frac{1}{k}\int_{0}^\infty \frac{e^{-(k(q-1)+1)w}-e^{-(k(q-1)+q)w} }{w}dw+1-\frac{1}{q}.
\intertext{Applying Lemma \ref{lem:int_id} (Appendix~\ref{sect:int_id}) to the integral, this is equal to}
     &= 1-\frac{1}{q} + \sum_{k=1}^\infty\frac{1}{k}\ln\left(\frac{kq-k+q}{kq-k+1}\right)\\
     &= 1-\frac{1}{q} + \sum_{k=1}^\infty\frac{1}{k}\ln\left(\frac{k+\frac{1}{q-1}+1}{k+\frac{1}{q-1}}\right).
\intertext{Hence $\phi(q)$ is }
\phi(q) &= \frac{1-1/q}{\ln 2} + \sum_{k=1}^\infty\frac{1}{k}\log_2 \left(\frac{k+\frac{1}{q-1}+1}{k+\frac{1}{q-1}}\right).
\end{align*}
Set $w=e^r$, then $\frac{dw}{dr}=w$. We have
\begin{align*}
\rho(q) &=\int_{0}^{\infty} \frac{(-we^{-w}+qwe^{-qw})^2}{w(e^{-w}-e^{-qw})}dw\\
    &= \int_{0}^{\infty} \frac{w e^{-w}(1-qe^{-(q-1)w})^2}{1-e^{-(q-1)w}}dw\\
    &= \int_{0}^{\infty} \frac{w e^{-w}q^2(1-e^{-(q-1)w}+\frac{1}q{-1})^2}{1-e^{-(q-1)w}}dw\\
    &= \int_{0}^{\infty} w e^{-w}q^2\left(1-e^{-(q-1)w}+2\left(\frac{1}{q}-1\right)
    + \frac{(\frac{1}{q}-1)^2}{1-e^{-(q-1)w}}\right)dw\\
    &= \int_{0}^{\infty} w e^{-w}\left(-q^2e^{-(q-1)w}+(-q^2+2q)
    +\frac{(q-1)^2}{1-e^{-(q-1)w}}\right)dw\\
    &= -q^2\int_{0}^{\infty} we^{-qw}dw
    +(-q^2+2q)\int_{0}^{\infty} we^{-w}dw
    +(q-1)^2\int_{0}^{\infty} \frac{we^{-w}}{1-e^{-(q-1)w}}dw.
\end{align*}
We calculate the three integrals separately. First we have
\begin{align*}
    \int_0^\infty we^{-qw}dw
    &=
    \left.e^{-qw}\left(-\frac{w}{q}-\frac{1}{q^2}\right)\right|_0^\infty = \frac{1}{q^2}.
\intertext{For the second we have}
    \int_0^\infty w e^{-w} dw &= e^{-w}(-w-1)\Big|_0^\infty=1.
\intertext{For the last, let $u=(q-1)w$. Then we have}
    \int_{0}^{\infty} \frac{we^{-w}}{1-e^{-(q-1)w}}dw
    &= \frac{1}{(q-1)^2}\int_{0}^{\infty} \frac{ue^{-\frac{u}{q-1}}}{1-e^{-u}}du,
\end{align*} 
where $\int_{0}^{\infty} \frac{ue^{-\frac{u}{q-1}}}{1-e^{-u}}du$ is just the integral representation of the Hurwitz zeta function $\zeta(2,\frac{1}{q-1})$. This can be written as a sum series as follows.
\begin{align*}
    \int_{0}^{\infty} \frac{ue^{-\frac{u}{q-1}}}{1-e^{-u}}du
    =
    \sum_{k=0}^\infty\frac{1}{(k+\frac{1}{q-1})^2}.
\end{align*}
Combining the three integrals, we conclude that
\[
    \rho(q) = -1-q^2+2q+\sum_{k=0}^\infty\frac{1}{(k+\frac{1}{q-1})^2} \;=\; \sum_{k=1}^\infty\frac{1}{(k+\frac{1}{q-1})^2}.
\]
\end{proof}

\subsection{Proof of Lemma~\ref{lem:HLL}}\label{sect:lemHLL}

\begin{proof}[Proof of Lemma~\ref{lem:HLL}]
For a fixed $\lambda$, by the definition of Shannon entropy, we have
\begin{align*}
    H(X_{\qLL,\lambda}) &= \sum_{k=-\infty}^\infty -(e^{-\frac{\lambda}{q^k}}-e^{-\frac{\lambda}{q^{k-1}}})\log_2(e^{-\frac{\lambda}{q^k}}-e^{-\frac{\lambda}{q^{k-1}}})\\
    &=\sum_{k=-\infty}^\infty -(e^{-\frac{q\lambda}{q^k}}-e^{-\frac{q\lambda}{q^{k-1}}})\log_2 (e^{-\frac{q\lambda}{q^k}}-e^{-\frac{q\lambda}{q^{k-1}}})
    \;=\; H(X_{\qLL,q\lambda}).
\end{align*}
Also, we have
\begin{align*}
    I_{\qLL}(\lambda) &=\sum_{k=-\infty}^\infty \frac{\left(-\frac{1}{q^k}e^{-\frac{\lambda}{q^k}}+\frac{1}{q^{k-1}}e^{-\frac{\lambda}{q^{k-1}}}\right)^2}{e^{-\frac{\lambda}{q^k}}-e^{-\frac{\lambda}{q^{k-1}}}}\\
    &= q^2\sum_{k=-\infty}^\infty \frac{\left(-\frac{1}{q^k}e^{-\frac{q\lambda}{q^k}}+\frac{1}{q^{k-1}}e^{-\frac{q\lambda}{q^{k-1}}}\right)^2}{e^{-\frac{q\lambda}{q^k}}-e^{-\frac{q\lambda}{q^{k-1}}}}
    \;=\; q^2\cdot I_{\qLL}(q\lambda).
\end{align*}
We conclude that $\qLL$ is weakly scale-invariant with base $q$.
We now turn to calculating $\mathcal{H}(\qLL)$ and $\mathcal{I}(\qLL)$.
By Definition~\ref{def:fish},
\begin{align*}
    \mathcal{H}(\qLL) &= \int_0^1 H(X_{\qLL},q^r)dr\\
    &= - \int_0^1 \sum_{k=-\infty}^\infty \left(e^{-q^{r-k}}-e^{-q^{r-k+1}}\right)\log_2 \left(e^{-q^{r-k}}-e^{-q^{r-k+1}}\right)dr\\
    &= - \int_{-\infty}^\infty \left(e^{-q^{r}}-e^{-q^{r+1}}\right)\log_2 \left(e^{-q^{r}}-e^{-q^{r+1}}\right)dr
    \;=\; \frac{\phi(q)}{\ln q}.
\intertext{Again, by Definition~\ref{def:fish},}
    \mathcal{I}(\qLL) &= \int_0^1 q^{2r}I_{\qLL}(q^r) dr  \\
    &= \int_0^1\sum_{k=-\infty}^\infty \frac{\left(-q^{r-k}e^{-q^{r-k}}+q^{r-k+1}e^{-q^{r-k+1}}\right)^2}{e^{-q^{r-k}}-e^{-q^{r-k+1}}}dr\\
    &= \int_{-\infty}^{\infty} \frac{\left(-q^{r}e^{-q^{r}}+q^{r+1}e^{-q^{r+1}}\right)^2}{e^{-q^{r}}-e^{-q^{r+1}}}dr
    \;=\; \frac{\rho(q)}{\ln q}.
\end{align*}
\end{proof}

\subsection{Proof of Lemma~\ref{lem:HLL_qhigh} }

\begin{proof}[Proof of Lemma~\ref{lem:HLL_qhigh} ]
Note that by Lemma \ref{lem:HLL}, 
\[
\fish(\qLL) = \frac{\mathcal{H}(\qLL)}{\mathcal{I}(\qLL)}=\frac{\phi(q)}{\rho(q)}.
\]
Then we have 
\begin{align*}
    \ln 2\cdot \phi'(q) &= \frac{1}{q^2}+\sum_{k=1}^\infty \frac{1}{k}\left(\frac{1}{k+\frac{1}{q-1}+1}-\frac{1}{k+\frac{1}{q-1}}\right)\frac{-1}{(q-1)^2} \\
    &= \frac{1}{q^2}+\frac{1}{(q-1)^2}\sum_{k=1}^\infty\frac{1}{k(k+\frac{1}{q-1})(k+\frac{q}{q-1})}\\
    &= \frac{1}{q(q-1)}\left(\frac{q-1}{q}+\sum_{k=1}^\infty\frac{1}{k(k+\frac{1}{q-1})(\frac{q-1}{q}k+1)}\right)\\
    &= \frac{1}{q(q-1)}\left(\sum_{k=1}^\infty \right( \frac{1}{k+\frac{1}{q-1}}-\frac{1}{k+\frac{1}{q-1}+1}\left)
    +\sum_{k=1}^\infty\frac{1}{k(k+\frac{1}{q-1})(\frac{q-1}{q}k+1)}\right)\\
    &= \frac{1}{q(q-1)}\left(\sum_{k=1}^\infty \frac{1}{(k+\frac{1}{q-1})(k+\frac{q}{q-1})}
    +\sum_{k=1}^\infty\frac{1}{k(k+\frac{1}{q-1})(\frac{q-1}{q}k+1)}\right)\\
    &= \frac{1}{q(q-1)}\left(\sum_{k=1}^\infty\frac{\frac{q-1}{q}k+1}{k(k+\frac{1}{q-1})(\frac{q-1}{q}k+1)}\right)\\
    &= \frac{1}{q(q-1)}\left(\sum_{k=1}^\infty\frac{1}{k(k+\frac{1}{q-1})}\right),
\intertext{and}
    \rho'(q) &=\frac{2}{(q-1)^2}\sum_{k=1}^\infty\frac{1}{(k+\frac{1}{q-1})^3}.
\intertext{Define $\alpha$ and $\beta$ as follows.}
    \alpha(q) &= \sum_{k=1}^\infty\frac{1}{k(k+\frac{1}{q-1})}\\
    \beta(q)  &= \sum_{k=1}^\infty\frac{1}{(k+\frac{1}{q-1})^3},
\end{align*}
We then have
\begin{align*}
   \ln2\cdot \frac{d}{dq}\frac{\mathcal{H}(\qLL)}{\mathcal{I}(\qLL)} 
    &= \ln2 \cdot\frac{d}{dq}\frac{\phi(q)}{\rho(q)}\\
    &= \ln2 \cdot\frac{\phi'(q)\rho(q)-\rho'(q)\phi(q)}{\rho(q)^2}\\
    &= \frac{\frac{q-1}{q}\alpha(q)\rho(q)-2\beta(q)\ln2 \cdot\phi(q)}{(q-1)^2\rho(q)^2}.
\end{align*}

We define $g(a,b)=\frac{b-1}{b}\alpha(b)\rho(b)-2\beta(a)\ln2\cdot\phi(a)$ and thus $\frac{d}{dq}\frac{\mathcal{H}(\qLL)}{\mathcal{I}(\qLL)}<0$ 
if and only if $g(q,q)<0$.

Note that we have $\phi'(q)>0$ and $\rho'(q)>0$ for all $q>1$ and thus both $\rho(q)$ and $\phi(q)$ are monotonically increasing for $q>1$. Note that $\frac{1}{q-1}$ is monotonically decreasing for $q>1$, thus both $\alpha(q)$ and $\beta(q)$ are also monotonically increasing.

Let $a<b$ where $a\in (1,\infty)$ and $b\in (1,\infty]$. Since $\alpha,\rho,\beta$ and $\phi$ are all monotonically increasing, if $g(a,b)<0$, then for any $q\in[a,b)$, 
\begin{align*}
    g(q,q) &=\frac{q-1}{q}\alpha(q)\rho(q)-2\beta(q)\ln2\cdot \phi(q)\\                  &<\frac{b-1}{b}\alpha(b)\rho(b)-2\beta(a)\ln2\cdot \phi(a)\\
            &=g(a,b)\\
            &<0.
\end{align*}

Thus, to prove that $\fish(\qLL)$ is strictly decreasing for $q\geq 1.4$, it is sufficient to find a sequence $1.4=q_0<q_1,\ldots<q_n=\infty$ such that for all $k\in[n]$,  $g(q_{k-1},q_k)<0$.
The following table shows the existence of such a sequence and thus completes the proof.
    \begin{center}
        \begin{tabular}{c|c|c}
        $k$  &  $q_k$ &  $g(q_{k-1},q_k)$ \\
        \hline
        0     & 1.4 & \\
        1     & 1.49 & -0.00228439\\
        2     & 1.62 & -0.00186328\\
        3     & 1.81 & -0.00522805\\
        4     & 2.12   & -0.00747658\\
        5     & 2.72   & -0.0038581\\
        6     & 4.25   & -0.00602114\\
        7     & 6  & -0.669626\\
        8     & $\infty$  & -0.216103\\
        \end{tabular}
    \end{center}
\end{proof}

\subsection{Proof of Lemma~\ref{lem:HLL_qlow}}

\begin{proof}[Proof of Lemma \ref{lem:HLL_qlow}]
We first calculate that $\ln2\cdot \fish(2\text{-}\LL)\approx 2.1097$. We then prove that for any $q\in(1,1.4]$, $\ln2\cdot \fish(\qLL)>2.11$.
We use the inequality $\ln x > 1-\frac{1}{x}$ for $x>0$, and find that 
\begin{align*}
    \ln\left(\frac{k+\frac{1}{q-1}+1}{k+\frac{1}{q-1}}\right) >1-\frac{k+\frac{1}{q-1}}{k+\frac{1}{q-1}+1}=\frac{1}{k+\frac{1}{q-1}+1}.
\end{align*}
Thus we have
\begin{align*}
    \ln 2\cdot \phi(q) &> 1-\frac{1}{q}+\sum_{k=1}^\infty \frac{1}{k(k+\frac{q}{q-1})}.
\end{align*}
\noindent We also have
\begin{align*}
    \rho(q) &<\sum_{k=1}^\infty\left(\frac{1}{\frac{1}{q-1}+k-1}-\frac{1}{\frac{1}{q-1}+k}\right)\\
    &= \frac{1}{\frac{1}{q-1}+1-1}\\
    &=q-1.
\end{align*}
\noindent Combining the two, we have
\begin{align*}
    \ln2\cdot \fish(q\text{-}\LL)=\frac{\ln2\cdot \phi(q)}{\rho(q)} &> \frac{1}{q}+ \sum_{k=1}^\infty\frac{1}{k((q-1)k+q)}\\
    &\geq  \frac{1}{1.4}+ \sum_{k=1}^\infty\frac{1}{k(0.4 k+1.4)}\\
    &\approx 2.11863>\ln2\cdot \fish(2\text{-}\LL).
\end{align*}
We conclude that for any $q\in(1,1.4]$, $\fish(q\text{-}\LL)
> 
\fish(2\text{-}\LL)$.
\end{proof}

\section{Lemmas and Proofs of Section \ref{sect:lowerbound}}\label{sect:proofs_lowerbound}

\subsection{Lemma \ref{lem:rational_inequality}}\label{sect:lem:rational_inequality}

\begin{lemma}\label{lem:rational_inequality}
Let $f,g$ and $h$ be real functions. If 
\begin{itemize}
    \item $f(t)\geq 0$, $g(t)>0$ and $h(t)\in [0,1]$ for all $ t\in \mathbb{R}$, 
    \item $f(t)/g(t)$ and $h(t)$  are (weakly) decreasing in $t$, and
    \item $\int_{-\infty}^\infty f(t)dt<\infty$ and $\int_{-\infty}^\infty g(t)dt<\infty$,
\end{itemize}
then for any $-\infty\leq a<b\leq \infty$ such that $\int_{a}^b g(t)h(t) dt>0$, we have
\begin{align}
    \frac{\int_{a}^b f(t)h(t) dt}{\int_{a}^b g(t)h(t) dt} \geq \frac{\int_{a}^b f(t) dt}{\int_a^b g(t) dt}. \label{eq:abfrac}
\end{align}
In particular, we have
\begin{align}
\frac{\int_{-\infty}^b f(t) dt}{\int_{-\infty}^b g(t) dt} \geq \frac{\int_{-\infty}^\infty f(t) dt}{\int_{-\infty}^\infty g(t) dt}\label{eq:ibfrac}
\end{align}
\end{lemma}

\begin{proof}
To show (\ref{eq:abfrac}), it is sufficient to prove the following difference is non-negative.
\begin{align*}
   &\quad 2 \int_{a}^b f(t)h(t) dt \int_a^b g(t) dt - 2\int_{a}^b g(t)h(t) dt\int_{a}^b f(t) dt\\
   &= \int_a^b\int_a^b f(x)h(x)g(y)+f(y)h(y)g(x)dx dy-\int_a^b\int_a^b g(x)h(x)f(y)+g(y)h(y)f(x)dx dy\\
   &= \int_a^b\int_a^b (f(x)g(y)-g(x)f(y))(h(x)-h(y))dx dy.
\end{align*}
Note that $(f(x)g(y)-g(x)f(y))(h(x)-h(y))=\frac{1}{g(x)g(y)}(f(x)/g(x)-f(y)/g(y))(h(x)-h(y))\geq 0$, since both $f(t)/g(t)$ and $h(t)$ are decreasing. Thus the difference is non-negative.

Inequality (\ref{eq:ibfrac}) follows from \ref{eq:abfrac} by setting $a=-\infty$, $b=\infty$ and $h(t)=\mathbbm{1}(t\leq b)$.
\end{proof}

\subsection{Proof of Lemma \ref{lem:monotonicity}}

\begin{proof}[Proof of Lemma \ref{lem:monotonicity}]
Note that
\begin{align*}
    \frac{\dot{H}(t)\ln 2}{\dot{I}(t)} = \frac{1-e^{-t}}{t} + \frac{(1-e^{-t})^2}{t^2}\cdot (-e^{t}\ln(1-e^{-t})).
\end{align*}
Since $-\ln(1-e^{-t})$ is decreasing on $(0,\infty)$, it suffices to prove that $\frac{1-e^{-t}}{t}$ and $-e^{t}\ln(1-e^{-t})$ are decreasing on $(0,\infty)$. Let $f(t)=\frac{1-e^{-t}}{t}$ and $g(t)=-e^{t}\ln(1-e^{-t})$. By taking the derivative of $f(t)$, we have
\begin{align*}
    f'(t)&=\frac{e^{-t}t-(1-e^{-t})}{t^2}=-\frac{e^{-t}}{t^2}(e^t-1-t)\leq 0
\end{align*}
which implies $f(t)$ is decreasing. 
By taking the derivative of $g(t)$, we have
\begin{align*}
    g'(t)=-e^t\left(\ln(1-e^{-t})+\frac{e^{-t}}{1-e^{-t}}\right),
\end{align*}
where we want to show that $g'(t)\leq 0$ for all $t>0$. Note that for any $x>0$, we have $x\geq \ln(1+x)$. Set $x=\frac{1}{e^t-1}$ and we have
\begin{align*}
    &\frac{1}{e^t-1} \geq \ln\left(1+\frac{1}{e^t-1}\right)
    \iff \frac{e^{-t}}{1-e^{-t}}-\ln\left(\frac{e^t}{e^t-1}\right)\geq 0
    \iff \ln\left(1-e^{-t}\right)+\frac{e^{-t}}{1-e^{-t}} \geq 0,
\end{align*}
which implies, indeed,  $g'(t)\leq 0$ for all $t>0$. This completes the proof.
\end{proof}

\subsection{Proof of Lemma \ref{lem:H_relax}}
Lemma \ref{lem:H_relax} is a property of the function $\dot{H}(\cdot)$. To prove that, we first need the following lemmas. Define $\dot{H}_e(\cdot)$ to be $\dot{H}(\cdot)\ln 2$, i.e.~the entropy measured in the natural base.

\begin{lemma}\label{lem:H_bound}
For all $t>0$,
\begin{align*}
    te^{-t} \leq \dot{H}_e(t) \leq 2\sqrt{t}.
\end{align*}
\end{lemma}
\begin{proof}
The lower bound follows directly from the definition:
\begin{align*}
    \dot{H}_e(t)=te^{-t}-(1-e^{-t})\ln (1-e^{-t})\geq te^{-t}.
\end{align*}

For  the upper bound, first note that since $\dot{H}_e(t)$ is the entropy (measured in ``\emph{nats}'') of a Bernoulli random variable, we have $\dot{H}_e(t)\leq \ln(2)$. Thus we only need to prove $\dot{H}_e(t)\leq 2\sqrt{t}$ for $t\in (0,\ln (2)^2/4]$. Note that $te^{-t}\leq t \leq \sqrt{t}$ for $t\in (0,1]$. It then suffices to show that $-(1-e^{-t})\ln(1-e^{-t})\leq \sqrt{t}$ for $t\in (0,\ln^2 (2)/4]$. Observe the following.
\begin{itemize}
    \item $-x\ln x$ is increasing in $x\in(0,1/e]$, since $(-x\ln x)'=-\ln x+1$. Note that $1/e>0.36>0.12> 1-e^{-\ln^2 (2)/4}$.
    \item $1-e^{-t}\leq t$ for all $t\in \mathbb{R}$.
    \item $-t\ln t< \sqrt{t}$ for $t\in (0,1]$. Let $f(t)=t^{-1/2}+\ln t$. Then we have $f'(t)=-t^{-3/2}/2+1/t=t^{-3/2}(-1/2+\sqrt{t})$. 
    Then only zero of $f'(t)$ is at $t=1/4$, which is the minimum point. Note that $f(1/4)=2-\ln 4>0$. Therefore $t^{-1/2}+\ln t>0$ for $t\in(0,1]$, which implies $-t\ln t<\sqrt{t}$.
\end{itemize}
Then we have, for $t\in (0,\ln^2 (2)/4]$,
\begin{align*}
    -(1-e^{-t})\ln(1-e^{-t}) \leq -t\ln t \leq \sqrt{t},
\end{align*}
where the first inequality results from the first two observations and the second inequality follows from the last observation.
\end{proof}

\begin{lemma}\label{lem:HtpHt}
For any $p>1$, $\dot{H}_e(t/p)/\dot{H}_e(t)$ is increasing in $t\in (0,\ln(2))$.
\end{lemma}
\begin{proof}
Fix $p>1$, let $f(t)=\dot{H}_e(t/p)/\dot{H}_e(t)$. First note that
\begin{align*}
\dot{H}_e(t)&=te^{-t}-(1-e^{-t})\ln(1-e^{-t})=t-(1-e^{-t})\ln(e^t-1),\\
    \dot{H}_e'(t)&=e^{-t}-te^{-t}-e^{-t}\ln(1-e^{-t})-e^{-t}=-e^{-t}\ln(e^{t}-1).
\end{align*}
We have
\begin{align*}
    f'(t)=\frac{\dot{H}_e'(t/p)\dot{H}_e(t)/p-\dot{H}_e(t/p)\dot{H}_e'(t)}{\dot{H}_e(t)^2}.
\end{align*}
Define 
\begin{align*}
    g(t)&=\dot{H}_e'(t/p)\dot{H}_e(t)/p-\dot{H}_e(t/p)\dot{H}_e'(t)\\
    &=-e^{-t/p}\ln\left(e^{t/p}-1\right)\left(t-(1-e^{-t})\ln\left(e^t-1\right)\right)/p\\
    &\quad\quad\quad +e^{-t}\ln\left(e^{t}-1\right)\left(t/p-\left(1-e^{-t/p}\right)\ln\left(e^{t/p}-1\right)\right)\\
    &=\frac{1}{pe^{t/p+t}}\bigg[-\ln\left(e^{t/p}-1\right)\left(te^t-\left(e^t-1\right)\ln\left(e^t-1\right)\right)\\
    &\quad\quad\quad\quad\quad
    +p\ln\left(e^{t}-1\right)\left(\frac{t}{p}e^{t/p}-\left(e^{t/p}-1\right)\ln\left(e^{t/p}-1\right)\right)\bigg].
\end{align*}
We want to show that $g(t)\geq 0$ for $t\in (0,\ln(2))$.
Define
\begin{align*}
    h(t)= \frac{e^t}{\ln(e^t-1)}-\frac{e^t-1}{t},
\end{align*}
and then we can write
\begin{align*}
    g(t)=\frac{t\ln(e^{t/p}-1)\ln(e^t-1)}{pe^{t/p+t}}(-h(t)+h(t/p)).
\end{align*}
Since $t\in(0,\ln(2))$ and $p>1$, we know $\frac{t\ln(e^{t/p}-1)\ln(e^t-1)}{pe^{t/p+t}}>0$. Therefore, it suffices to show $h(t)$ is decreasing in $(0,\ln(2))$. Note that for $t\in(0,\ln(2))$, we have $\ln(e^t-1)<0$ and $|\ln(e^t-1)|$ is decreasing while $e^t$ is increasing. Thus $\frac{e^t}{\ln(e^t-1)}$ is decreasing on $(0,\ln(2))$.  On the other hand,
\begin{align*}
    \frac{d}{dt}\frac{e^t-1}{t} = \frac{e^t t-e^t+1}{t^2}.
\end{align*}
Let $w(t)=e^t t-e^t+1$. We have $w'(t)=e^t+e^tt-e^t>0$. Thus we have $w(t)\geq w(0)=0$. Therefore, $\frac{e^t-1}{t}$ is increasing in $t$. Thus, indeed, $h(t)$ is decreasing. 

We conclude that $\dot{H}_e(t/p)/\dot{H}_e(t)$ is increasing in $t\in (0,\ln(2))$.
\end{proof}

Now we can prove Lemma \ref{lem:H_relax}.

\begin{proof}[Proof of Lemma \ref{lem:H_relax}]
Note that by the upper bound in Lemma \ref{lem:H_bound}, we have
\begin{align*}
   \int_{-\infty}^{-t} \dot{H}_e(e^{x}) dx \leq \int_{-\infty}^{-t} 2e^{x/2} dx=4 e^{-t/2}\leq 4d^{-1/4},
\end{align*}
where the last inequality follows from the assumption $t\geq \frac{1}{2}\ln d$.
If $-t+d> \ln\ln(2)$, then, by the lower bound in Lemma \ref{lem:H_bound},
\begin{align*}
    \int_{-\infty}^{-t+d} \dot{H}_e(e^x) dx\geq \int_{-\infty}^{\ln\ln(2)} \dot{H}_e(e^x) dx\geq\int_{-\infty}^{\ln\ln(2)} e^{x}e^{-e^x} dx= \frac{1}{2},
\end{align*}
where we use the fact that $\int e^xe^{-e^x}dx=-e^{-e^x}$.
 In this case we have 
\begin{align*}
     \frac{\int_{-\infty}^{-t} \dot{H}_e(e^x) dx}{\int_{-\infty}^{-t+d} \dot{H}_e(e^x) dx} \leq 8d^{-1/4}.
\end{align*}
If $-t+d\leq \ln\ln(2)$, then for any $x\leq -t+d$, $e^x\leq \ln(2)$. Thus by Lemma \ref{lem:HtpHt}, for any $x\leq -t+d$,  $\dot{H}_e(e^{x-d})/\dot{H}_e(e^{x})\leq \dot{H}_e(\ln(2)/e^d)/\dot{H}_e(\ln(2))$. Then we have
\begin{align*}
    \int_{-\infty}^{-t+d} \dot{H}_e(e^x) dx = \int_{-\infty}^{-t+d} \dot{H}_e(e^{x-d})\frac{\dot{H}_e(e^{x})}{\dot{H}_e(e^{x}/e^d)} dx\geq \frac{\dot{H}_e(\ln(2))}{\dot{H}_e(\ln2/ e^{d})}\int_{-\infty}^{-t} \dot{H}_e(e^{x}) dx,
\end{align*}
which implies, using the bounds of Lemma~\ref{lem:H_bound}, that
\begin{align*}
     \frac{\int_{-\infty}^{-t} \dot{H}_e(e^x) dx}{\int_{-\infty}^{-t+d} \dot{H}_e(e^x) dx} \leq \frac{\dot{H}_e(\ln(2)/e^{d})}{\dot{H}_e(\ln(2))}\leq \frac{2\sqrt{\ln(2) /e^{d}}}{\ln(2) e^{-\ln(2)}}<5e^{-d/2}.
\end{align*}
We conclude that, in both cases,
\begin{align*}
   \frac{\int_{-\infty}^{-t} \dot{H}(e^x) dx}{\int_{-\infty}^{-t+d} \dot{H}(e^x) dx}=\frac{\int_{-\infty}^{-t} \dot{H}_e(e^x) dx}{\int_{-\infty}^{-t+d} \dot{H}_e(e^x) dx} \leq \max\{8d^{-1/4},5e^{-d/2}\}.
\end{align*}
\end{proof}

\subsection{Proof of Lemma \ref{lem:I_relax}}
Lemma \ref{lem:I_relax} is a property of the function $\dot{I}(\cdot)$. To prove that, we first need the following lemmas.
 \begin{lemma}\label{lem:I_bound}
    For all $t> 0$,
    \begin{align*}
        \dot{I}(t)\leq 4e^{-t/2}.
    \end{align*}
    \end{lemma}
    \begin{proof}
    Recall that $\dot{I}(t)=\frac{t^2}{e^t-1}$. Therefore we have, for $t>0$,
    \begin{align*}
        \dot{I}(t)\leq 4e^{-t/2}
        \iff  4e^t -t^2 e^{t/2}- 4 \geq 0.
    \end{align*}
    Let $f(t)=4e^t -t^2 e^{t/2}- 4$. Then we have
    \begin{align*}
        f'(t)&=4e^t-e^{t/2}(t^2/2+2t)=e^{t/2}(4e^{t/2}-2t-t^2/2)\\&\geq e^{t/2}(4+2t+t^2/2-2t-t^2/2)\geq 0,
    \end{align*}
    since $e^x \geq 1+x+x^2/2$ for $x\geq 0$, which implies $e^{t/2}\geq 1+t/2+t^2/8$. Thus $f(t)$ is increasing. In addition, $f(0)=0$. Thus we know $f(t)\geq 0$ for all $t\geq 0$ and therefore $\dot{I}(t)\leq 4e^{-t/2}$ holds.
    \end{proof}
\begin{lemma}\label{lem:I_sum}
\begin{align*}
    \int_0^\infty \sum_{k\geq 0} \sup\{\dot{I}(e^{x+w})\mid w\in [k,k+1)\}dx \leq 119.
\end{align*}
\end{lemma}
\begin{proof}
Using the bound derived in Lemma \ref{lem:I_bound}, we have
\begin{align*}
   &\int_0^\infty \sum_{k\geq 0} \sup\{\dot{I}(e^{x+w})\mid w\in [k,k+1)\}dx\\
   &\leq \int_0^\infty \sum_{k\geq 0} \sup\{4\exp(-e^{x+w}/2)\mid w\in [k,k+1)\}dx\\
   &=  4\int_0^\infty \sum_{k\geq 0} \exp(-e^{x+k}/2)dx\\
   &\leq  4\int_0^\infty \int_0^\infty \exp(-e^{x+y-1}/2) dydx\\
   &\leq  4\int_0^\infty \int_0^\infty \exp(-(x+y)/(2e)) dydx\\
   &=  4 \left( \int_0^\infty e^{-x/(2e)} dx\right)^2\\
   &= 4\cdot (2e)^2<119.
\end{align*}
\end{proof}

\begin{lemma}\label{lem:h_underline}
For any interval $A$, define $\underline{h}_e(A)\bydef  \inf\{\dot{H}_e(e^x)\mid x\in A\}$. 
Then for any $a<b$, $\underline{h}_e([a,b])=\min\{\dot{H}_e(e^a),\dot{H}_e(e^b)\}$.
\end{lemma}
\begin{proof}
Note that, from the proof of Lemma \ref{lem:HtpHt}, we know 
    $\dot{H}_e'(t)=-e^{-t}\log(e^{t}-1)$, whose only zero is $\log 2$. From this we know when $t<\log 2$, $\dot{H}_e(t)$ increases and when  $t>\log 2$, it decreases. Thus given an interval $[e^{a},e^b]$, the minimum is always obtained at one of the end points.
\end{proof}

\begin{proof}[Proof of Lemma \ref{lem:I_relax}]
For any $k\geq 0$, define 
\begin{align*}
B_k &\bydef \{c_i\in \mathcal{C}\mid p_i\in [e^{-a+k},e^{-a+k+1})\}\\  B^* &\bydef \{c_i\in \mathcal{C}\mid p_i\in (e^{-a-\Delta},e^{-a})\}. 
\intertext{We see $\{B_k\}_{k\geq 0}$ together with $B^*$ form a partition of $\mathcal{C}\setminus\mathcal{C}^*$. Define $w(B_k)$ as}
w(B_k) &\bydef \sum_{c_i\in B_k}\Pr(\phi(\mathbf{Y}_{i-1,e^{a-k}})=0).
\end{align*}
Fix some $k\geq 0$. By the entropy assumption, when $\lambda = e^{a-k}$, we have
\begin{align}
\tilde{H}\geq H(\mathbf{Y}_{[\lambda]}) \geq  
\sum_{c_i\in B_k}\dot{H}_e(p_i e^{a-k})\Pr(\phi(\mathbf{Y}_{i-1,e^{a-k}})=0)/\ln(2)\geq \underline{h}_e([0,1])w(B_k)/\ln(2),\label{eq:hwbk}
\end{align}
since for all cells $c_i\in B_k$, $p_ie^{a-k}\in[e^0,e^1]$ by definition. This implies $w(B_k)\leq \frac{\tilde{H}\ln(2)}{\underline{h}_e([0,1])}$ for any $k\geq 0$.
By Proposition \ref{prop:phi_monotone}, we have, for any $k\geq 0$,
\begin{align*}
    \mathbf{I}(B_k \to [a,b])&=\int_a^b \sum_{c_j\in B_k} \dot{I}(p_je^x)\Pr(\phi(\mathbf{Y}_{j-1,e^x})=0)dx\\
    &\leq \int_a^b \sum_{c_j\in B_k} \dot{I}(p_je^x)\Pr(\phi(\mathbf{Y}_{j-1,e^{a-k}})=0)dx\\
    &\leq \int_a^b \sup\left\{\dot{I}(p_je^x) \;\middle|\; c_j\in B_k\right\} w(B_k)dx\\
    &\leq \int_a^b \sup\left\{\dot{I}(e^{-a+x+w}) \;\middle|\; w\in[k,k+1)\right\} w(B_k)dx,
\end{align*}
since for any cell $c_j\in B_k$, $p_j\in [e^{-a+k},e^{-a+k+1})$ by definition.
Then we have
\begin{align*}
    \mathbf{I}\left(\bigcup_{k\geq 0}B_k \to [a,b]\right) &=\sum_{k=0}^\infty \mathbf{I}(B_k \to [a,b])\\
    &\leq \sum_{k=0}^\infty \int_a^b \sup\left\{\dot{I}(e^{-a+x+w}) \;\middle|\; w\in[k,k+1)\right\} w(B_k)dx\\
    &= \int_0^\infty \sum_{k=0}^\infty  \sup\left\{\dot{I}(e^{-x+w}) \;\middle|\; w\in[k,k+1)\right\} w(B_k)dx.
\end{align*}
By Lemma \ref{lem:I_sum} and Inequality \ref{eq:hwbk}, we have
\begin{align*}
    \mathbf{I}(\cup_{k\geq 0}B_k \to [a,b]) & \leq \frac{\tilde{H}\ln(2)}{\underline{h}_e([0,1])} 119.
\end{align*}
By Lemma \ref{lem:h_underline}, we know $\underline{h}_e([0,1])>\min\{0.65,0.24\}=0.24$. We then have
\begin{align*}
\mathbf{I}(\cup_{k\geq 0}B_k \to [a,b]) &\leq 344\tilde{H}.
\end{align*}
On the other hand, when $\lambda =e^a$, we have,
by the assumption that the entropy never exceeds $\tilde{H}$, 
\begin{align*}
    \tilde{H} \geq H(\mathbf{Y}_{[\lambda]}) &\geq  \sum_{c_i\in B^*} \dot{H}_e(p_ie^a)\Pr(\phi(\mathbf{Y}_{i-1,e^a})=0)/\ln(2) \\
    &\geq \underline{h}_e([-\Delta,0]) \sum_{c_i\in B^*} \Pr(\phi(\mathbf{Y}_{i-1,e^a})=0)/\ln(2),
\end{align*}
since for any cell $c_i\in B^*$, $p_ie^a \in [-\Delta,0]$ by definition. Then we know that \begin{align*}
\Pr(\phi(\mathbf{Y}_{i-1,e^a})=0) &\leq \frac{\tilde{H}\ln(2)}{\underline{h}_e([-\Delta,0]) }. 
\end{align*}
We can calculate that, again, by Proposition \ref{prop:phi_monotone},
\begin{align*}
   \mathbf{I}(B^*\to [a,b]) &= \int_a^b \sum_{c_i\in B^*} \dot{I}(p_ie^x)\Pr(\phi(\mathbf{Y}_{i-1,e^x})=0)dx\\
    &\leq \int_a^b \sum_{c_i\in B^*} \dot{I}(p_ie^x)\Pr(\phi(\mathbf{Y}_{i-1,e^a})=0)dx\\
    &\leq \int_{-\infty}^\infty \sum_{c_i\in B^*} \dot{I}(p_ie^x)\Pr(\phi(\mathbf{Y}_{i-1,e^a})=0)dx.\\
    \intertext{Note that by Lemma \ref{lem:integration},
    for any $p_i>0$, 
    $\int_{-\infty}^\infty \dot{I}(p_ie^x)dx=\int_{-\infty}^\infty \dot{I}(e^x)dx=I_0$.  Continuing,}
    &= I_0\cdot \sum_{c_i\in B^*} \Pr(\phi(\mathbf{Y}_{i-1,e^a})=0)dx\\
    &\leq I_0 \frac{\tilde{H}\ln(2)}{\underline{h}_e([-\Delta,0])}.
\end{align*}
By Lemmas~\ref{lem:H_bound} and \ref{lem:h_underline}, $\underline{h}_e([-\Delta,0])=\min\{\dot{H}_e(e^{-\Delta}),\dot{H}_e(e^0)\}\geq \min\{e^{-\Delta}e^{-e^{-\Delta}}, 0.65\}\geq e^{-\Delta-1}$  since $\Delta >0$. Recall that by Lemma \ref{lem:H0I0}, $I_0=\pi^2/6$. 
Thus, we have $\mathbf{I}(B^*\to [a,b])\leq \tilde{H} 4 e^{\Delta}$.
Finally, we conclude that
\begin{align*}
    \mathbf{I}(\mathcal{C}\setminus\mathcal{C}^*\to [a,b])&=  \mathbf{I}(\cup_{k\geq 0}B_k \to [a,b])+\mathbf{I}(B^*\to [a,b])\leq (344+4e^\Delta)\tilde{H}.
\end{align*}
\end{proof}

\end{document}